\documentclass[11pt]{article}
\usepackage[utf8]{inputenc}
\usepackage[margin=1in]{geometry}
\usepackage{amsmath, amsthm, amssymb}
\usepackage{mathtools}
\mathtoolsset{showonlyrefs}
\usepackage{dsfont}
\usepackage[usenames,dvipsnames]{xcolor}
\usepackage{thm-restate}
\definecolor{Gred}{RGB}{219, 50, 54}
\definecolor{ToCgreen}{RGB}{0, 128, 0}

\usepackage{hyperref}
\hypersetup{
      colorlinks=true,
  citecolor=ToCgreen,
  linkcolor=Sepia,
  filecolor=Gred,
  urlcolor=Gred
  }

\title{Tight Bounds for Quantum State Certification with Incoherent Measurements}
\author{
    Sitan Chen\thanks{Email: \texttt{sitanc@berkeley.edu}. This work was supported in part by NSF Award 2103300.} \\
    UC Berkeley
        \and 
    Brice Huang\thanks{Email: \texttt{bmhuang@mit.edu}. Supported by an NSF graduate research fellowship, a Siebel scholarship, NSF awards DMS-2022448 and CCF-1940205, and NSF TRIPODS award 1740751.} \\
    MIT
        \and
    Jerry Li\thanks{Email: \texttt{jerrl@microsoft.com}} \\
    Microsoft Research
        \and
    Allen Liu\thanks{Email: \texttt{cliu568@mit.edu}. This work was supported in part by an NSF Graduate Research Fellowship and a Fannie and John Hertz Foundation Fellowship} \\
    MIT
}

\newtheorem{theorem}{Theorem}[section]

\newtheorem{lemma}[theorem]{Lemma}
\newtheorem{proposition}[theorem]{Proposition}
\newtheorem{claim}[theorem]{Claim}
\newtheorem{definition}[theorem]{Definition}
\newtheorem{fact}[theorem]{Fact}
\newtheorem{remark}[theorem]{Remark}
\newtheorem{corollary}[theorem]{Corollary}

\newcommand{\Gin}{\mathrm{Gin}}
\newcommand{\GOE}{\mathrm{GOE}}
\newcommand{\sGOE}{\GOE^*}
\newcommand{\Tr}{\mathrm{Tr}}
\newcommand{\R}{\mathbb{R}}

\newcommand{\diag}{\mathrm{diag}}
\newcommand{\norm}[1]{\lt\|#1\rt\|}
\newcommand{\op}{\mathrm{op}}

\newcommand{\cF}{\mathcal{F}}
\newcommand{\cN}{\mathcal{N}}

\newcommand{\de}{\mathrm{d}}
\newcommand{\eps}{\varepsilon}

\newcommand{\Mat}{\mathrm{Mat}}
\newcommand{\PMat}{\mathrm{PMat}}

\newcommand{\fr}{\frac}
\newcommand{\lt}{\left}
\newcommand{\rt}{\right}

\newcommand{\la}{\langle}
\newcommand{\ra}{\rangle}

\newcommand{\oG}{{\overline G}}

\newcommand{\oL}{{\overline L}}
\newcommand{\oM}{{\overline M}}

\newcommand{\eff}{\mathsf{eff}}

\newcommand{\bR}{\mathbb{R}}
\newcommand{\bN}{\mathbb{N}}

\newcommand{\polylog}{\mathrm{polylog}}
\newcommand{\TV}{d_{\mathrm{TV}}}

\DeclareMathOperator*{\bE}{\mathbb{E}}

\DeclarePairedDelimiter{\brk}{[}{]}
\DeclarePairedDelimiter{\brc}{\{}{\}}

\makeatletter
\def\Pr{\@ifnextchar[{\@witha}{\@withouta}}
\def\@witha[#1]{\mathop{{}\operator@font Pr}_{#1}\brk}
\def\@withouta{\mathop{{}\operator@font Pr}\brk}
\makeatother

\makeatletter
\def\E{\@ifnextchar[{\@withb}{\@withoutb}}
\def\@withb[#1]{\mathop{{}\mathbb{E}}_{#1}\brk}
\def\@withoutb{\mathop{{}\mathbb{E}}\brk}
\makeatother

\newcommand{\bone}{\mathds{1}\brk}

\DeclarePairedDelimiter{\iprod}{\langle}{\rangle}

\newcommand{\Id}{\mathds{1}}
\newcommand{\calA}{\mathcal{A}}
\newcommand{\calB}{\mathcal{B}}
\newcommand{\calD}{\mathcal{D}}
\newcommand{\calN}{\mathcal{N}}
\newcommand{\calT}{\mathcal{T}}
\newcommand{\wt}{\widetilde}
\newcommand{\calJ}{\mathcal{J}}

\newcommand{\Sfew}{S_{\mathsf{few}}}
\newcommand{\Slight}{S_{\mathsf{light}}}
\newcommand{\Smany}{S_{\mathsf{many}}}
\newcommand{\Stail}{S_{\mathsf{tail}}}
\newcommand{\mm}{{\mathsf{mm}}}
\newcommand{\wh}{\widehat}
\renewcommand{\S}{\mathbb{S}}
\newcommand{\C}{\mathbb{C}}

\newcommand{\bx}{{\boldsymbol x}}
\newcommand{\by}{{\boldsymbol y}}
\newcommand{\bz}{{\boldsymbol z}}
\newcommand{\bw}{{\boldsymbol w}}

\usepackage{enumitem}

\makeatletter
\renewcommand{\paragraph}{%
  \@startsection{paragraph}{4}%
  {\z@}{1.25ex \@plus 1ex \@minus .2ex}{-1em}%
  {\normalfont\normalsize\bfseries}%
}
\makeatother

\usepackage{accents}

\DeclareMathAccent{\wtilde}{\mathord}{largesymbols}{"65}
\newcommand{\ggsim}{\ggg}
\newcommand{\llsim}{\lll}

\begin{document}

\pagestyle{empty}
{
  \renewcommand{\thispagestyle}[1]{}
  \maketitle
    \begin{abstract}
    We consider the problem of quantum state certification, where we are given the description of a mixed state $\sigma \in \C^{d \times d}$, $n$ copies of a mixed state $\rho \in \C^{d \times d}$, and $\eps > 0$, and we are asked to determine whether $\rho = \sigma$ or whether $\| \rho - \sigma \|_1 > \eps$.
    When $\sigma$ is the maximally mixed state $\fr{1}{d} I_d$, this is known as mixedness testing.
    We focus on algorithms which use incoherent measurements, i.e. which only measure one copy of $\rho$ at a time. Unlike those that use entangled, multi-copy measurements, these can be implemented without persistent quantum memory and thus represent a large class of protocols that can be run on current or near-term devices.
    
    For mixedness testing, there is a folklore algorithm which uses incoherent measurements and only needs $O(d^{3/2} / \eps^2)$ copies.
    The algorithm is non-adaptive, that is, its measurements are fixed ahead of time, and is known to be optimal for non-adaptive algorithms.
    However, when the algorithm can make arbitrary incoherent measurements, the best known lower bound is only $\Omega (d^{4/3} / \eps^2)$~\cite{bubeck2020entanglement}, and it has been an outstanding open problem to close this polynomial gap. In this work:
    \begin{itemize}[leftmargin=*,itemsep=0pt]
        \item We settle the copy complexity of mixedness testing with incoherent measurements and show that $\Omega (d^{3/2} / \eps^2)$ copies are necessary. This fully resolves open questions of~\cite{wright2016learn} and~\cite{bubeck2020entanglement}.
        \item We show the instance-optimal bounds for state certification to general $\sigma$ first derived in~\cite{chen2021toward} for non-adaptive measurements also hold for arbitrary incoherent measurements.
    \end{itemize}
    Qualitatively, our results say that adaptivity does not help at all for these problems. Our results are based on new techniques that allow us to reduce the problem to understanding the concentration of certain matrix martingales, which we believe may be of independent interest.
    
    \end{abstract}
}

\clearpage
\pagestyle{plain}
\pagenumbering{arabic}

\section{Introduction}

Quantum mixedness testing, and more generally quantum state certification, are two of the most basic and fundamental tasks in quantum property testing.
In quantum state certification, the learner is given $n$ copies of a mixed state $\rho \in \C^{d \times d}$, and an explicit description of a mixed state $\sigma \in \C^{d \times d}$, and the objective is to distinguish with probability at least $0.99$ between the case where $\rho = \sigma$ or if it is $\eps$-far from $\sigma$ in trace distance.\footnote{Note that by standard bootstrapping arguments the choice of constant here is arbitrary, and can be any constant larger than $1/2$. This only changes the sample complexity by constant factors.}
Mixedness testing is the special case of state certification where $\sigma = \tfrac{1}{d} I_d$, i.e., when the target state is the maximally mixed state.

Mixedness testing and state certification are the natural quantum analogues of uniformity testing and identity testing, respectively, two of the most well-studied problems in distribution testing.
From a more practical point of view, state certification is also a key subroutine which allows experimentalists to verify the outcomes of their quantum experiments.
For instance, if an algorithmist wishes to check that a quantum algorithm with quantum output is correctly outputting the right state, then this is exactly the problem of state certification.

Despite the fundamental nature of the problems, it was not until relatively recently that the copy complexity of state certification and mixedness testing were first understood.
The seminal paper of~\cite{o2015quantum} first demonstrated that $n = \Theta (d / \eps^2)$ copies were necessary and sufficient to solve mixedness testing.
Follow-up work of~\cite{buadescu2019quantum} later demonstrated that $n = O (d / \eps^2)$ is also sufficient for the more general problem of state certification. Combined with the lower bound for mixedness testing, this resolved the copy complexity of state certification, in the worst case over $\sigma$.

However, a major downside of the estimators which achieve these copy complexities is that they require heavily entangled measurements over the joint state $\rho^{\otimes n}$.
This poses a number of challenges to porting these algorithms into practical settings.
First, the descriptions of the measurements are quite large (as the overall joint state is of size $d^n \times d^n$), and cannot be implemented on current (or near-term) quantum devices.
Second, the measurements require that all $n$ copies of $\rho$ are simultaneously present.
In many realistic settings, such as streaming settings where one copy of $\rho$ is given to the algorithm at a time, this would require that the quantum device be able to store all of these copies in persistent quantum memory.
Such a task is also out of reach for current or near-term quantum devices, in essentially any non-trivial regime of the parameters, especially when one considers that $d$ is exponential in the number of qubits in the system!

An appealing class of algorithms which avoids both these issues, and which can be implemented on real world noisy intermediate-scale quantum (NISQ) devices, are algorithms which only rely on \emph{incoherent (a.k.a. unentangled) measurements}.
In contrast to general protocols which perform arbitrary measurements on the joint state over all $n$ copies, these algorithms only apply measurements to one copy of $\rho$ at a time, although these measurements can possibly be adaptively chosen based on the (classical) outcomes of the previous measurements.
Consequently, these measurements are performed on much smaller states, and moreover, can be performed without any quantum memory.

For these reasons, there has been a considerable amount of attention in recent years devoted to understanding the statistical power of algorithms that only use incoherent measurements, which was also posed as an open problem in Wright's thesis~\cite{wright2016learn}.
A recent work of~\cite{bubeck2020entanglement} demonstrated that if the measurements are additionally chosen non-adaptively, then $n = \Theta (d^{3/2} / \eps^2)$ copies are necessary and sufficient to solve mixedness testing.
They also demonstrated that \emph{any} algorithm using incoherent measurements---even those chosen adaptively---must use at least $n = \Omega (d^{4/3} / \eps^2)$ copies.
In other words, there is a polynomial separation between the power of algorithms with and without quantum memory for this problem.
Still, this left a gap between the best known upper and lower bounds for mixedness testing with incoherent measurements.
This begs the question:
\begin{center}
    {\it Can we fully characterize the copy complexity of mixedness testing with incoherent measurements?}
\end{center}
Closing this gap was posed as an open question in the work of~\cite{bubeck2020entanglement}.

Underlying this question is another, more qualitative one, regarding the power of adaptivity.
Indeed, a recurring theme in a number of different quantum learning settings is that while proving tight lower bounds against adaptive algorithms is quite challenging, the state-of-the-art algorithms almost always tend to be the ``obvious'' non-adaptive strategies.
A very interesting meta-question is understanding for which natural quantum learning problems (if any) adaptivity helps at all for algorithms that use incoherent measurements.

Our first main contribution is to fully resolve this question for mixedness testing: we prove that adaptivity does not improve the sample complexity at all, except possibly up to constant factors.
% \bh{I think this is now Corollary~\ref{cor:mixedness-testing}}
\begin{theorem}[Informal, see Theorem~\ref{cor:mixedness-testing}]
\label{thm:informal-mixedness-testing}
    The copy complexity of mixedness testing using incoherent measurements is $n = \Theta (d^{3/2} / \eps^2)$. 
\end{theorem}
\noindent By completely pinning down the copy complexity of mixedness testing with incoherent measurements, this answers open questions of~\cite{wright2016learn} and~\cite{bubeck2020entanglement}.
Qualitatively, our theorem states that adaptivity does not help the copy complexity of this problem whatsoever.

\paragraph{Instance-optimal lower bounds for state certification.}
We next turn to state certification.
Because mixedness testing is a special case of state certification, Theorem~\ref{thm:informal-mixedness-testing} immediately implies that $n = \Omega (d^{3/2} / \eps^2)$ copies are necessary for state certification, in the worst case over all choices of the reference state $\sigma$.
This, coupled with a matching upper bound from~\cite[Lemma 6.2]{chen2021toward}, resolves the copy complexity of state certification with incoherent measurements for worst-case $\sigma$.

However, it should be clear that this bound is not the correct bound for all possible $\sigma$.
For instance, when $\sigma$ is pure, it is not hard to see that $\Theta (1 / \eps^2)$ copies are sufficient and necessary.
This raises the natural question: what is the copy complexity of state certification with incoherent measurements, as a function of the reference state $\sigma$?
This is the quantum analogue of the (classical) distribution testing problem of obtaining instance optimal bounds for identity testing against a known distribution over $d$ elements~\cite{acharya2011competitive,acharya2012competitive,valiant2017automatic,diakonikolas2016new,blais2019distribution,jiao2018minimax}.
In the classical version of the problem, there is a known distribution $p$ over $\{1, \ldots, d\}$, and we are given samples from a distribution $q$.
We are asked to distinguish between the case wher $p = q$, and the case when $\| p - q \|_1 > \eps$.
A landmark result of~\cite{valiant2017automatic} states that the sample complexity of this question is (essentially) characterized by the $\ell_{2/3}$-quasinorm of $p$.

In this work, we ask whether or not a similar characterization can be obtained for the quantum version of the question.
Prior work of~\cite{chen2021toward} demonstrated such a characterization, but under the caveat that the measurements are chosen non-adaptively.
At a high level, they showed that the copy complexity of the problem is governed by the \emph{fidelity} between $\sigma$ and the maximally mixed state.
More precisely, they showed that if $\overline{\sigma}$ and $\underline{\sigma}$ are states given by zeroing out eigenvalues of $\sigma$ that have total mass at most $\Theta (\eps^2)$ and $\Theta (\eps)$ respectively and normalizing, then the copy complexity with non-adaptive measurements, denoted $n$, satisfies
\begin{equation}
    \label{eq:instance-optimal-bound}
    \widetilde{\Omega} \biggl( \frac{d \cdot \underline{d}_{\eff}^{1/2}}{\eps^2} \cdot F( \underline{\sigma}, \frac{1}{d} I_d ) \biggr) \leq n \leq \widetilde{O} \biggl( \frac{d \cdot \overline{d}_{\eff}^{1/2}}{\eps^2} \cdot F( \overline{\sigma}, \frac{1}{d} I_d) \biggr) \; ,
\end{equation}
where $\underline{d}_\eff$ (resp. $\overline{d}_\eff$) is the ``effective dimension'' of the problem, namely, the rank of $\underline{\sigma}$ (resp. $\overline{\sigma}$).
In the same work, they also gave lower bounds for arbitrary (possibly adaptive) incoherent measurements, but, like with mixedness testing, these lower bounds were looser and did not match the corresponding upper bound.
In light of this, we ask:
\begin{center}
    {\it Can we give an instance-optimal characterization of the copy complexity of state certification with incoherent measurements?}
\end{center}

Our second main contribution is to give such a characterization: % demonstrate that the same characterization as in~\eqref{eq:instance-optimal-bound} holds for \emph{arbitrary} incoherent measurements:
\begin{theorem}[Informal, see Theorem~\ref{thm:inst_opt}]
    For any $\sigma$, and $\eps$ sufficiently small, the copy complexity of state certification w.r.t. $\sigma$ using incoherent measurements is upper and lower bounded by~\eqref{eq:instance-optimal-bound}.
\end{theorem}
\noindent We regard this as strong evidence that, as with mixedness testing, adaptivity does not help for state certification.
It is not always a tight bound, as there are states for which the upper and lower bounds in~\eqref{eq:instance-optimal-bound} can differ by polynomial factors for some choices of $\eps$, and so this bound can be loose, even in the non-adaptive setting.
Still, we conjecture that for all $\sigma$, the copy complexity of state certification to $\sigma$ with incoherent and non-adaptive measurements is the same as that with arbitrary incoherent measurements.
Indeed, when $\eps$ is sufficiently small compared to the smallest nonzero eigenvalue of $\sigma$, our bounds are tight up to logarithmic factors.

\paragraph{Our techniques.}
We achieve our new lower bounds via a new proof technique which we believe may be of independent interest.
As with other lower bounds in this area, we reduce to a ``one-versus-many'' distinguishing problem.
To construct this instance, prior work leveraged the natural quantum analogue of Paninski's famous construction in the lower bound for (classical) uniformity testing \cite{paninski2008coincidence} -- namely, an additive perturbation by a multiple of $UZU^\dagger$, where $U$ is a Haar random matrix and $Z = \diag(1,\ldots,-1,\ldots)$ has equally many $+1$s and $-1$s. 

We instead use a different hard instance based on Gaussian perturbations.
While this introduces a number of additional technical challenges, the key advantage of this instance is that the likelihood ratio for this instance has a very clean, self-similar form (see~\eqref{eq:recursive_sketch}).
This allows us to essentially reduce the problem into one of understanding the concentration of a certain matrix martingale defined by the learning process, as well as an auxiliary matrix balancing question.
We can then use classical tools from scalar and matrix concentration to demonstrate that the likelihood ratio is close to 1 with high probability over all possible outcomes of the learning algorithm, which yields our desired lower bound.
We defer a more detailed explanation of our techniques to Section~\ref{sec:overview}.

Not only does this framework dramatically simplify many of the difficult concentration calculations in prior work such as~\cite{bubeck2020entanglement}, it also has the conceptual advantage that it never requires a \emph{pointwise} bound on the likelihood ratio.
To our knowledge, all prior lower bounds against adaptive algorithms in this literature required some worst-case pointwise bound on the likelihood ratio.
For some problems, e.g. shadow tomography~\cite{chen2022exponential}, this was already sufficient to prove tight lower bounds.
However, for mixedness testing, a worst-case bound cannot be sufficient (as we explain in Section~\ref{sec:overview}), and from a technical perspective, the fact that~\cite{bubeck2020entanglement} had to balance between their (much tighter) average case bound on the likelihood ratio and this (fairly large) worst-case bound to control the contribution of certain tail events was why their overall lower bound was loose.
Consequently, we believe that this martingale-based technique may also yield tight lower bounds for a number of other problems in the literature.

\section{Preliminaries}
Throughout, let $\rho$ denote the unknown state, and let $\rho_\mm = \tfrac{1}{d} I_d$ denote the maximally mixed state. 

\paragraph{Measurements.} We now define the standard measurement formalism, which is the way algorithms are allowed to interact with the unknown quantum state $\rho$.
\begin{definition}[Positive operator valued measurement (POVM), see e.g.~\cite{nielsen2002quantum}]
A positive operator valued measurement $\mathcal{M}$ is a finite collection of psd matrices $\mathcal{M} = \{M_z\}_{z \in \mathcal{Z}}$ satisfying $\sum_z M_z = I_d$.
When a state $\rho$ is measured using $\mathcal{M}$, we get a draw from a classical distribution over $\mathcal{Z}$, where we observe $z$ with probability $\Tr (\rho M_z)$.
Afterwards, the quantum state is destroyed.
\end{definition}
\noindent

\paragraph{Incoherent Measurements.}
Next, we formally define what we mean by an algorithm that uses incoherent measurements.
Intuitively, such an algorithm operates as follows: given $n$ copies of $\rho$, it iteratively measures the $i$-th copy using a POVM (which could depend on the results of previous measurements), records the outcome, and then repeats this process on the $(i + 1)$-th copy.
After having performed all $n$ measurements, it must output a decision based on the (classical) sequence of outcomes it has received.
More formally, such an algorithm can be represented as a tree:

\begin{definition}[Tree representation, see e.g. \cite{chen2022exponential}]\label{def:tree}
    Fix an unknown $d$-dimensional mixed state $\rho$. A learning algorithm that only uses $n$ incoherent, possibly adaptive, measurements of $\rho$ can be expressed as a rooted tree $\calT$ of depth $n$ satisfying the following properties:
    \begin{itemize}[leftmargin=*,itemsep=0pt]
        \item Each node is labeled by a string of vectors $\bx = (x_1,\ldots,x_t)$, where each $x_i$ corresponds to measurement outcome observed in the $i$-th step. 
        \item Each node $\bx$ is associated with a probability $p^{\rho}(\bx)$ corresponding to the probability of observing $\bx$ over the course of the algorithm. The probability for the root is 1.
        \item At each non-leaf node, we measure $\rho$ using a rank-1 POVM $\brc{\omega_x d \cdot xx^{\dagger}}_x$ to obtain classical outcome $x\in\S^{d-1}$. The children of $\bx$ consist of all strings $\bx' = (x_1,\ldots,x_t,x)$ for which $x$ is a possible POVM outcome.
        \item If $\bx' = (x_1,\ldots,x_t,x)$ is a child of $\bx$, then
        \begin{equation}
            p^{\rho}(\bx') = p^{\rho}(\bx)\cdot \omega_x d \cdot x^{\dagger} \rho x.
        \end{equation}
        \item Every root-to-leaf path is length-$n$. Note that $\calT$ and $\rho$ induce a distribution over the leaves of $\calT$.
    \end{itemize}
\end{definition}
\noindent
We briefly note that in this definition, we assume that the POVMs are always rank-$1$. It is a standard fact that this is without loss of generality (see e.g. \cite[Lemma 4.8]{chen2022exponential}).

\section{Technical Overview}
\label{sec:overview}

\subsection{Mixedness Testing}

We begin by describing the proof of our optimal lower bound for mixedness testing. As is standard in this line of work, we first formulate a hard ``point-vs-mixture'' distinguishing task.
Here, we specify some set of states $\{ \rho_\alpha\}_{\alpha}$, and the goal is to distinguish the case where the state $\rho$ is maximally mixed (the ``null hypothesis''), and the case where $\rho = \rho_\alpha$, where $\alpha$ is chosen from some distribution $\mathcal{D}$ (the ``alternative hypothesis'').
Our goal will be to construct such a task so that (1) $\| \rho - \rho_\alpha \|_1 > \eps$ for all $\alpha$, and (2) for any algorithm that uses incoherent measurements, if $p_0$ is the distribution over outcomes of the algorithm when run on $n$ copies of the maximally mixed state, and $p_\alpha$ is the distribution over outcomes of the algorithm when run on $n$ copies of $\rho_\alpha$, then $\TV (p_0, \mathbb{E}_{\alpha \sim \mathcal{D}} [p_\alpha]) = o(1)$ as long as $n = o (d^{3/2} / \eps^2)$.
These two facts together immediately imply our desired lower bound.

\paragraph{Gaussian perturbations.} Our first departure from prior work is in the choice of the ensemble of perturbations. All known lower bounds for mixedness testing \cite{bubeck2020entanglement,chen2021hierarchy,chen2021toward,o2015quantum}, consider alternate hypotheses of the form $\frac{1}{d}(I_d + \eps\,UZU^{\dagger})$, where $U \in \R^{d\times d}$ is a Haar-random unitary matrix and $Z = \diag(1,\ldots,-1,\ldots)$ has $\fr{d}{2}$ $+1$s and $-1$s. A drawback of working with these perturbations is that the typical ways of analyzing such distinguishing tasks involve controlling higher-order moments, but the tricky representation-theoretic structure of moments of Haar unitary matrices makes them difficult to work with.

To circumvent this, we work with a Gaussian approximation to the standard Haar-random ensemble: in place of $\frac{1}{d}(I_d + \eps\,UZU^{\dagger})$, we consider the random state $\frac{1}{d}(I_d + \eps\,M)$, where $M$ is drawn from the \emph{Gaussian orthogonal ensemble} (GOE), suitably shifted to have trace zero (see Definition~\ref{def:goe}). This new alternative hypothesis exhibits comparable tail behavior and fluctuations of the same magnitude as the original, but its moments are much more tractable to analyze and, as we will see, exhibit useful self-similar structure that will be vital to our argument.

Note that strictly speaking, as the distribution over $M$ is supported over all symmetric matrices, with some low probability $\frac{1}{d}(I_d + \eps\,M)$ may not even be psd, or it may have trace distance $\ll \eps$ from the maximally mixed state. We thus technically need to work with a distribution over $M$ where we condition out these bad events, but it turns out that the impact of this conditioning on our calculations is negligible (see Lemma~\ref{lem:goe-trunc} in the proof of Theorem~\ref{thm:main_standard}), and in this overview we will work without conditioning, for simplicity.

\paragraph{Primer on adaptive lower bounds.} Having specified the distinguishing task, we now briefly review the usual framework for proving lower bounds against adaptively chosen incoherent measurements. Recall from Definition~\ref{def:tree} that any learning strategy that uses such measurements can be thought of as specifying a tree, where each internal node corresponds to the transcript of measurement outcomes seen so far, and the edges emanating from that node correspond to the possible outcomes of the POVM that gets chosen to measure the next copy of $\rho$. At any leaf node, the learner decides based on all the outcomes they have seen along their root-to-leaf path whether the node is maximally mixed or not. As the probabilities for transitioning from any given node to one of its children depend on the unknown state being measured, we can thus think of the null hypothesis and alternative hypothesis as inducing two different distributions $p_0$ and $p_1$ over the leaves of the tree. As described above, to show our lower bound for mixedness testing, it suffices to show that for $n = o(d^{3/2}/\eps^2)$, the total variation distance between these distributions satisfies $d_{\mathrm{TV}}(p_0,p_1) = o(1)$.

The main challenge in controlling $d_{\mathrm{TV}}(p_0,p_1)$, and also the key difference from classical distribution testing, is the adaptivity in the measurements. Whereas \cite{bubeck2020entanglement} dealt with this by passing to KL divergence and using chain rule, we will instead work directly with the total variation distance.

% \begin{enumerate}[leftmargin=*,itemsep=0pt]
%     \item Pass to KL divergence via Pinsker's and use chain rule \cite{bubeck2020entanglement,chen2021toward}
%     \item Explicitly compute $d_{\mathrm{TV}}(p_0,p_1)$ \cite{aharonov2022quantum,chen2022exponential}
%     \item Show that the \emph{likelihood ratio}, that is, the ratio between the probability of reaching any leaf $\bx$ under $p_1$ versus under $p_0$ is not too small for all leaves $\bx$ \cite{huang2021information,chen2022exponential,chen2021hierarchy}
% \end{enumerate}
% While 1. was the basis for the previous best lower bound of $\Omega(d^{4/3}/\eps^2)$, 2. and 3. were only implemented in 

\paragraph{Likelihood ratio martingale.}

In this overview, we will assume for simplicity that every POVM used by the learner consists of rank-1 projectors $yy^{\dagger}$ to some (adaptively chosen) orthonormal basis.

To bound the total variation distance, we focus on controlling the \emph{likelihood ratio} $L(\bx)$, i.e. the ratio between the probability masses that $p_1$ and $p_0$ place on a given leaf $\bx$. As $d_{\mathrm{TV}}(p_0,p_1) = \E{|L(\bx) - 1|}$, where the expectation is over $\bx\sim p_0$, it is enough to show that $L(\bx) \approx 1$ with high probability over $p_0$. Henceforth we will thus think of $L(\bx)$ as a random variable where $\bx\sim p_0$.

Note that for any leaf $\bx = (x_1,\ldots,x_n)$ specifying a transcript of measurement outcomes corresponding to rank-1 POVM elements $x_1x_1^{\dagger},\ldots,x_nx_n^{\dagger}$, the likelihood ratio between reaching $\bx$ under the alternative hypothesis versus under the null hypothesis can be expressed as
\begin{equation}
    L(\bx) \triangleq \frac{p_1(\bx)}{p_0(\bx)} = \E[M]*{\prod^n_{i=1}(1 + \eps\, x_i^{\dagger}Mx_i)}. \label{eq:Ldef_sketch}
\end{equation}
We can also extend this to non-leaf nodes $\bx$: if $\bx = (x_1,\ldots,x_t)$ is a partial transcript for some $t < n$, then $L(\bx) = \E[M]{\prod^t_{i=1}(1+\eps\, x_i^{\dagger}Mx_i)}$ is simply the ratio between the probability of reaching $\bx$ after $t$ measurements under the alternative hypothesis versus under the null hypothesis.

Roughly speaking, our strategy will be to track the evolution of the likelihood ratio as $t$ increases. Note that for a fixed node $\bx$, if $\bx'$ is the random child node that one transitions to upon measuring another copy of the maximally mixed state, then $\E{L(\bx')/L(\bx)} = 1$. In other words, the likelihood ratio evolves like a \emph{multiplicative martingale} indexed by $t$. While this is a basic feature of any likelihood ratio between two sequences of random variables, we are not aware of prior work in quantum learning that exploits this, whereas for us this will be essential to dealing with adaptivity.

We pause to remark that while there have been a number of previous works establishing quantum testing lower bounds by bounding the likelihood ratio \cite{chen2022exponential,chen2021hierarchy,huang2021information}, in their settings they simply show that the likelihood ratio is bounded for \emph{every leaf}. In contrast, in mixedness testing, such a strategy cannot work, as there can be leaves which are much rarer under the alternative hypothesis than the null hypothesis. For instance, if the algorithm always measures in the standard basis, then a transcript which consists of an equal number of every measurement outcome will be much rarer under the alternative hypothesis than the null. %\sitan{this kinda interrupts the flow, maybe we want to move this elsewhere?}

% \begin{remark}\label{remark:pointwise}
%     \jerry{why is this in a remark?} While there have been a number of previous works establishing quantum testing lower bounds by directly bounding the likelihood ratio \cite{chen2022exponential,chen2021hierarchy,huang2021information}, in their settings those works simply show that the likelihood ratio is lower bounded for \emph{every leaf}. In contrast, in mixedness testing, such a strategy cannot work, as there can be leaves which are much rarer under the alternative hypothesis than the null hypothesis. For instance, if the algorithm always measures in the standard basis, then a transcript which consists of an equal number of every measurement outcome will be exponentially rarer under the alternative hypothesis than the null. %\sitan{check whether this true}
% \end{remark}

\paragraph{Recursive structure of $L$.} 

We now explain how our choice of Gaussian ensemble makes controlling the likelihood ratio martingale particularly convenient. By Isserlis' theorem, one can evaluate \eqref{eq:Ldef_sketch} explicitly: for (leaf or internal node) $\bx'$ given by a transcript $x_1,\ldots,x_t,x_{t+1}$, we get
\begin{equation}
    L(\bx') = \sum^{\lfloor (t+1)/2\rfloor}_{k=0} \left(\frac{2\eps^2}{d^2}\right)^k \sum_{\brc{\brc{a_i,b_i}}} \prod^k_{i=1} (d\iprod{x_{a_i},x_{b_i}}^2 - 1), \label{eq:Lsketch}
\end{equation}
where the latter sum is over all partial matchings of $\brc{1,\ldots,t+1}$ consisting of $k$ pairs.
% \bh{should this be $\brc{1,\ldots,t+1}$?} \sitan{yup, thanks!}
Now observe that the expression~\eqref{eq:Lsketch} contains a copy of the likelihood ratio for the \emph{parent} of $\bx'$. If $\bx$ is the parent corresponding to transcript $x_1,\ldots,x_t$, then $L(\bx)$ is precisely the sum of the terms in \eqref{eq:Lsketch} given by partial matchings which only consist of $x_s$ for $1 \le s\le t$. Moreover, the remaining terms given by partial matchings that contain $x_{t+1}$ also contain likelihood ratio-like terms. Specifically, defining $L(\bx_{\sim i}) \triangleq \E[M]*{\prod_{j\in[t]: j\neq i} (1 +\eps\, x^{\dagger}_j M x_j)}$,\footnote{Note that strictly speaking the transcript $x_1,\ldots,x_{i-1},x_{i+1},\ldots,x_t$ does not appear in the tree (unless $i = t$), but this quantity is still well-defined even if it is not a ``real'' likelihood ratio.} one can verify (Lemma~\ref{lem:L-recursion}) that
\begin{equation}
    L(\bx') = L(\bx) + \frac{2\eps^2}{d^2}\sum^t_{i=1} (d\iprod{x_i,x_{t+1}}^2-1)\cdot L(\bx_{\sim i}). \label{eq:recursive_sketch}
\end{equation}

Now consider the following thought experiment. Imagine for the moment that $L(\bx_{\sim i}) \approx L(\bx)$ for all $i\in[t]$. Then we could divide by $L(\bx)$ on both sides of \eqref{eq:recursive_sketch} to get that
\begin{equation}
    \frac{L(\bx')}{L(\bx)} - 1 \approx \frac{2\eps^2}{d^2}\sum^t_{i=1} (d\iprod{x_i,x_{t+1}}^2-1) = \frac{2\eps^2}{d^2}\, x^{\dagger}_{t+1} \biggl(\sum^t_{i=1} (dx_ix_i^{\dagger} - I_d)\biggr) x_{t+1}. \label{eq:fake_sketch}
\end{equation}
As $dx_1x_1^{\dagger}-I_d, dx_2x_2^{\dagger}-I_d, \cdots$ is a matrix martingale difference sequence, by matrix Freedman \cite[Theorem 1.2]{tropp2011freedman} we expect the right hand side of \eqref{eq:fake_sketch} to have fluctuations of order roughly $\pm (\eps^2/d^2)\cdot\sqrt{td} = \pm (\eps^2/d^{3/2})\cdot \sqrt{t}$ (ignoring logarithmic factors). 
% \bh{strictly speaking this suffers a log} 
In other words, the likelihood ratio martingale jumps by a multiplicative factor of $1\pm (\eps^2/d^{3/2})\cdot \sqrt{t}$ in every step, which means that cumulatively over $n$ steps, it changes by a multiplicative factor of $1 \pm (\eps^2/d^{3/2})n$ with high probability. So if $n = o(d^{3/2}/\eps^2)$, the likelihood ratio is $1 + o(1)$ with high probability over the leaves as desired, and we get the optimal lower bound for mixedness testing.

\paragraph{Bootstrapping.} This thought experiment is of course inherently circular. Our goal was to show that the likelihood ratio doesn't change very much, but to prove this we assumed that $L(\bx_{\sim i}) \approx L(\bx)$, i.e. that removing one element from the transcript doesn't change the likelihood ratio very much! Here we outline our approach for resolving this chicken-and-egg problem. The high-level idea is that for $n \le O(d^{3/2}/\eps^2)$, it is actually easy to show that the likelihood ratio can never change by more than a $1 + o(1)$ factor in a single step (see e.g. \eqref{eq:ratio-crude}). For the likelihood ratio martingale argument to work, we need a more refined bound on these multiplicative jumps on the order of $1\pm (\eps^2/d^{3/2})\cdot\sqrt{n}$, which we will achieve by recursively bootstrapping the cruder bound\--- see the proof of Lemma~\ref{lem:bootstrap}, which we now sketch.

First, note that the correct version of \eqref{eq:fake_sketch}, without approximation, is actually given by
% \begin{align}
%     \frac{L(\bx')}{L(\bx)} - 1 &= \frac{2\eps^2}{d^2}\, x^{\dagger}_{t+1} \biggl(\sum^t_{i=1} \frac{dx_ix_i^{\dagger} - I_d}{L(\bx)/L(\bx_{\sim i})}\biggr) x_{t+1} \\
%     &= \frac{2\eps^2}{d^2}x^{\dagger}_{t+1}\biggl(\sum^t_{i=1} (dx_ix^{\dagger}_i - I_d)\biggr) + \frac{2\eps^2}{d^2}x^{\dagger}_{t+1}\biggl(\sum^t_{i=1} (dx_ix^{\dagger}_i - I_d)\cdot \biggl(\frac{L(\bx_{\sim i})}{L(\bx)} - 1\biggr)\biggr)x_{t+1}. \label{eq:discrep_sketch}
% \end{align}
\begin{equation}
    \frac{L(\bx')}{L(\bx)} - 1 = \frac{2\eps^2}{d^2}x^{\dagger}_{t+1}\biggl(\sum^t_{i=1} (dx_ix^{\dagger}_i - I_d)\biggr)x_{t+1} + x^{\dagger}_{t+1}\biggl(\underbrace{\frac{2\eps^2}{d^2}\sum^t_{i=1} (dx_ix^{\dagger}_i - I_d)\cdot \biggl(\frac{L(\bx_{\sim i})}{L(\bx)} - 1\biggr)}_{\Delta}\biggr)x_{t+1}. \label{eq:discrep_sketch}
\end{equation}
So the quantity dictating how much the thought experiment deviates from reality is the operator norm of the matrix $\Delta$ in \eqref{eq:discrep_sketch}. Suppose inductively that we have shown that each of the multiplicative jumps $\frac{L(\bx)}{L(\bx_{\sim i})}$ is bounded by $1\pm \alpha$ for some $0 < \alpha \ll 1$. Then we can upper bound $\Delta$ by
\begin{equation}
    \norm{\Delta}_{\op} \le O\left(\frac{\eps^2}{d}\cdot \alpha\right)\cdot \sup_{b_1,\ldots,b_t\in[-1,1]} \biggl\|\sum^t_{i=1} b_i(dx_ix_i^{\dagger} - I_d)\biggr\|_{\op}. \label{eq:sup_sketch}
\end{equation}
If $\sum^t_{i=1}(dx_ix_i^{\dagger} - I_d)$ is close to its typical value of $\sqrt{td}$ and $t = \Theta(n)$, then it is not hard to show using a few applications of triangle inequality that the supremum above is upper bounded by $O(t)$ (see Lemma~\ref{lem:uniform-frob-bd}).
% \sitan{i guess technically this lemma is for Frobenius norm instead of operator because we want to use scalar martingale in the outer argument and avoid paying the log factor from matrix Freedman. could this part of the overview be confusing to the reader?}). \bh{let's discuss on Wednesday what to do about this... this seems annoying to write lol} 
In this case, $\norm{\Delta} \le O\left(\frac{\eps^2 t}{d^2} \cdot \alpha\right)$, whereas recall that the other term in \eqref{eq:discrep_sketch} is of order $(\eps^2/d^{3/2})\cdot \sqrt{t}$.

The upshot is that we have bootstrapped a bound of $1 \pm \alpha$ on the multiplicative jumps into a better bound on the next multiplicative jump $L(\bx')/L(\bx)$ which is of order
\begin{equation}
    1 \pm \left(\frac{\eps^2}{d^{3/2}}\cdot \sqrt{t} + \frac{\eps^2 t}{d^2}\cdot\alpha\right).
\end{equation}
In particular, because $t\le n \ll d^2/\eps^2$, our bound has contracted towards the ideal bound of $(\eps^2/d^{3/2})\cdot \sqrt{t}$ from the thought experiment! Repeating this bootstrapping $O(\log n)$ many rounds and noting that the matrices $\sum_{s\in S} (dx_{i_s}x_{i_s}^{\dagger} - I_d)$, for $S\subseteq [t], |S| \ge t - O(\log n)$, that arise in recursive applications of the argument above will not be that different from $\sum^t_{i=1}(dx_ix_i^{\dagger} - I_d)$, we ensure that $\Delta$'s contribution to \eqref{eq:discrep_sketch} becomes negligible, thus resolving the chicken-and-egg problem.
% \bh{furthermore the matrix sum can't change too much over $O(\log n)$ terms, so if $\sum^t_{i=1}(dx_ix_i^{\dagger} - I_d)$ isn't too huge in operator/frobenius norm neither are the sub-sums where you drop $O(\log n)$ terms. so this whole bootstraping argument proceeds ``deterministically" on the event that $\sum^t_{i=1}(dx_ix_i^{\dagger} - I_d)$ is the right scale} \sitan{slightly edited to reflect this, feel free to tweak if it's too sketchy}

\paragraph{Log factors.} As described, the above would appear to only achieve the optimal bound of $d^{3/2}/\eps^2$ up to log factors. For one, we are using matrix martingale concentration to bound $\sum^t_{i=1} (dx_ix_i^{\dagger} - I_d)$ and its operator norm thus has fluctuations of order $\sqrt{td\log d}$ rather than $\sqrt{td}$. We also appear to be conditioning on concentration holding for all $t\in[n]$, thus losing another log factor.

To avoid this, instead of bounding the multiplicative jumps pointwise using operator norm, we directly bound the \emph{second moment} of the multiplicative jumps using \emph{expected Frobenius norm}. More precisely, we show that it suffices to control the expected maximum of $\|\sum^t_{i=1} (dx_ix_i^{\dagger}-I_d)\|^2_F$ across $1\le t \le n$ (see Lemma~\ref{lem:doob}). This can then be bounded without additional log factors using an argument reminiscent of the proof of Doob's $L^2$ maximal inequality (see Section~\ref{sec:doob}).
%We also note that our proof for state certification, which we describe next, avoids log factors via similar tricks (see Section~\ref{sec:doob-offdiag}).
% \bh{just noting that sections 3 and 4 both don't lose logs, only the final reduction still does. we could say something about having fully tight results against both ``perturbation types"}

% \bh{quibble with the current wording: our proof for general state certification requires the section 5 reduction, which loses logs. our proof for the off diagonal perturbation type, when that's the entire matrix, does not lose logs}\sitan{hmm true, maybe i'll just comment out the mention about state certification and we can talk about the off-diagonal not losing logs somewhere in the intro}

\subsection{State certification} 

Here we describe how to extend these techniques to the more general setting of state certification with respect to an arbitrary state $\sigma$. Without loss of generality we will assume $\sigma$ is diagonal.

\paragraph{Eigenvalue bucketing.} We first describe the hard distinguishing task that we consider. \cite{chen2021toward} gave a reduction, up to log factors, from showing instance-optimal lower bounds for state certification with respect to arbitrary $\sigma$, to showing such bounds when $\sigma$ takes one of two forms:
\begin{enumerate}[itemsep=0pt]
    \item[(A)] $\sigma$ has eigenvalues that are all within a small multiplicative factor of $1/d$
    \item[(B)] There are two values $0\le a,b\le 1$ such that each of $\sigma$'s eigenvalues is within a small multiplicative factor of either $a$ or $b$.
\end{enumerate}
For completeness, we give a self-contained proof of this reduction in Section~\ref{sec:full}. At a high level, the idea is that we divide the eigenvalues of $\sigma$ into logarithmically many buckets where in each bucket, any two eigenvalues are multiplicatively close. 
Then, the hardest possible distinguishing task one can formulate, up to log factors, is to take the alternative hypothesis to either perturb the submatrix of $\sigma$ corresponding to a single bucket (this submatrix corresponds to category A above), or to perturb the off-diagonal submatrices of $\sigma$ corresponding to a pair of buckets (the submatrix of entries from these two buckets corresponds to category B above). The former distinguishing task is sufficient to show optimal lower bounds for states $\sigma$ like the maximally mixed state, whereas the latter may be harder e.g. for certain approximately low-rank $\sigma$.
% For instance, if $\sigma = \diag(1-1/d^2,1/d,\ldots,1/d)$, the former yields a lower bound which is decreasing in $d$, whereas the latter actually yields an (optimal) lower bound of $\sqrt{d}/\eps^2$.

For $\sigma$ in category A, the lower bound follows by a simple modification of our analysis for mixedness testing. 
This proof is presented in Section~\ref{sec:paninski}, and includes the proof of the mixedness testing lower bound as a special case.
The remaining technical challenge is to prove the lower bound for category B, which we now sketch.
This proof is carried out in Section~\ref{sec:offdiag}.

\paragraph{Off-diagonal perturbations.} For simplicity, consider $\sigma$ of the form $(a\cdot I_{d_1})\oplus(b\cdot I_{d_2})$ for $a,b> 0$ and $d_1 \ge d_2$. Concretely, the distinguishing task considered in \cite{chen2021toward} is the following. The null hypothesis is that $\rho = \sigma$, and the alternative hypothesis is that
\begin{equation}
    \rho = \begin{pmatrix}
        a\cdot I_{d_1} & \frac{\eps}{d_2}W \\
        \frac{\eps}{d_2}W^{\dagger} & b\cdot I_{d_2}
    \end{pmatrix}, \label{eq:off_diag_sketch}
\end{equation}
where $W$ consists of the first $d_2$ columns of a Haar-random $d_1\times d_1$ unitary. Motivated by the Gaussian perturbations used in our proof for mixedness testing, here we consider a Gaussian version of this alternative hypothesis where we instead take $W$ to be a $d_1\times d_2$ matrix whose entries are independent mean-zero Gaussians with variance $1/d_1$ (see Definition~\ref{def:ginibre}).

% As we show in Theorem~\ref{thm:main_offdiag}, the copy complexity for this task turns out to be $d_2\sqrt{d_1}/\eps^2$.\jerry{expand on why this is sufficient? or remove}

\paragraph{Likelihood ratio pitfalls.} To prove this, our goal as before is to show that the likelihood ratio between the distributions $p_1, p_0$ over leaves of the learning tree induced by the alternative and null hypotheses is close to 1 with high probability with respect to $p_0$. Here it will be convenient to refer to a transcript $\bx = (x_1,\ldots,x_t)$ as $(\bz,\bw) = ((z_1,w_1),\ldots,(z_t,w_t))$, where $z_i\in\mathbb{C}^{d_1}, w_i\in\mathbb{C}^{d_2}$.
We can explicitly compute the likelihood ratio to be
\begin{equation}
    L((\bz,\bw)) = \E[W]*{\prod^n_{i=1}\left(1 + \frac{2\eps}{d_2}\cdot \frac{z^{\dagger}_i W w_i}{a\|z_i\|^2 + b\|w_i\|^2}\right)},
\end{equation}
and analogously to \eqref{eq:recursive_sketch}, we can prove (see Lemma~\ref{lem:L-recursion-offdiag}) that this likelihood ratio has the following nice recursive form. For (leaf or internal node) $(\bz',\bw')$ corresponding to the transcript $((z_1,w_1),\ldots,$ $(z_{t+1},w_{t+1}))$, if $(\bz,\bw)$ is its parent corresponding to transcript $((z_1,w_1),\ldots,(z_t,w_t))$, then
\begin{equation}
    L((\bz',\bw')) = L((\bz,\bw)) + \frac{4\eps^2}{d_1d^2_2} \sum^t_{i=1}\brk*{\frac{\iprod{z_i,z_{t+1}}\iprod{w_i,w_{t+1}}}{(a\|z_i\|^2+b\|w_i\|^2)(a\|z_{t+1}\|^2+b\|w_{t+1}\|^2)}\cdot L((\bz,\bw)_{\sim i})}. \label{eq:recursive_sketch2}
\end{equation}
The first indication that this distinguishing task could be harder to analyze is the $a\|z\|^2+b\|w\|^2$ terms that appear in the denominator. For the parameter regimes where we consider this distinguishing task, it turns out that $a$ can be quite small. 
% \bh{technically not arbitrarily small: we require $\eps \le \fr{1}{100} d_2 \sqrt{ab}$ so the perturbation makes sense.} \sitan{is it fine now with the stipulation that $\eps\to 0$?} \bh{still no - $a$ isn't arbitrarily small compared to $\eps$, because $a$ is lower bounded by something depending on $\eps$. can we just say things go bad when $a$ is small, without saying arbitrarily small?} \sitan{fair point}
So any POVM with elements that are aligned with the directions corresponding to the $a\cdot I_{d_1}$ block will lead to measurement outcomes that are rare under the null hypothesis, but not necessarily under the alternative hypothesis.

To see how this issue arises, consider the thought experiment where we imagine  $L((\bz,\bw)_{\sim i}) \approx L((\bz,\bw))$ for every $i$. Then if we divide by $L((\bz,\bw))$ on both sides of \eqref{eq:recursive_sketch2} and define
\begin{equation}
    K_t \triangleq \sum^t_{i=1} \frac{z_iw_i^{\dagger}}{a\|z_i\|^2+b\|w_i\|^2},
\end{equation}
we get
\begin{equation}
    \frac{L((\bz',\bw'))}{L((\bz,\bw))} - 1 \approx \frac{4\eps^2}{d_1d^2} \cdot \frac{z_{t+1}^{\dagger} K_t w_{t+1}}{a\|z_{t+1}\|^2+b\|w_{t+1}\|^2}. \label{eq:divide_2_sketch}
\end{equation}
The matrix $K_t$ is the analogue of the $\sum^t_{i=1} dx_ix_i^{\dagger} - I_d$ from mixedness testing. Because
\begin{equation}
    \frac{|z^{\dagger}_{t+1} K_t w_{t+1}|}{a\|z_{t+1}\|^2+b\|w_{t+1}\|^2} \le \frac{|z^{\dagger}_{t+1} K_t w_{t+1}|}{2\sqrt{ab}\|z_{t+1}\| \|w_{t+1}\|} \le \frac{\norm{K_t}_{\op}}{2\sqrt{ab}},
\end{equation} we might be tempted to imitate the proof for mixedness testing by bounding $\norm{K_t}_{\op}$ using matrix Freedman. Unfortunately this doesn't work: as $a \to 0$, with high probability the operator norm of this matrix is of order at least $\sqrt{t/b}$, so by \eqref{eq:divide_2_sketch}, the multiplicative jumps in the likelihood ratio martingale are of order
$1 \pm \frac{\eps^2}{d_1d^2_2}\cdot \frac{\sqrt{t/b}}{\sqrt{ab}} = 1\pm \frac{\eps^2}{d_1d^2_2\sqrt{a}b}\cdot \sqrt{t}$. So cumulatively over $n$ steps, the likelihood ratio changes by a multiplicative factor of $1\pm \frac{\eps^2}{d_1d^2_2\sqrt{a}b}\cdot n$. This translates to a copy complexity lower bound of $d_1d^2_2\sqrt{a}b/\eps^2$. When $a$ and $b$ are both of order $1/d$, this recovers the $d^{3/2}/\eps^2$ lower bound for mixedness testing.\footnote{The reason we didn't also use this off-diagonal perturbation to prove our mixedness testing lower bound is that this instance is only well-defined for $\eps$ sufficiently small; otherwise, the instance \eqref{eq:off_diag_sketch} is not psd.} But when $a \to 0$, this lower bound becomes vacuous.

\paragraph{From operator to Frobenius.} In other words, for this distinguishing task, working with the operator norm is too crude even in the thought experiment! Intuitively the issue is that it yields a \emph{uniform} upper bound on the magnitude of every multiplicative jump, regardless of $(z_{t+1},w_{t+1})$. But given that there can be measurement outcomes which are extremely unlikely under the null hypothesis and thus induce rare, huge jumps in the likelihood ratio, it makes more sense to give an upper bound on the magnitude of a \emph{typical} multiplicative jump.

To bound a typical jump, we thus look at the second moment of the jump $\frac{L((\bz',\bw'))}{L((\bz,\bw))} - 1$ as a random variable in $(z_{t+1},w_{t+1})$ under the null hypothesis: 
\begin{align}
    \E[(z_{t+1},w_{t+1})]*{\left(\frac{z^{\dagger}_{t+1} K_t w_{t+1}}{a\|z_{t+1}\|^2+b\|w_{t+1}\|^2}\right)^2} &\le \sum_{(z_{t+1},w_{t+1})} \frac{z^{\dagger}_{t+1} K_t (w_{t+1}w^{\dagger}_{t+1}) K_t^{\dagger} z_{t+1}}{b\|w_{t+1}\|^2} \\
    &\le \frac{1}{b}\sum z^{\dagger}_{t+1} K_tK_t^{\dagger} z_{t+1} = \frac{1}{b}\|K_t\|^2_F, \label{eq:secondmoment_sketch}
\end{align}
where in the second step we used that $w_{t+1}w_{t+1}^{\dagger}/\|w_{t+1}\|^2 \preceq I_{d_2}$.

It is not hard to show that $\|K_t\|^2_F$ is typically of order $td_1d_2$ (see Lemma~\ref{lem:doob-offdiag}). So by \eqref{eq:secondmoment_sketch}, the typical multiplicative jump in the likelihood ratio martingale is of order $1\pm \frac{\eps^2}{d_1d^2_2}\cdot \frac{\sqrt{td_1d_2}}{\sqrt{b}} = 1 \pm \frac{\eps^2}{\sqrt{d_1 d^3_2 b}}\cdot \sqrt{t}$. So cumulatively over $n$ steps, the likelihood ratio changes by a factor of $1\pm \frac{\eps^2}{\sqrt{d_1d^3_2 b}}\cdot n$. This translates to a copy complexity lower bound of $\sqrt{d_1d^3_2 b}/\eps^2$. In the parameter regime we care about, $d_2b \ge \Omega(1)$ (see Fact~\ref{fact:d2b}), so this yields the (optimal) lower bound of $d_2\sqrt{d_1}/\eps^2$.

\paragraph{Bootstrapping.} As with our proof for mixedness testing, the above thought experiment is circular. If we no longer pretend that $L((\bz,\bw)_{\sim i}) \approx L((\bz,\bw))$ for every $i$, then in place of $K_t$, the matrix whose Frobenius norm we actually need to bound is
\begin{equation}
    % \Delta \triangleq \frac{4\eps^2}{d_1d^2_2} \cdot \frac{1}{a\|z_{t+1}\|^2+b\|w_{t+1}\|^2} \cdot z^{\dagger}_{t+1}\left(\sum^t_{i=1} \frac{z_iw^{\dagger}_i}{a\|z_i\|^2 + b\|w_i\|^2}\cdot \left(\frac{L((\bz,\bw)_{\sim i}}{L((\bz,\bw))} - 1\right)\right)w_{t+1}
    H_t \triangleq \sum^t_{i=1} \frac{z_iw^{\dagger}_i}{a\|z_i\|^2 + b\|w_i\|^2}\cdot \frac{L((\bz,\bw)_{\sim i})}{L((\bz,\bw))} =  K_t + \underbrace{\sum^t_{i=1} \frac{z_iw^{\dagger}_i}{a\|z_i\|^2 + b\|w_i\|^2}\cdot \left(\frac{L((\bz,\bw)_{\sim i})}{L((\bz,\bw))} - 1\right)}_{\Delta}, \label{eq:Hsketch_offdiag}
\end{equation}
but controlling $H_t$ relies on recursively controlling $\frac{L((\bz,\bw)_{\sim i})}{L((\bz,\bw))} - 1$. We solve this chicken-and-egg problem by bootstrapping the following crude upper bound. The idea is that for ``$H_t$-like'' matrices, %\bh{missing some word here?}, 
the Frobenius norm can always be very loosely upper bounded by $n/\sqrt{ab}$, essentially because the multiplicative jumps in the likelihood ratio are never greater than $O(1)$ (see Lemma~\ref{lem:crude-frob-ub-offdiag})\--- we note that the precise polynomial dependence on $n$ in this crude bound is unimportant, as our goal will be to contract this bound by a constant factor in each of $O(\log(n))$ rounds of bootstrapping.
% \bh{can we say that this $n^2$ is a very crude bound; anything works here, even $n^{1000}$. the important point is the contraction of the $\Delta$ by a constant factor each round.}\sitan{good point! how does the new version look?}\bh{great!} 

So if we apply the aforementioned \emph{operator norm} bound to control $\frac{L((\bz,\bw)_{\sim i})}{L((\bz,\bw))} - 1$ and naively upper bound the operator norm of the resulting $H_t$-like matrix by its Frobenius norm, we get
\begin{equation}
    \frac{L((\bz,\bw)_{\sim i})}{L((\bz,\bw))} - 1 \le \frac{4\eps^2 n}{d_1d^2_2\sqrt{ab}}\cdot \frac{\norm{z_i}\norm{w_i}}{a\|z_i\|^2+b\|w_i\|^2}.
\end{equation}
Substituting this into the right-hand side of \eqref{eq:Hsketch_offdiag}, we obtain the following analogue of \eqref{eq:sup_sketch}:
\begin{equation}
    \norm{\Delta}_F \le \frac{4\eps^2 n}{d_1d^2_2\sqrt{ab}}\cdot \sup_{b_1,\ldots,b_t\in[-1,1]} \norm{\sum^t_{i=1} b_i \frac{z_iw_i^{\dagger}\|z_i\|\|w_i\|}{(a\|z_i\|^2+b\|w_i\|^2)^2}}_F. \label{eq:sup_sketch_2}
\end{equation}
As we show in Lemma~\ref{lem:balanced-offdiag}, with high probability over $(\bz,\bw)$ this supremum is at most $\frac{1}{10}d_1d^2_2/\eps^2$, so $\norm{\Delta}_F \le n/2\sqrt{ab}$. Before we sketch how to prove this, let us see how to conclude the argument.

Indeed, by plugging the bound on the supremum into \eqref{eq:Hsketch_offdiag}, we find that we have bootstrapped a crude bound of $n/\sqrt{ab}$ on the Frobenius norm of the ``$H_t$-like'' matrices that dictate the preceding multiplicative jumps $\frac{L((\bz,\bw)_{\sim i})}{L((\bz,\bw))}$ into a better bound on the Frobenius norm of $H_t$, namely
\begin{equation}
    \norm{H_t} \le \norm{K_t} + n/2\sqrt{ab}. \label{eq:H_bound_sketch}
\end{equation}
By repeating this bootstrapping logarithmically many rounds, we can thus shrink the second term in \eqref{eq:H_bound_sketch} until it is dominated by the contribution from $\norm{K_t}$, showing that the above thought experiment is valid.

\paragraph{Supremum bound.} Recall that for mixedness testing, we could show that the analogous supremum was bounded as long as $\|\sum^t_{i=1} (dx_ix_i^{\dagger} - I_d)\|_{\op}$ was (Lemma~\ref{lem:uniform-frob-bd}). Analogously, one might hope that \eqref{eq:sup_sketch_2} is bounded as long as $\norm{K_t}_F$ is. Unfortunately, this turns out to be false (see Appendix~\ref{app:Kvskappa}), essentially because the off-diagonal structure of the distinguishing task makes it possible for $K$ to be small, in fact zero, even under extremely atypical transcripts (e.g. consider a transcript that repeatedly alternates between a vector $(z,w)$ and the vector $(z,-w)$), whereas the supremum for such transcripts will be extremely large.

This necessitates an entirely different argument for the supremum. The proof involves a careful net argument that is facilitated by a judicious application of Grothendieck's inequality. We defer the details to Section~\ref{sec:balanced-offdiag}.

\paragraph{Roadmap.}In Section~\ref{sec:related} we survey relevant prior work. In Section~\ref{sec:add_not} we provide additional technical preliminaries and formally define the ensembles of perturbations we use.  In Section~\ref{sec:paninski}, we prove our lower bound for mixedness testing, and in Section~\ref{sec:offdiag}, we prove our lower bound for the distinguishing task involving ``off-diagonal'' perturbations that was described in the overview. In Section~\ref{sec:full} we state our instance-optimal lower bound for state certification and use the results of Section~\ref{sec:paninski} and \ref{sec:offdiag} to give a simple proof of a slightly weaker version of it. In Appendix~\ref{app:block} and \ref{app:full_full} we refine our analysis to give a full proof of the instance-optimal bound. In Appendix~\ref{sec:basic-comps} we present the deferred proofs that the bad events we condition out when we define our Gaussian perturbations occur with small probability.

\section{Related Work}\label{sec:related}
A full literature review on quantum (and classical) testing is out of the scope of this paper. 
For concision we only discuss some of the more relevant works below.

The questions we consider in this paper fall under the domain of quantum state property testing.
See~\cite{montanaro2016survey} for a more complete survey on property testing of quantum states.
In this literature, roughly speaking, there are two settings considered, the \emph{asymptotic regime}, and the \emph{non-asymptotic regime}, the latter of which is the setting we study.

In the former setting, one considers the regime of parameters where $n \to \infty$ and $d, \eps$ are held fixed and relatively small, and the goal is to precisely characterize the exponential rate of convergence as a function of $n$.
In this setting, quantum state certification is usually called \emph{quantum state discrimination}, see e.g.~\cite{chefles2000quantum,audenaert2008asymptotic,barnett2009quantum} and references within.
However, since $d$ and $\eps$ are fixed, this allows for rates which could depend exponentially on the dimensionality of the problem.

Instead, we consider the ``non-asymptotic regime,'' where the goal is to characterize the statistical rate, as a function of $d$ and $\eps$.
Similar work in this regime includes the aforementioned works of~\cite{o2015quantum} and~\cite{buadescu2019quantum}.
However, as described previously, their algorithms require using fully entangled measurements.

Our work falls into the line of work considering restricted classes of measurements, and specifically, those with without quantum memory.
Understanding the power of such algorithms in the context of mixedness testing and,  more generally, spectrum testing was posed as an open problem in~\cite{wright2016learn}.
Similar questions have also been considered in other settings, such as shadow tomography~\cite{aaronson2018shadow}. 
However, until recently, lower bounds for algorithms without quantum memory usually only held in the non-adaptive setting, e.g.~\cite{haah2017sample,chen2021toward}.
Recent work of~\cite{bubeck2020entanglement} demonstrated the first lower bound against general (possibly adaptive) incoherent measurements for such a task.  
Subsequently, there has been a flurry of work demonstrating similar bounds in a variety of settings~\cite{huang2020predicting,aharonov2022quantum,huang2021information,huang2021quantum,chen2022exponential,anshu2021distributed,chen2021hierarchy,lowe2021learning,chen2022quantum}.
It is an interesting question if our techniques can be extended to also improve any of the lower bounds in these works.

Other restricted models of computation have also been considered in the literature.  \cite{yu2019quantum} gives algorithms for various quantum property testing problems using local measurements which act on each individual qubit, and in an non-adaptive manner. A number of works considers the special case where the measurements are only Pauli matrices~\cite{flammia2011direct,flammia2012quantum,da2011practical,aolita2015reliable}.
Overall, these classes of measurements seem to be much more restrictive than general non-adaptive measurements.
In particular, the copy complexity of tasks such as mixedness testing under these measurements seem to be asymptotically higher than general incoherent measurements.

\section{Additional Preliminaries}\label{sec:add_not}
\paragraph{Notation.} Given $z\in\R$, we use $z_-$ to denote $-\min(z,0)$. We use $\wedge$ and $\vee$ to denote min and max.
We use $f\lesssim g$ to denote $f = O(g)$, $f\ll g$ to denote $f = o(g)$, and $f \llsim g$ to denote that there exists some absolute constant $c$ for which $f = o(g / \log^c g)$. 
We will always implicitly assume a sufficiently large system; for example, if $f \gg g$ we will assume where necessary that $f \ge 100g$.
We use $f = \wt{O}(g)$ (resp. $f = \wt{\Omega}(g)$) to denote that there exists some absolute constant $c$ for which $f = O(g\cdot \log^c g)$ (resp. $f = \Omega(g/\log^c g)$).

Given a vector $v$, we use $\norm{v}_p$ to denote its $\ell^p$ norm; when $p = 2$, we sometimes drop the subscript. Given a matrix $M$, we use $\norm{M}_{\op}$ or $\norm{M}$ to denote its operator norm, $\norm{M}_1$ to denote its trace norm, and $\norm{M}_F$ to denote its Frobenius norm.

For a string $\bx = (x_1, \ldots, x_n)$, we let $\bx_{\sim i}$ and $\bx_{\sim i,j}$ denote the string with the $i$-th index removed and the string with the $i$-th and $j$-th indices removed. For any set $S \subseteq [n]$, we let $\bx_S$ denote the string restricted to the entries in $S$.

We will work with the following random matrix ensembles:
\begin{definition}[Trace-centered Gaussian orthogonal ensemble (GOE)]\label{def:goe}
    For $d\in \bN$, let $G\sim \GOE(d)$, that is, $G\in\R^{d\times d}$ is symmetric with upper diagonal entries sampled independently from $\calN(0,1/d)$ and diagonal entries sampled independently from $\calN(0,2/d)$.
    
    Define $M = G - \fr{\Tr (G)}{d}I_d$. We say that $M$ is a \emph{trace-centered GOE matrix} and denote its distribution $\sGOE(d)$. 
    For $U\subseteq \bR^{d\times d}$, $\oM$ is a \emph{$U$-truncated trace-centered GOE matrix} if it is drawn from $\sGOE(d)$ conditioned on $\oM \in U$. 
    We denote the distribution of $\oM$ by $\sGOE_U(d)$.
\end{definition}

\begin{definition}[Truncated Ginibre]\label{def:ginibre}
    For $d_1,d_2\in \bN$, let $G\sim \Gin(d_1,d_2)$ be the (normalized) \emph{$d_1\times d_2$ Ginibre matrix}, that is, $G\in \R^{d_1\times d_2}$ has i.i.d. entries $\calN(0, 1/d_1)$. 
    For $U\subseteq \R^{d_1\times d_2}$, $\oG$ is a \emph{$U$-truncated $d_1\times d_2$ Ginibre matrix} if it is drawn from $\Gin(d_1,d_2)$ conditioned on $\oG \in U$.
    We denote the distribution of $\oG$ by $\Gin_U(d)$.
\end{definition}

\noindent Our result for state certification uses the following notion of fidelity.
\begin{definition}[Fidelity between two quantum states]
    The fidelity of quantum states $\rho, \sigma \in \C^{d\times d}$ is $F(\rho,\sigma) = (\Tr \sqrt{\rho^{1/2} \sigma \rho^{1/2}})^2$.
\end{definition}

\noindent Our lower bounds are based on Le Cam's two-point method which we briefly review here. The following is an elementary result in binary hypothesis testing:

\begin{fact}[See e.g. Theorem 4.3 from \cite{wu2017lecture}]\label{fact:tv}
    Given distributions $p_0,p_1$ over a domain $\mathcal{S}$, if $d_{\mathrm{TV}}(p_0,p_1) < 1/3$, there is no $\mathcal{A}:\mathcal{S}\to\brc{0,1}$ for which $\Pr[x\sim p_i]{\mathcal{A}(x) = i} \ge 2/3$ for both $i = 0,1$.
\end{fact}

\noindent Now consider a state distinguishing task of the form
\begin{equation}
    H_0: \rho = \sigma \qquad \text{and} \qquad H_1: \rho = \sigma_M,
\end{equation}
where $\sigma_M$ is a random state sampled from some distribution $\calD$ over the set of states satisfying $\norm{\sigma - \sigma_M}_1 > \epsilon$. Recall from Definition~\ref{def:tree} that a learning algorithm that uses $n$ incoherent measurements corresponds to a tree $\calT$ of depth $n$, and $\rho = \sigma$ and $\rho = \sigma_M$ induce distributions $p_0$ and $p_M$ on the leaves of this tree. We can use Fact~\ref{fact:tv} to reduce proving a copy complexity lower bound for state certification with respect to $\sigma$, which is a worst-case guarantee over all possible input states $\rho$, to bounding $d_{\mathrm{TV}}(p_0, \E[M]{p_M})$, which is an average-case bound.

\begin{lemma}[Le Cam's two-point method, see e.g. Lemma 1 in \cite{yu1997assouad}]
    If there is a distribution $\calD$ over states satisfying $\norm{\sigma - \sigma_M}_1 > \epsilon$ for which $d_{\mathrm{TV}}(p_0, \E[M]{p_M}) \le 1/3$ for any tree $\calT$ of depth $n$, then any algorithm $\calA$ using incoherent measurements for state certification with respect to $\sigma$ must make more than $n$ incoherent measurements to achieve success probability at least $2/3$.
\end{lemma}

\begin{proof}
    Suppose to the contrary there existed such an algorithm $\calA$ using at most $n$ incoherent measurements, and let $p_0$ and $p_M$ denote the distributions over the leaves of the tree corresponding to $\calA$ when $\rho = \sigma$ and $\rho = \sigma_M$ respectively. Suppose when it succeeds, $\calA$ outputs 0 when $\rho = \sigma$ and 1 when $\norm{\rho - \sigma}_1 > \epsilon$. Let $p_1 \triangleq \E[M\sim\calD]{p_M}$. Because $\calA$ successfully outputs 1 with probability 2/3 when given as input the state $\sigma_M$ for any $M$, $2/3 \le \E[M]{\Pr[\bx\sim p_M]{\calA(\bx) = 1}} = \E[\bx\sim p_1]{\calA(\bx) = 1}$. Similarly, $2/3 \le \E[\bx\sim p_0]{\calA(\bx) = 0}$. By Fact~\ref{fact:tv}, this would contradict the bound on $d_{\mathrm{TV}}(p_0, p_1)$.
\end{proof}

\section{Lower Bound for Mixedness Testing}
\label{sec:paninski}

In this section we prove the following theorem, which is the formal version of Theorem~\ref{thm:informal-mixedness-testing}.
\begin{theorem}\label{cor:mixedness-testing}
    Let $d \gg 1$ and $0 < \eps \le 1/12$.
    Any algorithm that uses incoherent measurements which, given $n$ copies of a mixed state $\rho \in \C^{d \times d}$, can distinguish between the case where $\rho = \rho_\mm$ and where $\| \rho - \rho_\mm \|_1 > \eps$ with probability at least $2/3$, must use at least $n = \Omega(d^{3/2}/\eps^2)$ copies.
\end{theorem}
\noindent
By the upper bound in \cite{bubeck2020entanglement}, this is tight up to constant factors. Also note that by standard amplification arguments, the choice of constant in the success probability is arbitrary, and can be taken to be any constant which is strictly larger than $1/2$. 

In fact, we will prove a slightly stronger theorem, which will be useful later on for our lower bounds against state certification.
Namely, we will show that the same bound holds not just when the null hypothesis is the maximally mixed state, but for any state whose smallest and largest eigenvalues are comparable.

More formally, let $A \in \bR^{d\times d}$ be a diagonal matrix with diagonal entries $a_1\ge \cdots \ge a_d > 0$, satisfying $2a_d \ge a_1$, and $\Tr(A) = d$.
We consider the task of distinguishing between the following two alternatives:
\begin{equation}
    \label{eq:paninski}
    H_0: \rho = \frac{1}{d} A
    \qquad 
    \text{and} 
    \qquad 
    H_1: \rho = \frac{1}{d} (A + \eps \oM).
\end{equation}
Here, $\overline{M} \sim \sGOE_U(d)$ for the $U$ given by Lemma~\ref{lem:goe-trunc} below.

\begin{restatable}{lemma}{goetrunc}
    \label{lem:goe-trunc}
    There exists $U\subseteq \bR^{d\times d}$ such that if $M\sim \sGOE(d)$, then $\Pr{M\not\in U} \le \exp(-\Omega(d))$ and on the event $M\in U$, we have $\norm{M}_{\op} \le 3$ and $\norm{M}_1 \ge d/12$.
\end{restatable}

\noindent
We defer the proof of this lemma to Appendix~\ref{sec:basic-comps}.
Our main result for the distinguishing task \eqref{eq:paninski} is the following.

\begin{theorem}\label{thm:main_standard}
    If $d \gg 1$ and $\eps \le 1/12$, then any algorithm using incoherent measurements that distinguishes between $H_0$ and $H_1$ with success probability at least $2/3$ requires $n = \Omega(d^{3/2}/\eps^2)$ copies.
\end{theorem}

\noindent 
Again, by standard amplification arguments, the choice of constant in the success probability is arbitrary, and can be taken to be any constant greater than $1/2$.
Note that the bounds in Lemma~\ref{lem:goe-trunc} ensure that under $H_1$, $\rho$ is psd (and thus a valid quantum state) and has trace distance $\Omega(\eps)$ to $\fr{1}{d}A$. 
In particular, since any algorithm for mixedness testing must solve this distinguishing problem as well, setting $A=I_d$ into Theorem~\ref{thm:main_standard} immediately implies Theorem~\ref{cor:mixedness-testing}.

Take any learning tree $\calT$ corresponding to an algorithm for this task that uses $n$ incoherent measurements. 
Recalling the terminology from Definition~\ref{def:tree}, we let $p_0$ and $p_1$ denote the distributions over leaves of $\calT$ induced by $\rho$ under $H_0$ and $H_1$ respectively. 
In the rest of this section, we assume $n \ll d^{3/2}/\eps^2$ and will prove $\TV(p_0,p_1) = o(1)$. 
It is clear that this immediately implies Theorem~\ref{thm:main_standard}.

We let $L^*(\cdot)$ denote the likelihood ratio between $p_1$ and $p_0$. 
That is, for a sequence of vectors $\bx = (x_1,\ldots,x_n)$, let $L^*(\bx) \triangleq p_1(\bx)/p_0(\bx)$. 
Note that
\begin{equation}
    \label{eq:def-Lstar}
    L^*(\bx) 
    = 
    \E[\oM\sim \sGOE_U(d)]*{
        \prod^n_{i=1} \lt(1 + \eps \fr{x^{\dagger}_i \oM x_i}{x^{\dagger}_i A x_i}
        \rt)
    }. 
\end{equation}
Define similarly
\begin{equation}
    \label{eq:def-L}
    L(\bx) 
    \triangleq
    \E[M\sim \sGOE(d)]*{
        \prod^n_{i=1} \lt(1 + \eps \fr{x^{\dagger}_i M x_i}{x^{\dagger}_i A x_i}
        \rt)
    }. 
\end{equation}
This is an estimate for the likelihood ratio $L^*(\bx)$ where the conditioned Gaussian integral is replaced by a true Gaussian integral.
Most of the computations in this section will be done in terms of $L(\bx)$; the proof of Theorem~\ref{thm:main_standard} below quantifies that $L(\bx)$ is a close approximation of $L^*(\bx)$.

Throughout this section, we will somewhat abuse notation and write $L(\bz)$ for any sequence of unit vectors $\bz = (z_1,\ldots,z_t)$ of length not necessarily $n$.
This is defined the same way as in \eqref{eq:def-L}.
We also write $L(\bx,\bx)$ to denote the value of $L$ on input $(x_1,x_1,x_2,x_2,\ldots,x_n,x_n)$.
% \bh{flavortext here: this is a ``fake likelihood ratio," $L$ makes sense as a function on any vector sequence, we will slightly abuse notation and write $L(\bx)$ for vector sequence $\bx = (x_1,\ldots,x_t)$ of length not necessarily $n$}

The main ingredient in the proof of Theorem~\ref{thm:main_standard} is the following high-probability bound on $L$ evaluated at the leaves of $\calT$.

\begin{proposition}
    \label{prop:lstar-reduction}
    There exists a subset $S$ of the leaves of $\calT$ such that $\Pr[p_0]*{S} = 1-o(1)$ and for all $\bx \in S$, $|L(\bx)-1|=o(1)$ and $L(\bx,\bx) \ll e^{\sqrt{d}}$.
\end{proposition}

\noindent Let us first prove Theorem~\ref{thm:main_standard} assuming Proposition~\ref{prop:lstar-reduction}. 

\begin{proof}[Proof of Theorem~\ref{thm:main_standard}]
    Let $U$ be as in Lemma~\ref{lem:goe-trunc}.
    Define
    \[
        \oL(\bx)
        = 
        \E[M\sim \sGOE(d)]*{
            \Id\{M\in U\}
            \prod_{i=1}^n
            \lt(
                1 + \eps \fr{x_i^\dagger M x_i}{x_i^\dagger A x_i}
            \rt)
        }.
    \]
    It is clear that $L^*(\bx) = \Pr*{U}^{-1} \oL(\bx)$.
    For all $\bx \in S$, by Cauchy-Schwarz
    \begin{align*}
        |L(\bx) - \oL(\bx)|
        &= 
        \lt|
            \E[M\sim \sGOE(d)]*{
                \Id\{M\not\in U\}
                \prod_{i=1}^n
                \lt(
                    1 + \eps \fr{x_i^\dagger M x_i}{x_i^\dagger A x_i}
                \rt)
            }
        \rt| \\
        &\le 
        \Pr*{U^c}^{1/2}
        \E[M\sim \sGOE(d)]*{
            \prod_{i=1}^n
            \lt(
                1 + \eps \fr{x_i^\dagger M x_i}{x_i^\dagger A x_i}
            \rt)^2
        }^{1/2} \\
        &= 
        \sqrt{\Pr*{U^c} L(\bx,\bx)}
        =o(1).
    \end{align*}
    Here we use that $\Pr*{U^c} \le \exp(-\Omega(d))$ and $L(\bx,\bx) \ll e^{\sqrt{d}}$.
    Moreover, we have $|L(\bx)-1| = o(1)$.
    Thus, for all $\bx \in S$, $\oL(\bx) = 1+o(1)$ and 
    \begin{align*}
        |L^*(\bx)-1|
        &\le 
        |L^*(\bx)-\oL(\bx)|
        +
        |\oL(\bx) - 1| \\
        &= 
        \fr{\Pr*{U^c}}{\Pr*{U}}
        \oL(\bx) + o(1) 
        = o(1).
    \end{align*}
    Finally, 
    \begin{align*}
        \TV(p_0,p_1)
        &= 
        2\E[\bx\sim p_0]*{
            (L^*(\bx)-1)_-
        } \\
        &= 
        2\E[\bx\sim p_0]*{
            \Id\{\bx\in S\} (L^*(\bx)-1)_-
        } +
        2\E[\bx\sim p_0]*{
            \Id\{\bx\not\in S\} (L^*(\bx)-1)_-
        } \\
        &\le 
        2\sup_{\bx \in S} (L^*(\bx)-1)_-
        +
        2\Pr[p_0]*{S^c} 
        = 
        o(1). \qedhere
    \end{align*}
\end{proof}

\subsection{Recursive evaluation of likelihood ratio}

Let $\bz = (z_1,\ldots,z_t)$ be a sequence of unit vectors. 
% \sitan{thoughts on maybe changing $\underline{y}$ to $\mathbf{y}$? might make things like $\overline{K}(\underline{y})$ look less weird}
% \bh{yeah sounds good, will do}
For $1\le i\le t$, let $\bz_{\sim i}$ be the sequence $\bz$ with $z_i$ omitted.
Similarly, for $1\le i<j\le t$, let $\bz_{\sim i,j}$ be the sequence $\bz$ with $z_i,z_j$ omitted.
The main result of this subsection is the following recursive formula for $L(\bz)$. 
\begin{lemma}
    \label{lem:L-recursion}
    The function $L$ satisfies
    \[
        L(\bz) = L(\bz_{\sim t}) + \fr{2\eps^2}{d^2} 
        \sum_{i=1}^{t-1} 
        \lt[
            \fr{d\la z_i,z_t\ra^2-1}{(z_i^\dagger A z_i)(z_t^\dagger A z_t)}
            L(\bz_{\sim i,t})
        \rt].
    \]
\end{lemma}

\noindent The proof is based on Isserlis' theorem, which we record below.
For $k$ even, let $\PMat(k)$ denote the set of perfect matchings of $\{1,\ldots,k\}$.
\begin{theorem}[\cite{isserlis1918moment}]
    \label{thm:isserlis}
    Let $g=(g_1,\ldots,g_k)$ be a jointly Gaussian vector. 
    If $k$ is odd, then $\E{\prod_{i=1}^k g_i} = 0$.
    If $k$ is even, then 
    \[
        \E*{\prod_{i=1}^k g_i}
        =
        \sum_{\{\{a_1,b_1\},\ldots,\{a_{k/2},b_{k/2}\}\} \in \PMat(k)}
        \prod_{i=1}^{k/2} \E*{g_{a_i}g_{b_i}}.
    \]
\end{theorem}

\begin{proof}[Proof of Lemma~\ref{lem:L-recursion}]
    For a set $S\subseteq [t]$ with $|S|$ even, let $\PMat(S)$ denote the set of perfect matchings of $S$. 
    For even $k\le t$, let $\Mat(t,k)$ denote the set of matchings of $[t]$ consisting of $k/2$ pairs.
    We compute that
    \begin{align}
        \notag
        L(\bz)
        &= 
        \sum_{S\subseteq [t]}
        \eps^{|S|}
        \E[M\sim \sGOE(d)]*{\prod_{i\in S}\fr{z_i^\dagger M z_i}{z_i^\dagger A z_i}} 
        \qquad \text{(expanding \eqref{eq:def-L})} \\
        \notag
        &= 
        \sum_{\substack{S\subseteq [t] \\ |S|~\text{even}}}
        \eps^{|S|}
        \sum_{\{\{a_1,b_1\},\ldots,\{a_{|S|/2},b_{|S|/2}\}\} \in \PMat(S)}
        \prod_{i=1}^{|S|/2} \E[M\sim \sGOE(d)]*{
            \fr{z_{a_i}^\dagger M z_{a_i}}{z_{a_i}^\dagger A z_{a_i}} \cdot 
            \fr{z_{b_i}^\dagger M z_{b_i}}{z_{b_i}^\dagger A z_{b_i}}
        } 
        \qquad \text{(Th.~\ref{thm:isserlis})} \\
        \notag
        &= 
        \sum_{k=0}^{\lfloor t/2\rfloor}
        \eps^{2k}
        \sum_{\{\{a_1,b_1\},\ldots,\{a_k,b_k\}\} \in \Mat(t,2k)}
        \prod_{i=1}^{k} \E[M\sim \sGOE(d)]*{
            \fr{z_{a_i}^\dagger M z_{a_i}}{z_{a_i}^\dagger A z_{a_i}} \cdot 
            \fr{z_{b_i}^\dagger M z_{b_i}}{z_{b_i}^\dagger A z_{b_i}}
        } \\
        \label{eq:L-expansion}
        &= 
        \sum_{k=0}^{\lfloor t/2\rfloor}
        \lt(\fr{2\eps^2}{d^2}\rt)^{k}
        \sum_{\{\{a_1,b_1\},\ldots,\{a_{k},b_{k}\}\} \in \Mat(t,2k)}
        \prod_{i=1}^{k} 
        \fr{d\la z_{a_i}, z_{b_i}\ra^2-1}{(z_{a_i}^\dagger A z_{a_i})(z_{b_i}^\dagger A z_{b_i})}.
    \end{align}
    In the final step we use that for unit vectors $x,y \in \C^d$, 
    \[
        \E[M\sim \sGOE(d)]*{
            (x^\dagger M x)(y^\dagger M y)
        }
        = 
        \fr{2}{d^2}(d\la x,y\ra^2-1),
    \]
    which can be verified by direct computation.
    The lemma follows by partitioning the summands in \eqref{eq:L-expansion} based on whether $t$ appears in the matching, and if so which $i\in \{1,\ldots,t-1\}$ it is paired with.
\end{proof}

\subsection{High probability bound on likelihood ratio at leaves}
% \bh{remaining subsection titles TBD}

This subsection gives the main part of the proof of Proposition~\ref{prop:lstar-reduction}.
For any sequence of unit vectors $\bz = (z_1,\ldots,z_t)$, define
\[
    H(\bz) 
    = 
    \sum_{i=1}^{t} 
    \fr{dz_iz_i^\dagger - I_d}{z_i^\dagger A z_i} \cdot 
    \fr{L(\bz_{\sim i})}{L(\bz)}
    \qquad
    \text{and}
    \qquad
    K(\bz) 
    = 
    \sum_{i=1}^{t} 
    \fr{dz_iz_i^\dagger - I_d}{z_i^\dagger A z_i}.
\]
The function $H$ enters our calculations by the following rewriting of Lemma~\ref{lem:L-recursion}:
\begin{equation}
    \label{eq:L-ratio-recursion}
    \fr{L(\bz)}{L(\bz_{\sim t})}
    = 
    1 + 
    \fr{2\eps^2}{d^2}
    \cdot 
    \fr{z_t^\dagger H(\bz_{\sim t}) z_t}{z_t^\dagger A z_t}.
\end{equation}
If $\bz = \bx_{\le t} \triangleq (x_1,\ldots,x_t)$ is a prefix of $\bx \sim p_0$, then $\fr{L(\bz)}{L(\bz_{\sim t})} = \fr{L(\bx_{\le t})}{L(\bx_{\le t-1})}$ is one step in the likelihood ratio martingale.
As we will see (proof of Claim~\ref{cl:lx-bd}) below, the multiplicative fluctuation of this step is 
\[
    \E[x_t]*{\lt(\fr{L(\bx_{\le t})}{L(\bx_{\le t-1})}\rt)^2} 
    = 
    1 + O\lt(\fr{\eps^4}{d^5}\rt) \norm{H(\bx_{\le t-1})}_F^2.
\]

Thus, an upper bound on $\norm{H(\bz)}_F$ over all prefixes $\bz$ of $\bx$ controls the fluctuations of the likelihood ratio martingale.
Because the matrices output by $H$ are hard to control directly, we will use the function $K$ as a proxy for $H$. 
The following lemma quantifies this relationship, showing that if $K(\bz)$ is bounded in Frobenius norm, $H(\bz)$ is bounded at the same scale.
% The function $K$ will enter our calculations as a proxy for $H$, in a sense made precise by the following lemma.
\begin{lemma}
    \label{lem:bootstrap}
    % Suppose $\gamma \gg 1$. %$1 \ll \gamma \ll \sqrt{n/d}$.
    Suppose $1 \ll \gamma \ll d / (\eps^2 n^{1/2})$.
    If $\bz=(z_1,\ldots,z_t)$ is a sequence of unit vectors satisfying $t \le n$ and $\norm{K(\bz)}_F \le n^{1/2} d \gamma$, then $\norm{H(\bz)}_F \le 3 n^{1/2} d \gamma$.
\end{lemma}
\noindent Note that this lemma is a ``deterministic" statement about a sequence of vectors.
We will prove this in Subsection~\ref{sec:bound-signed-sum} using the bootstrap argument alluded to earlier.
The following lemma bounds $K(\bz)$ in Frobenius norm uniformly over all prefixes $\bz$ of $\bx$. 
We will prove this lemma in Subsection~\ref{sec:doob} by mimicking the proof of Doob's $L^2$ maximal inequality for the matrix valued martingale $K(\bx_{\le t})$.

\begin{lemma}
    \label{lem:doob}
    If $\bx \sim p_0$, then $\E*{\sup_{1\le t\le n} \norm{K(\bx_{\le t})}_F^2} \lesssim nd^2$.
\end{lemma}

% \bh{oops, it's a little unfortunate that $M$ means martingale and matrix... can fix this later. maybe bold M for matrix?}

\noindent We will now prove Proposition~\ref{prop:lstar-reduction} assuming Lemmas~\ref{lem:bootstrap} and \ref{lem:doob}.
We set $\alpha, \beta$ to be slowly-growing functions such that $1 \ll \alpha \ll \beta \ll d^{3/2}/(\eps^2n) \wedge d/(\eps^2 n^{1/2})$, % \wedge \sqrt{n/d}
and furthermore $\alpha^2 \ll d^{3/2}/(\eps^2n)$.
% \sitan{$\alpha^2$ and $\beta^2$ should be $\alpha$ and $\beta$ here, right?}
% This is possible because $d\ll n \ll d^{3/2}/\eps^2$.
This is possible because $n \ll d^{3/2}/\eps^2$.

Let $\bx \sim p_0$.
% For $1\le t\le n$, define $H_t = H(\bx_{\le t})$, $K_t = K(\bx_{\le t})$.
% Further define the filtration  $\cF_t = \sigma(\bx_{\le t})$ and the sequence $\Phi_t = L(\bx_{\le t})$.
For $1\le t\le n$, define the filtration $\cF_t = \sigma(\bx_{\le t})$ and the sequences
\begin{equation}
    H_t = H(\bx_{\le t}), \qquad K_t = K(\bx_{\le t}), \qquad \Phi_t = L(\bx_{\le t}).
\end{equation}
Consider the time
\[
    \tau = \inf \lt\{
        t : \norm{K_t}_F > n^{1/2} d\alpha
        ~\text{or}~
        |\Phi_t-1| > \fr{\eps^2 n}{d^{3/2}} \beta
    \rt\} \cup \{\infty\},
\]
which is clearly a stopping time with respect to $\cF_t$.
Also define the stopped sequence $\Psi_t = \Phi_{t\wedge \tau}$.

\begin{claim}
    \label{cl:frob-sup}
    With probability $1-o(1)$, $\norm{K_t}_F \le n^{1/2} d\alpha$ for all $1\le t\le n$.
\end{claim}
\begin{proof}
    By Lemma~\ref{lem:doob}, 
    \begin{align}
        \Pr*{\sup_{1\le t\le n} \norm{K_t}_F > n^{1/2}d\alpha}
        &\le 
        \fr{\E*{\sup_{1\le t\le n} \norm{K_t}_F^2}}{nd^2\alpha^2}
        \lesssim \alpha^{-2} = o(1). \qedhere
    \end{align}
\end{proof}

\begin{claim}
    \label{cl:lx-bd}
    With probability $1-o(1)$, $|\Psi_n-1| \le \fr{\eps^2 n}{d^{3/2}}\beta$.
\end{claim}
\begin{proof}
    Note that $\Psi_t$ is a multiplicative martingale: if $\tau \le t-1$ then certainly $\E{\fr{\Psi_t}{\Psi_{t-1}} | \cF_{t-1}}=1$, and if $\tau > t-1$, \eqref{eq:L-ratio-recursion} implies 
    \[
        \E*{\fr{\Psi_t}{\Psi_{t-1}} | \cF_{t-1}}
        = 
        1 + 
        \fr{2\eps^2}{d^2} \E*{\fr{x_t^\dagger H_{t-1} x_t}{x_t^\dagger A x_t} | \cF_{t-1}}
        = 1,
    \]
    using that 
    \begin{equation}
        \label{eq:disc-deg1}
        \E*{\fr{x_t^\dagger H_{t-1} x_t}{x_t^\dagger A x_t} | \cF_{t-1}}
        = 
        \sum_{x_t} \omega_{x_t} (x_t^\dagger H_{t-1} x_t)
        = 
        \lt\la H_{t-1}, \sum_{x_t} \omega_{x_t} x_tx_t^\dagger \rt\ra 
        = 
        \la H_{t-1}, I_d/d\ra
        = 0.
    \end{equation}
    We next bound the quadratic increment $\E{(\fr{\Psi_t}{\Psi_{t-1}})^2 | \cF_{t-1}}$.
    If $\tau \le t-1$ this is $1$, and otherwise
    \begin{equation}
        \label{eq:N-ratio-sq}
        \E*{\lt(\fr{\Psi_t}{\Psi_{t-1}}\rt)^2 | \cF_{t-1}}
        = 
        1 + 
        \fr{4\eps^2}{d^2} \E*{\fr{x_t^\dagger H_{t-1} x_t}{x_t^\dagger A x_t} | \cF_{t-1}}
        +
        \fr{4\eps^4}{d^4}
        \E*{\fr{(x_t^\dagger H_{t-1} x_t)^2}{(x_t^\dagger A x_t)^2} | \cF_{t-1}}.
    \end{equation}
    The first expectation is zero by \eqref{eq:disc-deg1}.
    To bound the remaining expectation, note that for any unit vector $x$,
    \begin{equation}
        \label{eq:xAx-lb}
        x^\dagger A x \ge a_d x^\dagger x = a_d \ge \fr12.
    \end{equation}
    So, 
    \begin{align}
        \MoveEqLeft
        \E*{\fr{(x_t^\dagger H_{t-1} x_t)^2}{(x_t^\dagger A x_t)^2} | \cF_{t-1}}
        \le
        2\E*{\fr{(x_t^\dagger H_{t-1} x_t)^2}{x_t^\dagger A x_t} | \cF_{t-1}} 
        = 
        2\sum_{x_t}
        \omega_{x_t}
        x_t^\dagger H_{t-1} (x_t x_t^\dagger) H_{t-1} x_t 
        \\ &\le 
        2\sum_{x_t}
        \omega_{x_t}
        x_t^\dagger H_{t-1}^2 x_t 
        % \label{eq:disc-deg2}
        = 
        2 \lt\la H_{t-1}^2, \sum_{x_t} \omega_{x_t} x_tx_t^\dagger \rt\ra
        = 
        2\la H_{t-1}^2, I_d/d\ra 
        = 
        \fr{2}{d}\norm{H_{t-1}}_F^2.
    \end{align}
    Moreover, since $\tau > t-1$, $\norm{K_{t-1}}_F \le n^{1/2}d\alpha$ and Lemma~\ref{lem:bootstrap} implies $\norm{H_{t-1}}_F \le 3n^{1/2}d\alpha$.  
    Thus,
    \[
        \E*{\lt(\fr{\Psi_t}{\Psi_{t-1}}\rt)^2 | \cF_{t-1}}
        \le 
        1 + \fr{8\eps^4}{d^5} \norm{H_{t-1}}_F^2
        \le 
        1 + \fr{72\eps^4n}{d^3}\alpha^2.
    \]
    So, for all $1\le t\le n$,
    \[
        \E{\Psi_t^2}
        =
        \E*{\E*{\lt(\fr{\Psi_t}{\Psi_{t-1}}\rt)^2 | \cF_{t-1}} \Psi_{t-1}^2}
        \le 
        \lt(1 + \fr{72\eps^4n}{d^3}\alpha^2\rt)
        \E{\Psi_{t-1}^2},
    \]
    and therefore
    \[
        \E{\Psi_t^2}
        \le 
        \lt(1 + \fr{72\eps^4n}{d^3}\alpha^2\rt)^n
        \le 
        \exp\lt(\fr{72\eps^4n^2}{d^3}\alpha^2\rt)
        \le 2
    \]
    since $\fr{\eps^4n^2}{d^3}\alpha^2 \ll 1$.
    Moreover, 
    \begin{align*}
        \E{(\Psi_t-1)^2}
        &=
        \E*{\E*{\lt(\fr{\Psi_t}{\Psi_{t-1}}\rt)^2 | \cF_{t-1}} \Psi_{t-1}^2 - 2 \E*{\fr{\Psi_t}{\Psi_{t-1}} | \cF_{t-1}}\Psi_{t-1} + 1} \\
        &\le 
        \fr{72\eps^4n}{d^3}\alpha^2 \E{\Psi_{t-1}^2} + \E{(\Psi_{t-1}-1)^2} \\
        &\le 
        \fr{144\eps^4n}{d^3}\alpha^2 + \E{(\Psi_{t-1}-1)^2},
    \end{align*}
    so by induction
    \[
        \E{(\Psi_n-1)^2}
        \le 
        \fr{144\eps^4n^2}{d^3}\alpha^2.
    \]
    Thus
    \[
        \Pr*{|\Psi_n-1| > \fr{\eps^2n}{d^{3/2}}\beta}
        \le 
        \fr{\E*{|\Psi_n-1|^2}}{\fr{\eps^4n^2}{d^3}\beta^2}
        \le 
        \fr{144\alpha^2}{\beta^2}
        =o(1).
    \]
    Therefore, $|\Psi_n-1|\le \fr{\eps^2n}{d^{3/2}}\beta$ with probability $1-o(1)$.
\end{proof}

% \bh{todo rewrite this with allen and sitan's simpler argument} \sitan{took a pass at writing this up} \bh{looks good, thanks! much cleaner than what I was doing before lol}
\begin{claim}
    \label{cl:lxx-bd}
    If $\norm{K_n}_F \le n^{1/2}d\alpha$, then $L(\bx,\bx) \ll e^{\sqrt{d}}$.
\end{claim}
\begin{proof}
Using the elementary inequality $e^{2z + z^2}\geq (1 + z)^2$ and then Cauchy Schwarz, we can write
\begin{align}
L(\bx,\bx) &=  \E[M\sim\sGOE(d)]*{ \prod_{i = 1}^n \left(1 + \eps \frac{x_i^\dagger M x_i}{x_i^\dagger A x_i} \right)^2 } \\ &\leq  \E[M\sim\sGOE(d)]*{\exp\left(\sum^n_{i=1} 2\eps \frac{x_i^\dagger M x_i}{x_i^\dagger A x_i} + \left( \eps \frac{x_i^\dagger M x_i}{x_i^\dagger A x_i} \right)^2 \right)} 
\\ & \leq \sqrt{ \E[M\sim\sGOE(d)]*{\exp\left(4\sum^n_{i=1} \eps \frac{x_i^\dagger M x_i}{x_i^\dagger A x_i}\right)} \E[M\sim\sGOE(d)]*{\exp\left(2\sum^n_{i=1}  \left( \eps \frac{x_i^\dagger M x_i}{x_i^\dagger A x_i} \right)^2\right)}} \,.\label{eq:product1}
\end{align}
Now we bound each of the terms in~\eqref{eq:product1}.  For the first term, we have
\begin{equation}
    \E[M\sim\sGOE(d)]*{\exp\left(4\sum^n_{i=1} \eps \frac{x_i^\dagger M x_i}{x_i^\dagger A x_i}\right)}   = \E[M\sim\sGOE(d)]*{\exp\left(\frac{4\eps}{d}\iprod*{M, K_n}\right)}, \label{eq:lux2}
\end{equation}
    where we used that $\Tr(M) = 0$. As $M = G - \frac{\Tr(G)}{d}I_d$ for $G\sim\GOE(d)$, we have that $\iprod{M,K_n} = \iprod{G,K_n}$ is distributed as a Gaussian with variance at most $\frac{2}{d}\norm{K_n}^2_F \le 2nd\alpha^2$. So we can bound \eqref{eq:lux2} by
    \begin{equation}
        \E[g\sim\cN(0,32\eps^2 n\alpha^2 / d)]{\exp(g)} = e^{16\eps^2 n\alpha^2 / d} 
        \ll e^{\sqrt{d}}
    \end{equation}
    as $\alpha^2 \ll d^{3/2}/(\eps^2n)$ by assumption.  Next we bound the second term in the product in~\eqref{eq:product1}.  Use $\textsf{vec}(M)$ to denote rearranging $M$ as a vector in $\R^{d^2}$ (done in a consistent way) and use $\otimes $ to denote the Kronecker product of two matrices.  Let $Q \in \R^{d^2 \times d^2}$ be defined as
    \[
    Q = \sum_{i = 1}^n \frac{x_i x_i^\dagger \otimes x_i x_i^\dagger }{(x_i^\dagger A x_i)^2} \,.
    \]
    We have
    \begin{align*}
     \E[M\sim\sGOE(d)]*{\exp\left(2\sum^n_{i=1}  \left( \eps \frac{x_i^\dagger M x_i}{x_i^\dagger A x_i} \right)^2\right)} = \E[M\sim\sGOE(d)]*{ \exp\left(2 \eps^2  \textsf{vec}(M)^\dagger Q \textsf{vec}(M)\right)} \,.
    \end{align*}
    Now note that 
    \[
    Q \preceq 2\sum_{i = 1}^n \frac{x_i x_i^\dagger \otimes I_d}{x_i^\dagger A x_i} \preceq \left(\frac{2K_n}{d} + \frac{4nI_d}{d}\right) \otimes I_d 
    \]
    so we have $\norm{Q}_{\op} \leq 6n/d$.  Also, we have
    \[
    \norm{Q}_1 \leq 2\sum_{i = 1}^n \frac{x_i^\dagger x_i}{x_i^\dagger A x_i} \leq 4n %8\left( \frac{\norm{K_n}_1}{d} + n\right) \leq  8\left( \frac{\norm{K_n}_F}{\sqrt{d}} + n\right) \leq 8n 
    \]
    % can't we just bound the trace norm of each summand by $2$ because $\norm{x_ix_i^\dagger}_1 = 1$ and $x^\dagger_i A x_i \ge 1/2$?}
    % where $\norm{Q}_1$ denotes the Schatten-$1$-norm.  
    Let $\lambda_1, \dots , \lambda_{d^2}$ be the eigenvalues of $Q$.  Next note that we have
    \begin{align*}
    \E[M\sim\sGOE(d)]*{ \exp\left(2 \eps^2  \textsf{vec}(M)^\dagger Q \textsf{vec}(M)\right)} &\leq \E[v \sim N(0, I_{d^2})]*{ \exp\left(10 \eps^2 /d \cdot v^\dagger Q  v\right)}    \\ &\leq \prod_{j = 1}^{d^2} \E[g \sim N(0, 1)]*{\exp(10\eps^2/d \cdot g^2\lambda_j ) }  \\ &\leq e^{20\eps^2(\lambda_1 + \dots + \lambda_{d^2}) /d} \ll e^{\sqrt{d}} \,.
    \end{align*}
    In the first step above, we used the convexity of the function inside the expectation to replace the distribution over $M \sim \sGOE(d)$ with another distribution that can be obtained by adding independent, mean-$0$ noise to $M$.  Afterwards, we used the rotational invariance of $N(0, I_{d^2})$ and then the bound on $\norm{Q}_\op$ (together with the fact that $\E[g]{e^{cx^2}} = (1 - 2c)^{-1} \le e^{2c}$ for sufficiently small $c$), and finally the bound on $\norm{Q}_1$.  Putting everything together, we conclude that $L(\bx,\bx)  \ll e^{\sqrt{d}}$ as desired.
\end{proof}

\begin{proof}[Proof of Proposition~\ref{prop:lstar-reduction}]
    Define the event
    \[
        S = \lt\{
            \sup_{1\le t\le n} \norm{K_t}_F \le n^{1/2}d\alpha
            ~\text{and}~
            |\Psi_n - 1| \le \fr{\eps^2 n}{d^{3/2}}\beta
        \rt\}.
    \]
    By Claims~\ref{cl:frob-sup} and \ref{cl:lx-bd}, $\Pr[p_0]{S}=1-o(1)$.
    We will show that if $S$ holds, then $\tau=\infty$. 
    Indeed, if $\tau = t<\infty$, then either $\norm{K_t}_F > n^{1/2}d\alpha$ or $|\Phi_t-1|>\fr{\eps^2n}{d^{3/2}}\beta$ holds.
    Since $\Psi_n = \Phi_t$, this contradicts $S$.
    
    So, $\tau=\infty$ on $S$.
    This implies $|L(\bx)-1| = |\Phi_n-1| \le \fr{\eps^2n}{d^{3/2}}\beta = o(1)$. 
    Moreover $\norm{K_n}_F \le n^{1/2}d\alpha$, so by Claim~\ref{cl:lxx-bd} we have $L(\bx,\bx) \ll e^{\sqrt{d}}$.
\end{proof}

\subsection{Bounding \texorpdfstring{$H$}{H} in Frobenius norm by bootstrapping}\label{sec:bound-signed-sum}

In this subsection, we prove Lemma~\ref{lem:bootstrap}.
Throughout this subsection, let $\bz = (z_1,\ldots,z_t)$ be a sequence of unit vectors satisfying $t \le n$ and
\begin{equation}
    \norm{K(\bz)}_F \le n^{1/2} d \gamma
\end{equation} 
% for some $\gamma \gg 1$. %$1 \ll \gamma \ll \sqrt{n/d}$. 
for some $1 \ll \gamma \ll d/(\eps^2 n^{1/2})$.

The following lemma bounds a variant of $K(\bz)$ where we multiply each summand by an adversarial $b_i\in [-1,1]$.
This will be used to control the discrepancy $H(\bz)-K(\bz)$ in the bootstrapping argument.
\begin{lemma}
    \label{lem:uniform-frob-bd}
    Uniformly over $b_1,\ldots,b_t\in [-1,1]$, we have
    \[
        \norm{\sum_{i=1}^{t} b_i \fr{dz_iz_i^\dagger - I_d}{z_i^\dagger A z_i}}_F
        % \lesssim 
        % nd^{1/2}.
        \le 
        n^{1/2} d \gamma + 2 nd^{1/2}.
    \]
\end{lemma}
\begin{proof}
    For any choice of $b_1,\ldots,b_t$,
    \begin{align*}
        \norm{\sum_{i=1}^{t} b_i \fr{dz_iz_i^\dagger - I_d}{z_i^\dagger A z_i}}_F
        &\le 
        \norm{\sum_{i=1}^{t} b_i \fr{dz_iz_i^\dagger}{z_i^\dagger A z_i}}_F +
        \norm{\sum_{i=1}^{t} b_i \fr{I_d}{z_i^\dagger A z_i}}_F \\
        &\le 
        \norm{\sum_{i=1}^{t}  \fr{dz_iz_i^\dagger}{z_i^\dagger A z_i}}_F +
        \norm{\sum_{i=1}^{t}  \fr{I_d}{z_i^\dagger A z_i}}_F \\
        &\le 
        \norm{K(\bz)}_F + 2 \norm{\sum_{i=1}^{t}  \fr{I_d}{z_i^\dagger A z_i}}_F.
    \end{align*}
    The second inequality holds because the matrices $dz_iz_i^\dagger$ and $I_d$ are both psd.
    Using \eqref{eq:xAx-lb}, we have 
    \[
        \norm{\sum_{i=1}^{t}  \fr{I_d}{z_i^\dagger A z_i}}_F
        \le 
        2td^{1/2} 
        % \lesssim 
        % nd^{1/2}.
        \le 
        2nd^{1/2}.
    \]
    % Moreover, because $\gamma \ll \sqrt{n/d}$, $\norm{K(\bz)}_F \ll nd^{1/2}$.
    % The result follows.
    The result follows by the assumed bound on $\norm{K(\bz)}_F$.
\end{proof}

For $S\subseteq [t]$, let $\bz_S = (z_i)_{i\in S}$.
Further, let
\[
    H_{S} = \sum_{i\in S} 
    \fr{dz_iz_i^\dagger - I_d}{z_i^\dagger A z_i} \cdot \fr{L(\bz_{S\setminus \{i\}})}{L(\bz_{S})}
    \qquad
    \text{and}
    \qquad
    K_{S} = \sum_{i\in S} 
    \fr{dz_iz_i^\dagger - I_d}{z_i^\dagger A z_i}.
\]
The following lemma gives a preliminary bound on $\norm{H_S}_F$.
In the proof of Lemma~\ref{lem:bootstrap}, we will use this bound to control $\norm{H_S}_F$ for $|S| = t - O(\log n)$, followed by the bootstrap argument over $O(\log n)$ recursive rounds to contract the bound to $O(n^{1/2}d)$.
\begin{lemma}
    \label{lem:crude-frob-ub}
    % There exists an absolute constant $C$ such that for all $S\subseteq [t]$, $\norm{H_S}_F \le Cnd^{1/2}$.
    For all $S\subseteq [t]$, $\norm{H_S}_F \le 2n^{1/2}d\gamma + 4nd^{1/2}$.
\end{lemma}
\begin{proof}
    % We take $C$ to be twice the constant hidden by the $\lesssim$ in Lemma~\ref{lem:uniform-frob-bd}.
    Note that for any fixed $\oM \in U$, for the $U$ given by Lemma~\ref{lem:goe-trunc}, and any unit vector $z$,
    \[
        \eps \lt|\fr{z^\dagger \oM z}{z^\dagger A z}\rt|
        \le 
        \fr{1}{12} \cdot \fr{3}{1/2} = \fr12,
    \]
    so $1 + \eps \fr{z^\dagger \oM z}{z^\dagger A z} \in [1/2, 3/2]$.
    Thus, for all $i$, $L(\bz_S) / L(\bz_{S\setminus \{i\}}) \in [1/2, 3/2]$, which implies
    \begin{equation}
        \label{eq:ratio-crude}
        \fr{L(\bz_{S\setminus \{i\}})}{L(\bz_S)} \in [2/3,2].
    \end{equation}
    Lemma~\ref{lem:uniform-frob-bd} gives 
    \[
        % \fr12 \norm{H_S}_F \le \fr12 Cnd^{1/2},
        \fr12 \norm{H_S}_F \le n^{1/2} d \gamma + 2 nd^{1/2},
    \]
    as desired.
\end{proof}   
    
    % We will prove the lemma by induction on $|S|$. 
    % When $|S|=0$ the statement is trivial, establishing the base case.
    % For the inductive step, note that for all $i\in S$, equations \eqref{eq:L-ratio-recursion} and \eqref{eq:xAx-lb} imply
    % \begin{equation}
    %     \label{eq:ratio-dev-bd}
    %     \lt|\fr{L(\bz_{S})}{L(\bz_{S \setminus \{i\}})} - 1\rt|
    %     \le 
    %     \fr{2\eps^2}{d^2} \cdot \norm{\fr{H_{S\setminus i}}{z_i^\dagger A z_i}}_{\op}
    %     \le 
    %     \fr{4\eps^2}{d^2} \norm{H_{S\setminus i}}_F.
    % \end{equation}
    % By the inductive hypothesis and the fact that $n \ll d^{3/2}/\eps^2$, this upper bound is $o(1)$. 

\begin{proof}[Proof of Lemma~\ref{lem:bootstrap}]
    Let $D = \log \sqrt{n/d}$.
    If $t < D$, then by equations \eqref{eq:xAx-lb} and \eqref{eq:ratio-crude},
    \[
        \norm{H(\bz)}_F
        \le 
        \sum_{i=1}^{t}
        \fr{\norm{dz_iz_i^\dagger - I_d}_F}{z_i^\dagger A z_i} \cdot \fr{L(\bz_{\sim i})}{L(\bz)}
        \le 4dD
        \ll 
        n^{1/2}d \gamma
    \]
    as desired.
    Otherwise $t\ge D$.
    % By Lemma~\ref{lem:crude-frob-ub}, there exists a constant $C$ such that for all $S\subseteq [t]$, $\norm{H_S}_F \le Cnd^{1/2}$.
    We will prove by induction on $a\ge 0$ that if $S\subseteq [t]$ satisfies $|S|=t-D+a$, then
    \[
        \norm{H_S}_F 
        \le 
        \xi_a
        \triangleq
        % 2 n^{1/2}d\gamma + Ce^{-a} nd^{1/2}.
        2 n^{1/2}d\gamma + 4e^{-a} nd^{1/2}.
    \]
    % The base case $a=0$ clearly holds.
    The base case $a=0$ holds by Lemma~\ref{lem:crude-frob-ub}.
    For the inductive step, assume $a\ge 1$. 
    By the inductive hypothesis and equations \eqref{eq:L-ratio-recursion} and \eqref{eq:xAx-lb}, for all $i\in S$
    \[
        \lt|\fr{L(\bz_{S})}{L(\bz_{S \setminus \{i\}})} - 1\rt|
        \le 
        \fr{2\eps^2}{d^2} \cdot \norm{\fr{H_{S\setminus i}}{z_i^\dagger A z_i}}_{\op}
        \le 
        \fr{4\eps^2}{d^2} \norm{H_{S\setminus i}}_F
        \le 
        \fr{4\eps^2}{d^2}\xi_{a-1}.
    \]
    Since this upper bound is $o(1)$, we also have 
    \[
        \lt|\fr{L(\bz_{S\setminus \{i\}})}{L(\bz_{S})}-1\rt|
        \le 
        \fr{5\eps^2}{d^2}\xi_{a-1}.
    \]
    Write $\fr{L(\bz_{S\setminus \{i\}})}{L(\bz_{S})}-1 = \fr{5\eps^2}{d^2}\xi_{a-1}b_i$ for $b_i\in [-1,1]$.
    By Lemma~\ref{lem:uniform-frob-bd}, 
    %there is a constant $c$ such that
    \begin{align*}
        \norm{
            \sum_{i\in S} 
            \fr{dz_iz_i^\dagger - I_d}{z_i^\dagger A z_i}
            \cdot
            \lt(\fr{L(\bz_{S\setminus \{i\}})}{L(\bz_{S})}-1\rt)
        }_F
        &=
        \fr{5\eps^2}{d^2}\xi_{a-1}
        \norm{
            \sum_{i\in S} 
            \fr{dz_iz_i^\dagger - I_d}{z_i^\dagger A z_i}
            \cdot
            b_i
        }_F \\
        &\le 
        \lt(\fr{5\eps^2 n^{1/2}}{d} \gamma + \fr{10\eps^2 n}{d^{3/2}}\rt)
        \xi_{a-1}
        % \fr{5c\eps^2n}{d^{3/2}}\xi_{a-1}
        \le 
        e^{-1}\xi_{a-1},
    \end{align*}
    using the hypotheses $\gamma \ll d/(\eps^2 n^{1/2})$ and $n \ll d^{3/2}/\eps^2$.
    By the triangle inequality, equation~\eqref{eq:xAx-lb}, and our choice of $D$,
    \[
        \norm{K_S}_F
        \le 
        \norm{K(\bz)}_F
        + 
        \sum_{i\in [t]\setminus S} 
        \fr{\norm{dz_iz_i^\dagger - I_d}_F}{z_i^\dagger A z_i}
        \le 
        n^{1/2}d\gamma + 2dD
        \le 
        \fr{101}{100}n^{1/2}d\gamma.
    \]
    Hence
    \begin{align*}
        \norm{H_S}_F
        &\le 
        \norm{K_S}_F
        +
        \norm{
            \sum_{i\in S} 
            \fr{dz_iz_i^\dagger - I_d}{z_i^\dagger A z_i}
            \cdot
            \lt(\fr{L(\bz_{S\setminus \{i\}})}{L(\bz_{S})}-1\rt)
        }_F \\
        &\le 
        \fr{101}{100}n^{1/2}d\gamma
        +
        e^{-1}\xi_{a-1}
        \le \xi_a,
    \end{align*}
    as $\fr{101}{100} + 2e^{-1} \le 2$.
    This completes the induction.
    Finally, 
    \begin{align}
        \norm{H(\bz)}_F
        &=
        \norm{H_{[t]}}_F
        \le 
        % 2n^{1/2}d\gamma + Ce^{-D} nd^{1/2}
        2n^{1/2}d\gamma + 4e^{-D} nd^{1/2}
        =
        % 2n^{1/2}d\gamma + Cn^{1/2}d
        2n^{1/2}d\gamma + 4n^{1/2}d
        \le 
        3n^{1/2}d\gamma. \qedhere
    \end{align}
\end{proof}

\subsection{Uniform Frobenius bound on the \texorpdfstring{$K(\bx_{\le t})$}{Kxt} matrix martingale}\label{sec:doob}

In this subsection, we will prove Lemma~\ref{lem:doob}. 
The proof mimics the proof of Doob's $L^2$ maximal inequality.
Let $\bx \sim p_0$, recall that $K_t = K(\bx_{\le t})$, and define $X = \sup_{1\le t\le n}\norm{K_t}_F$.

\begin{lemma}
    \label{lem:doob-aux1}
    We have that $\E{X^2}\le 4\E{\norm{K_n}_F^2}$
\end{lemma}
\begin{proof}
    We will first upper bound $\Pr*{X\ge x}$ for all $x>0$.
    Consider the stopping time $\tau = \inf \{t : \norm{K_t}_F \ge x\}\cup \{n\}$.
    Then, 
    \begin{align*}
        \Pr*{X\ge x}
        &= 
        \Pr*{\norm{K_\tau}_F\ge x} \\
        &\le 
        x^{-1} \E*{\norm{K_\tau}_F \Id\{\norm{K_\tau}_F \ge x\}} \\
        &\le 
        x^{-1}
        \E*{\E*{\norm{K_n}_F|\cF_\tau} \Id\{\norm{K_\tau}_F \ge x\}} \\
        &=
        x^{-1}
        \E*{\norm{K_n}_F \Id\{X \ge x\}}.
    \end{align*}
    The first estimate is by Markov's inequality, and the second is by convexity of the norm $\norm{\cdot}_F$.
    Thus,
    \begin{align*}
        \E{X^2}
        &=
        \int_0^\infty \Pr{X^2 \ge x} ~\de x 
        =
        \int_0^\infty \Pr{X \ge x} 2x ~\de x 
        \le 
        \int_0^\infty 2\E*{\norm{K_n}_F \Id\{X \ge x\}} ~\de x \\
        &= 2\E*{\norm{K_n}_F X} 
        \le 
        2\sqrt{\E*{\norm{K_n}_F^2} \E*{X^2}}.
    \end{align*}
    Rearranging yields the result.
\end{proof}

\begin{lemma}
    \label{lem:doob-aux2}
    We have that $\E*{\norm{K_n}_F^2} \lesssim nd^2$.
\end{lemma}
\begin{proof}
    We can expand
    \begin{equation}
        \label{eq:kn-fnorm}
        \E{\norm{K_n}_F^2}
        = 
        \sum_{i=1}^{n} 
        \E*{\norm{
            \fr{dx_ix_i^\dagger - I_d}{x_i^\dagger A x_i}
        }_F^2}
        +
        2
        \sum_{1\le i<j\le n}
        \E*{\lt\la 
            \fr{dx_ix_i^\dagger - I_d}{x_i^\dagger A x_i}, 
            \fr{dx_jx_j^\dagger - I_d}{x_j^\dagger A x_j}
        \rt\ra}.
    \end{equation}
    Since 
    \[
        \E*{\fr{dx_jx_j^\dagger - I_d}{x_j^\dagger A x_j} | \cF_{j-1}}
        = 
        \sum_{x_j} \omega_{x_j} (dx_jx_j^\dagger - I_d)
        = 0,
    \]
    for any $i < j$ we have 
    \[
        \E*{\lt\la 
            \fr{dx_ix_i^\dagger - I_d}{x_i^\dagger A x_i}, 
            \fr{dx_jx_j^\dagger - I_d}{x_j^\dagger A x_j}
        \rt\ra}
        = 
        \E*{
            \lt\la 
                \fr{dx_ix_i^\dagger - I_d}{x_i^\dagger A x_i},
                \E*{\fr{dx_jx_j^\dagger - I_d}{x_j^\dagger A x_j} | \cF_{j-1}}
            \rt\ra
        }
        =0.
    \]
    The other expectation in \eqref{eq:kn-fnorm} can be bounded by (recalling \eqref{eq:xAx-lb})
    \[
        \E*{\norm{
            \fr{dx_ix_i^\dagger - I_d}{x_i^\dagger A x_i}
        }_F^2}
        \le 
        2\E*{
            \fr{\la dx_ix_i^\dagger - I_d, dx_ix_i^\dagger - I_d\ra}{x_i^\dagger A x_i}
        }
        =
        2d(d-1)\E*{
            \fr{1}{x_i^\dagger A x_i}
        }
        =2d(d-1).
    \]
    Therefore $\E*{\norm{K_n}_F^2} \le 2nd(d-1) \lesssim nd^2$.
\end{proof}

\begin{proof}[Proof of Lemma~\ref{lem:doob}]
    Follows immediately from Lemmas~\ref{lem:doob-aux1} and \ref{lem:doob-aux2}.
\end{proof}

\section{Lower Bound for Off-Diagonal Perturbations}
\label{sec:offdiag}
In this section we consider the family of perturbations which correspond to the ``off-diagonal'' case described in Section~\ref{sec:overview}. 
More formally, let $d_1\ge d_2$ and $A\in\R^{d_1\times d_1}$ and $B\in\R^{d_2\times d_2}$ be diagonal matrices with diagonal entries $a_1\ge \cdots \ge a_{d_1} > 0$ and $b_1 \ge \cdots \ge b_{d_2} > 0$ satisfying $2a_{d_1} \ge a_1$, $2b_{d_2} \ge b_1$, and $\Tr(A)+\Tr(B)=1$.
We abbreviate $a_{d_1} = a$, $b_{d_2}=b$.
With these settings, we consider the task of distinguishing between the following two alternatives:
\begin{equation}
    H_0: \rho = \begin{pmatrix}
        A & 0 \\
        0 & B
    \end{pmatrix}
    \qquad 
    \text{and} 
    \qquad 
    H_1: \rho = \begin{pmatrix}
        A & \frac{\eps}{d_2} \oG \\
        \frac{\eps}{d_2} \oG^{\dagger} & B
    \end{pmatrix}. \label{eq:offdiag_task}
\end{equation}
Here, $\oG \sim \Gin_U(d_1,d_2)$ for the $U$ given by Lemma~\ref{lem:gin-trunc} below. 

% For convenience, we will refer to $a_1$ and $b_1$ as $a$ and $b$ respectively.
% \bh{any chance we can pick one of $(a,b)$ and $(a, b)$ to use in the calculations below, since they differ by constant factor anyway?}
% \bh{rescaling $G$'s by $1/\sqrt{d_1}$ in progress}
\begin{restatable}{lemma}{gintrunc}
    \label{lem:gin-trunc}
    For $d_1 \ge d_2$, there exists $U\subseteq \R^{d_1\times d_2}$ such that if $G\sim \Gin(d_1,d_2)$, then $\Pr{G\not\in U} \le \exp(-0.1d_1)$ and on the event $G\in U$, we have $\norm{G}_{\op} \le 3$ and $\norm{M}_1 \ge d_2/3$ for
    \[
        M = 
        \begin{pmatrix}
            0 & G \\ 
            G^\dagger & 0
        \end{pmatrix}.
    \]
\end{restatable}
\noindent
We defer the proof of this lemma to Appendix~\ref{sec:basic-comps}.

% \bh{say something about $na \gg \sqrt{d_1}$}

\paragraph{Parameter Settings.} We will assume the parameters $a,b,d_1, d_2, \eps$ satisfy the following relations: 
\begin{equation}
    d_1 \gg 1 \qquad \qquad \eps \le \fr{1}{10^{20}} \fr{d_2 \sqrt{ab}}{\log \fr{1}{a}} \qquad \qquad d_1 \sqrt{a} \le d_2 \sqrt{b} \qquad \qquad d_1 \ge d_2
\end{equation}
% \begin{itemize}
% \item $d_1 \gg 1$
% \item $\eps \le \fr{1}{10^6} \fr{d_2 \sqrt{ab}}{\log \fr{1}{a}}$
% \item $d_1 \sqrt{a} \le d_2 \sqrt{b}$
% \item $d_1 \ge d_2$
% \end{itemize}
% \jerry{can we not make this an itemize?}

\begin{remark}
For most places, it suffices to use $\eps \le \fr{1}{10^{20}} d_2 \sqrt{ab}$ so we will often drop the $\log(1/a)$ except for the few places where it is actually necessary.
\end{remark}

\noindent Our main result for the distinguishing task \eqref{eq:offdiag_task} is the following.
\begin{theorem}\label{thm:main_offdiag}
    Under the assumed parameter settings, %if $d_1 \gg 1$, %and $na \gg \sqrt{d_1}$, 
    % then 
    the copy complexity of distinguishing between $H_0$ and $H_1$ with incoherent measurements is $\Omega(d^{1/2}_1 d_2/\eps^2)$.
\end{theorem}

\noindent We first record several elementary consequences of the parameters settings.

\begin{fact}
    Under the parameter settings, $\rho \sim H_1$ is psd with trace distance at least $\eps/3$ to $H_0$. 
\end{fact}
\begin{proof}
    The trace distance bound is immediate from the properties of $U$ given by Lemma~\ref{lem:gin-trunc}.
    To show $\rho$ is psd, note that for any nonzero $x\in \C^{d_1}$, $y\in \C^{d_2}$, 
    \begin{align*}
        (x,y)^\dagger \rho (x,y)
        &= 
        x^\dagger A x + y^\dagger B y + \fr{2\eps}{d_2} x^\dagger \oG y
        \ge 
        a\norm{x}^2 + b\norm{y}^2 - \fr{6\eps}{d_2} \norm{x}\norm{y} \\
        &\ge 
        a\norm{x}^2 + b\norm{y}^2 - \fr{6}{10^{20}} \sqrt{ab} \norm{x}\norm{y}
        > 0. \qedhere
    \end{align*}
\end{proof}
\begin{fact}
    \label{fact:d2b}
    Under the parameter settings, $bd_2 \in [1/4, 1]$.
\end{fact}
\begin{proof}
    Since $d_1\ge d_2$ and $d_1\sqrt{a} \le d_2\sqrt{b}$, we have $d_1a \le \fr{d_2}{d_1} d_2b \le d_2b$.
    Thus $1 = \Tr(A)+\Tr(B) \le 2d_1a + 2d_2b \le 4d_2b$ and $1 \ge \Tr(B) \ge d_2b$.
\end{proof}

% We let $n$ be the number of measurements we can make and assume $n \ll d_2d_1^{1/2}/\eps^2$. 
% By the bound on $\eps$ and the relation $d_1 \ge d_2$, we have $\frac{d_2\sqrt{d_1}}{\eps^2} \gtrsim \frac{\sqrt{d_1}}{ab\cdot d_2} \ge \frac{\sqrt{d_1/d_2}}{a\sqrt{b}} \ge 1/\sqrt{ab}$, so we may further assume that $n \ge 1/(\alpha\sqrt{ab})$ for some slowly growing $\alpha$ by adding superfluous measurements. \bh{sorry, where is this used? also let's not use $\alpha$ because that's already taken} These will be assumed throughout the remainder of this section. 
Take any learning tree $\calT$ corresponding to an algorithm for this task that uses $n \ll d_2d_1^{1/2}/\eps^2$ incoherent measurements. 
The parameter settings imply that we may further assume, by taking additional superfluous measurements, that $(\log \fr{n}{\sqrt{ab}})^2 / (d_1 a) \le n$.
Similarly to the previous section, we let $p_0$ and $p_1$ denote the distributions over leaves of $\calT$ induced by $\rho$ under $H_0$ and $H_1$ respectively, and we will show $\TV(p_0,p_1) \to 0$. 

Because of the block structure in \eqref{eq:offdiag_task}, we denote leaves of $\calT$ by $(\bx,\by) = ((x_1,y_1),\ldots,(x_n,y_n))$.
Here, each $(x_i,y_i)$ satisfies $x_i\in \C^{d_1}$, $y_i\in \C^{d_2}$, and $(x_i,y_i)\in \C^{d_1+d_2}$, and corresponds to an outcome from some (adaptively chosen) rank-1 POVM which we write as
\[
    \brc{ (x,y)(x,y)^{\dagger}}_{(x,y) \in \mathcal{P}}.
\]
Note that $(x,y)$ are not necessarily unit vectors.  We only require that 
\[
\sum_{(x,y) \in \mathcal{P}}  (x,y)(x,y)^{\dagger} = I_{d_1+d_2} \,.
\]
We let $L^*(\cdot)$ denote the likelihood ratio between $p_1$ and $p_0$, i.e. $L^*((\bx,\by)) \triangleq p_1((\bx,\by))/p_0((\bx,\by))$. 
Note that
\[
    L^*((\bx,\by))
    = 
    \E[\oG\sim \Gin_U(d_1,d_2)]*{
        \prod_{i = 1}^n
        \left(1 + \fr{2\eps}{d_2} \cdot \frac{x^{\dagger}_i \oG y_i}{x^{\dagger}_i A x_i + y^{\dagger}_i B y_i} \right)
    }
\]
We similarly define the non-truncated estimate
\begin{equation}
    \label{eq:def-L-offdiag}
    L((\bx,\by))
    = 
    \E[G\sim \Gin(d_1,d_2)]*{
        \prod_{i = 1}^n
        \left(1 + \fr{2\eps}{d_2} \cdot \frac{x^{\dagger}_i \oG y_i}{x^{\dagger}_i A x_i + y^{\dagger}_i B y_i} \right)
    }.
\end{equation}
We will abuse notation and write $L((\bz,\bw))$ for any sequence of unit vectors $(\bz,\bw) = ((z_1,w_1),\ldots,$ $(z_t,w_t))$ of length not necessarily $n$.
This is defined identically to \eqref{eq:def-L-offdiag}.
We let $L((\bx,\by),(\bx,\by))$ denote the value of $L$ on input $((x_1,y_1),(x_1,y_1),\ldots,(x_n,y_n),(x_n,y_n))$.

The following proposition is analogous to Proposition~\ref{prop:lstar-reduction} and will be the main ingredient in our proof.
\begin{proposition}
    \label{prop:lstar-reduction-offdiag}
    There exists a subset $S$ of the leaves of $\calT$ such that $\Pr[p_0]{S} = 1-o(1)$ and for all $(\bx,\by)\in S$, $|L((\bx,\by))-1|=o(1)$ and $L((\bx,\by),(\bx,\by)) \leq e^{0.02\sqrt{d_1d_2}}$.
\end{proposition}

\noindent We now prove Theorem~\ref{thm:main_offdiag} assuming Proposition~\ref{prop:lstar-reduction-offdiag}.

\begin{proof}[Proof of Theorem~\ref{thm:main_offdiag}]
    Analogous to the proof of Theorem~\ref{thm:main_standard} assuming Proposition~\ref{prop:lstar-reduction}.
    % For the $U$ from Lemma~\ref{lem:gin-trunc}, let
    % \[
    %     \oL((\bx,\by))
    %     = 
    %     \E[G\sim \Gin(d_1,d_2)]*{
    %         \Id\{G\in U\}
    %         \prod_{i = 1}^n
    %         \left(1 + \frac{2\eps x^{\dagger}_i G y_i}{d_2 (x^{\dagger}_i A x_i + y^{\dagger}_i B y_i)} \right)
    %     }.
    % \]
    % Then $L^*((\bx,\by)) = \Pr{U}^{-1} \oL((\bx,\by))$. 
    % For all $(\bx,\by)\in S$, by Cauchy-Schwarz,
    % \begin{align*}
    %     |L((\bx,\by)) - \oL((\bx,\by))| 
    %     &= 
    %     \Pr{U^c}^{1/2} 
    %     \E[G\sim \Gin(d_1,d_2)]*{
    %         \prod_{i = 1}^n
    %         \left(1 + \frac{2\eps x^{\dagger}_i G y_i}{d_2 (x^{\dagger}_i A x_i + y^{\dagger}_i B y_i)} \right)^2
    %     }^{1/2} \\
    %     &=
    %     \sqrt{\Pr{U^c}^{1/2} L((\bx,\by),(\bx,\by))}
    %     = o(1)
    % \end{align*}
    % using that $\Pr{U^c} \le \exp(-\Omega(d))$ and $L((\bx,\by),(\bx,\by)) \ll e^{\sqrt{d_1}}$.
    % Moreover, we have $|L((\bx,\by)) - 1| = o(1)$.
    % Thus for all $(\bx,\by)\in S$, $\oL((\bx,\by)) = 1+o(1)$ and 
    % \[
    %     |L^*((\bx,\by))-1| 
    %     \le 
    %     \fr{\Pr{U^c}}{\Pr{U}} \oL((\bx,\by))
    %     +
    %     |\oL((\bx,\by))-1| 
    %     = o(1).
    % \]
    % Finally,
    % \[
    %     \TV(p_0,p_1)
    %     =
    %     2\E[(\bx,\by)\sim p_0]{(L^*((\bx,\by))-1)_-}
    %     \le 
    %     2\sup_{(\bx,\by)\in S} (L^*((\bx,\by))-1)_- + 2 \Pr[p_0]{S^c} 
    %     = o(1).
    % \]
\end{proof}

\subsection{Recursive evaluation of likelihood ratio}

Similarly to the previous section, we obtain a recursive expression for $L$.
Let the sequence of unit vectors $(\bz,\bw) = ((z_1,w_1),\ldots,(z_t,w_t))$ be as above.
For $1\le i\le t$, let $(\bz,\bw)_{\sim i}$ be this sequence with $(z_i,w_i)$ omitted.
Similarly, for $1\le i<j\le t$, let $(\bz,\bw)_{\sim i,j}$ be this sequence with $(z_i,w_i)$ and $(z_j,w_j)$ omitted.

\begin{lemma}
    \label{lem:L-recursion-offdiag}
    The function $L$ satisfies
    \[
        L((\bz,\bw))
        =
        L((\bz,\bw)_{\sim t})
        +
        \fr{4\eps^2}{d_1d_2^2}
        \sum_{i=1}^{t-1}
        \lt[
            \fr{\la z_i,z_t\ra \la w_i,w_t\ra}{(z_i^\dagger A z_i + w_i^\dagger A w_i)(z_t^\dagger A z_t + w_t^\dagger A w_t)}
            \cdot 
            L((\bz,\bw)_{\sim i,t})
        \rt].
    \]
\end{lemma}
\begin{proof}
    Analogous to Lemma~\ref{lem:L-recursion}.
    The pairwise moments are evaluated by
    \begin{align}
        \E[G\sim \Gin(d_1,d_2)]{(x^\dagger G y)(z^\dagger G w)} &= \fr{1}{d_1} \la x,z\ra \la y,w\ra. \qedhere
    \end{align}
\end{proof}

\subsection{High probability bound on likelihood ratio at leaves}

This subsection gives the main part of the proof of Proposition~\ref{prop:lstar-reduction-offdiag}.
For the sequence of unit vectors $(\bz,\bw) = ((z_1,w_1),\ldots,(z_t,w_t))$ as above, define
\[
    H((\bz,\bw)) = 
    \sum_{i=1}^t 
    \fr{z_iw_i^\dagger}{z_i^\dagger A z_i + w_i^\dagger B w_i}
    \cdot 
    \fr{L((\bz,\bw)_{\sim i})}{L((\bz,\bw))}.
\]
Lemma~\ref{lem:L-recursion-offdiag} can be rewritten as
\begin{equation}
    \label{eq:L-ratio-recursion-offdiag}
    \fr{L((\bz,\bw))}{L((\bz,\bw)_{\sim t})}
    =
    1 + \fr{4\eps^2}{d_1d_2^2} \cdot 
    \fr{z_t^\dagger H((\bz,\bw)_{\sim t}) w_t}{z_t^\dagger A z_t + w_t^\dagger B w_t}.
\end{equation}
Further define
\[
    K((\bz,\bw)) = 
    \sum_{i=1}^t 
    \fr{z_iw_i^\dagger}{z_i^\dagger A z_i + w_i^\dagger B w_i}
    \qquad
    \text{and}
    \qquad
    \kappa((\bz,\bw)) = 
    \sup_{b_1,\ldots,b_t\in [-1,1]}
    \norm{
        \sum_{i=1}^t 
        b_i \fr{z_iw_i^\dagger \norm{z_i} \norm{w_i}}{(z_i^\dagger A z_i + w_i^\dagger B w_i)^2}
    }_F.
\]
As in the previous section, $K$ will be our proxy for $H$. 
The condition that the error terms in the bootstrapping argument contract correspond to an upper bound on $\kappa((\bz,\bw))$.
In contrast to Lemma~\ref{lem:uniform-frob-bd}, it is no longer true in this setting that boundedness of $K((\bz,\bw))$ implies the required bound on $\kappa((\bz,\bw))$ (see Appendix~\ref{app:Kvskappa}); this will instead be separately proved in Lemma~\ref{lem:balanced-offdiag} below.

The following three lemmas are the analogs of Lemmas~\ref{lem:bootstrap} and \ref{lem:doob}. 
Lemma~\ref{lem:bootstrap-offdiag} deterministically controls $H$ given bounds on $K$ and $\kappa$, and Lemmas~\ref{lem:doob-offdiag} and \ref{lem:balanced-offdiag} give the required high probability bounds on $K((\bx,\by))$ and $\kappa((\bx,\by))$.

\begin{lemma}
    \label{lem:bootstrap-offdiag}
    Suppose $\gamma \gg 1$.
    If $(\bz, \bw) = ((z_1,w_1),\ldots,(z_t,w_t))$ satisfies $t\le n$, $\norm{K((\bz,\bw))}_F \le \sqrt{nd_1d_2} \gamma$ and $\kappa((\bz,\bw)) \le \fr{1}{10^3} d_1d_2^2 / \eps^2$ %\sitan{should this be $\frac{1}{10}d_1d_2^2/\eps^2$?} \bh{yes, ahh thanks for catching all my typos}
    , then $\norm{H((\bz,\bw))}_F \le 3\sqrt{nd_1d_2} \gamma$.
\end{lemma}

\begin{lemma}
    \label{lem:doob-offdiag}
    For $(\bx,\by)\sim p_0$, let $(\bx,\by)_{\le t} = ((x_1,y_1),\ldots,(x_t,y_t))$ be the length-$t$ prefix of $(\bx,\by)$. Then $\E*{\sup_{1\le t\le n} \norm{K((\bx,\by)_{\le t})}_F^2} \lesssim nd_1d_2$.
\end{lemma}

\begin{lemma}
    \label{lem:balanced-offdiag}
    If $(\bx,\by)\sim p_0$, then $\Pr{\kappa((\bx,\by)) > \fr{1}{10^3} d_1d_2^2/\eps^2} = o(1)$.
\end{lemma}

\noindent We now prove Proposition~\ref{prop:lstar-reduction-offdiag} assuming Lemmas~\ref{lem:bootstrap-offdiag}, \ref{lem:doob-offdiag}, and \ref{lem:balanced-offdiag}.
These lemmas will be proved in Subsections~\ref{sec:bootstrap-offdiag}, \ref{sec:doob-offdiag}, and \ref{sec:balanced-offdiag}.

Let $\alpha, \beta$ be slowly-growing functions with $1 \ll \alpha \ll \beta \ll d_1^{1/2}d_2/(\eps^2 n)$ and furthermore $\alpha^2 \ll d_1^{1/2}d_2/(\eps^2 n)$. 
This is possible because $n \ll d_1^{1/2}d_2 / \eps^2$.
Let $(\bx,\by)\sim p_0$. For $1\le t\le n$, define the filtration $\cF_t = \sigma((\bx,\by)_{\le t})$ and the sequences
\[
    H_t = H((\bx,\by)_{\le t}),
    \qquad
    K_t = K((\bx,\by)_{\le t}),
    \qquad
    \kappa_t = \kappa((\bx,\by)_{\le t}),
    \qquad
    \Phi_t = L((\bx,\by)_{\le t}).
\]
Consider the stopping time (with respect to $\cF_t)$ 
\[
    \tau = \inf\lt\{
        t : \norm{K_t}_F > \sqrt{nd_1d_2} \alpha
        ~\text{or}~
        \kappa_t > \fr{1}{10^3} d_1d_2^2
        ~\text{or}~
        |\Phi_t-1| > \fr{\eps^2 n}{d_1^{1/2}d_2} \beta
    \rt\}
    \cup
    \{\infty\}
\]
and stopped sequence $\Psi_t = \Phi_{t\wedge \tau}$.

\begin{claim}
    \label{cl:frob-sup-offdiag}
    With probability $1-o(1)$, $\norm{K_t}_F \le \sqrt{nd_1d_2} \alpha$ for all $1\le t\le n$.
\end{claim}
\begin{proof}
    Follows from Lemma~\ref{lem:doob-offdiag} and Markov's inequality.
\end{proof}

\begin{claim}
    \label{cl:lx-offdiag}
    With probability $1-o(1)$, $|\Psi_n-1| \le \fr{\eps^2 n}{d_1^{1/2}d_2} \beta$.
\end{claim}
\begin{proof}
    This is analogous to Claim~\ref{cl:lx-bd}, and we only sketch the differences.
    Note that $\Psi_t$ is a multiplicative martingale. 
    We will bound the quadratic increment $\E{(\fr{\Psi_t}{\Psi_{t-1}})^2 | \cF_{t-1}}$. 
    This is $1$ if $\tau \le t-1$, and otherwise by \eqref{eq:L-ratio-recursion-offdiag}, (because the linear term expects to $0$)
    \[
        \E*{\lt(\fr{\Psi_t}{\Psi_{t-1}}\rt)^2 | \cF_{t-1}}
        = 
        1 + \fr{16\eps^2}{d_1^2d_2^4} 
        \E*{\fr{(x_t^\dagger H_{t-1}y_t)^2}{(x_t^\dagger A x_t + y_t^\dagger B y_t)^2} | \cF_{t-1}}.
    \]
    This last expectation is bounded by
    \begin{align*}
        \E*{\fr{(x_t^\dagger H_{t-1}y_t)^2}{(x_t^\dagger A x_t + y_t^\dagger B y_t)^2} | \cF_{t-1}}
        &= 
        \E*{\fr{x_t^\dagger H_{t-1}y_ty_t^\dagger H_{t-1}x_t}{(x_t^\dagger A x_t + y_t^\dagger B y_t)^2} | \cF_{t-1}}
        \le 
        \E*{\fr{\norm{y_t}^2 x_t^\dagger H_{t-1}^2 x_t}{(x_t^\dagger A x_t + y_t^\dagger B y_t)^2} | \cF_{t-1}} \\
        &\le
        \fr{1}{b}
        \E*{\fr{x_t^\dagger H_{t-1}^2 x_t}{x_t^\dagger A x_t + y_t^\dagger B y_t} | \cF_{t-1}} 
        = 
        \fr{\norm{H_{t-1}}_F^2}{b}
        \le 
        4d_2 \norm{H_{t-1}}_F^2
    \end{align*}
    using Fact~\ref{fact:d2b}.
    Since $\tau > t-1$, we have $\norm{K_{t-1}}_F \le \sqrt{nd_1d_2}\alpha$ and $\kappa_t \le \fr{1}{10} d_1d_2^2$. 
    Thus Lemma~\ref{lem:bootstrap-offdiag} implies $\norm{H_{t-1}}_F \le 3\sqrt{nd_1d_2}\alpha$, and 
    \[
        \E*{\lt(\fr{\Psi_t}{\Psi_{t-1}}\rt)^2 | \cF_{t-1}}
        \le 
        1 + \fr{64\eps^4}{d_1^2d_2^3} \norm{H_{t-1}}_F^2
        \le 
        1 + \fr{576\eps^4n}{d_1d_2^2} \alpha^2.
    \]
    Analogously to the proof of Claim~\ref{cl:lx-bd}, this implies
    \[
        \E{(\Psi_n-1)^2} 
        \le 
        \fr{2\cdot 576\eps^4n^2}{d_1d_2^2}\alpha^2.
    \]
    The result now follows from Markov's inequality.
\end{proof}

\begin{claim}
    \label{cl:lxx-offdiag}
    If $\norm{K_n}_F \le \sqrt{nd_1d_2} \alpha$ and $\kappa_n \le \fr{1}{10^3} \frac{d_1d_2^2}{\eps^2}$, then $L((\bx,\by),(\bx,\by)) \leq e^{0.02 \sqrt{d_1d_2}}$.
\end{claim}

% \begin{lemma}\label{lem:trunc_offdiag}
%     If $(\bx,\by)$ satisfies
%     \begin{equation}
%         \norm{\sum^n_{i=1} \frac{x_iy_i^{\dagger}}{x^{\dagger}_iA x_i + y^{\dagger}_i By_i}}_F \le \sqrt{\frac{nd_1}{b}}, \label{eq:frobbound}
%     \end{equation} \TODO{we should define this event somewhere separately} then $L((\bx,\by),(\bx,\by)) \ll e^{\sqrt{d_1}}$.
% \end{lemma}

\begin{proof}
    Using the elementary inequality $e^z \ge 1 + z$ and then Cauchy-Schwarz, we can upper bound $L((\bx,\by),(\bx,\by))$ by
    \begin{equation}\label{eq:product2}
    \begin{split}
        L((\bx,\by),(\bx,\by)) 
        &\le \E[G\sim\Gin(d_1,d_2)]*{\exp\left( \sum_{i = 1}^n \frac{4\eps}{d_2}\frac{x_i^\dagger G y_i}{x^{\dagger}_iAx_i + y^{\dagger}_iBy_i} + \left(\frac{2\eps}{d_2}\frac{x_i^\dagger G y_i}{x^{\dagger}_iAx_i + y^{\dagger}_iBy_i}\right)^2\right)}
        \\ &\leq \sqrt{\E[G\sim\Gin(d_1,d_2)]*{\exp\left( \sum_{i = 1}^n \frac{8\eps}{d_2}\frac{x_i^\dagger G y_i}{x^{\dagger}_iAx_i + y^{\dagger}_iBy_i} \right)}} \\ & \qquad \times \sqrt{ \E[G\sim\Gin(d_1,d_2)]*{\exp\left( \sum_{i = 1}^n \frac{8 \eps^2}{d_2^2} \left(\frac{x_i^\dagger G y_i}{x^{\dagger}_iAx_i + y^{\dagger}_iBy_i}\right)^2\right)}} \,.
    \end{split}
    \end{equation}
    Now, we bound the two terms in the last product in~\eqref{eq:product2} separately.  First, we have
    \begin{equation}
        \E[G\sim\Gin(d_1,d_2)]*{\exp\left( \sum_{i = 1}^n \frac{8\eps}{d_2}\frac{x_i^\dagger G y_i}{x^{\dagger}_iAx_i + y^{\dagger}_iBy_i} \right)} =
        \E[G\sim\Gin(d_1,d_2)]*{
            \exp\lt(\fr{8\eps}{d_2} \iprod*{G, K_n}\rt)
        }.
        \label{eq:luxy2}
    \end{equation}
    Note that $\fr{8\eps}{d_2} \iprod*{G,K_n}$ is distributed as a Gaussian with variance $\fr{64\eps^2}{d_1d_2^2} \norm{K_n}_F^2 \le \fr{64\eps^2n}{d_2}\alpha^2$.
    So we can bound \eqref{eq:luxy2} by 
    % \end{equation}
    % By the assumed bound in \eqref{frobbound}, the inner product in the exponent in \eqref{eq:luxy2} is distributed according to a Gaussian with variance $O\left(\frac{\eps^2}{d_2^2 d_1}\cdot \frac{n d_1}{b}\right) = O\left(\frac{\eps^2 n}{d_2}\right)$, where we used that $b = \Omega(1/d_2)$. So we can bound \eqref{eq:luxy2} by
    \begin{equation}
        \E[g\sim\calN(0,64\eps^2 n\alpha^2/d_2))]*{\exp(g)} = e^{64\eps^2 n \alpha^2/d_2} \ll e^{\sqrt{d_1}}, \label{eq:luxy2_bound}
    \end{equation}
    as $\alpha^2 \ll d_1^{1/2}d_2 /(\eps^2 n)$ by assumption.  Next, we bound the second term in the product in~\eqref{eq:product2}.  Use $\textsf{vec}(G)$ to denote rearranging $G$ as a vector in $\R^{d_1d_2}$ (done in a consistent way) and use $\otimes $ to denote the Kronecker product of two matrices.  Define $Q \in \R^{d_1d_2 \times d_1d_2}$ as 
    \[
    Q = \sum_{i = 1}^n \frac{x_i x_i^\dagger \otimes y_i y_i^\dagger}{(x^{\dagger}_iAx_i + y^{\dagger}_iBy_i)^2} \,.
    \]
    Let $X \in \C^{d_1d_2 \times n}$ be the matrix with columns given by $\frac{\textsf{vec}(x_i y_i^\dagger)}{\norm{x_i}\norm{y_i}}$ for $i = 1,2, \dots , n$.  Let 
    \[
    \theta_i = \frac{\norm{x_i}^2\norm{y_i}^2}{(x^{\dagger}_iAx_i + y^{\dagger}_iBy_i)^2}
    \]
    and let $D \in \R^{n \times n}$ be the diagonal matrix whose diagonal entries are $\theta_i$ for $i = 1,2, \dots , n$.  We can write
    \[
    \norm{Q}_F^2 = \sum_{i,j}  \theta_i \theta_j \bigg\la \frac{x_iy_i^\dagger}{\norm{x_i}\norm{y_i}}, \frac{x_j y_j^\dagger}{\norm{x_j}\norm{y_j}} \bigg\ra^2 = \la DX^\dagger X D , X^\dagger X \ra  \,. 
    \]
    By Grothendieck's inequality, we can replace the second term $X^\dagger X$ with $\sigma^\dagger \sigma$ for some $\sigma \in [-1,1]^n$ while incurring at most a factor of $2$ in the inequality.  Thus, we have
    \[
    \norm{Q}_F \leq \sqrt{2} \max_{\sigma \in [-1,1]^n} \norm{ \sum_{i = 1}^n \frac{\sigma_i x_iy_i^\dagger \norm{x_i}\norm{y_i}}{(x^{\dagger}_iAx_i + y^{\dagger}_iBy_i)^2} }_F \leq \sqrt{2} \kappa_n \leq \frac{d_1d_2^2}{700\eps^2} \,.
    \]
    In particular, we have $\norm{Q}_{\op} \leq (d_1d_2^2)/(700\eps^2)$ and $\norm{Q}_1 \leq d_1^{3/2}d_2^{5/2}/(700\eps^2)$.  Let $\lambda_1, \dots , \lambda_{d_1d_2}$ be the eigenvalues of $Q$.  Returning to the last term in~\eqref{eq:product2}, we can write
    \begin{align}
        \E[G\sim\Gin(d_1,d_2)]*{\exp\left( \sum_{i = 1}^n \frac{8 \eps^2}{d_2^2} \left(\frac{x_i^\dagger G y_i}{x^{\dagger}_iAx_i + y^{\dagger}_iBy_i}\right)^2\right)} &= \E[G\sim\Gin(d_1,d_2)]*{\exp\left(\frac{8 \eps^2}{d_2^2} \textsf{vec}(G)^\dagger  Q \textsf{vec}(G) \right)} \\ & = \E[v \sim N(0, I_{d_1d_2})]*{\exp\left(\frac{8\eps^2}{d_1d_2^2}v^\dagger  Q v \right)} \\ & = \prod_{j = 1}^{d_1d_2} \E[g \sim N(0, 1)]*{ \exp\left( \frac{8\eps^2g^2}{d_1d_2^2} \lambda_j \right)} \\ &\leq e^{10\eps^2(\lambda_1 + \dots + \lambda_{d_1d_2})/(d_1d_2^2)} \\ & \leq e^{\sqrt{d_1d_2}/70} \,,\label{eq:quad-term-2}
    \end{align}
    where in the fourth step we used our bound on $\norm{Q}_\op$ together with the fact that $\E[g]{e^{cx^2}} = (1 - 2c)^{-1} \le e^{5c/4}$ for sufficiently small $c$, and in the last step we used our bound on $\norm{Q}_1$.

Putting~\eqref{eq:product2},~\eqref{eq:luxy2},~\ref{eq:luxy2_bound}, and~\eqref{eq:quad-term-2} together, we conclude $L((\bx,\by),(\bx,\by)) \ll e^{\sqrt{d_1d_2}}$ as desired.
\end{proof}

\begin{proof}[Proof of Proposition~\ref{prop:lstar-reduction-offdiag}]
    Define the event 
    \[
        S = \lt\{
            \sup_{1\le t\le n} \norm{K_t}_F \le \sqrt{nd_1d_2} \alpha
            ~\text{and}~
            \kappa_n \le \fr{1}{10^3} \frac{d_1d_2^2}{\eps^2}
            ~\text{and}~
            \Psi_n\le \fr{\eps^2 n}{d_1^{1/2}d_2}\beta
        \rt\}.
    \]
    By Lemma~\ref{lem:balanced-offdiag} and Claims~\ref{cl:frob-sup-offdiag} and \ref{cl:lx-offdiag}, $\Pr[p_0]{S} = 1-o(1)$.
    If $S$ holds, then $\tau = \infty$: indeed, if $\tau = t < \infty$, then one of $\norm{K_t}_F > \sqrt{nd_1d_2}\alpha$, $\kappa_t > \fr{1}{10^3} \frac{d_1d_2^2}{\eps^2}$, and $|\Phi_t-1| > \fr{\eps^2n}{d_1^{1/2}d_2}\beta$ holds. 
    Since $\Psi_n = \Phi_t$ and $\kappa_n \ge \kappa_t$ (because in the definition of $\kappa((\bx,\by))$ we can take $b_{t+1}=\cdots=b_n=0$) this contradicts $S$. 
    
    So, $\tau=\infty$ on $S$. This implies $|L((\bx,\by))-1| = |\Phi_n-1| \le \fr{\eps^2n}{d_1^{1/2}d_2} \beta = o(1)$.
    Moreover, $\norm{K_n}_F \le \sqrt{nd_1d_2}\alpha$, so by Claim~\ref{cl:lxx-offdiag} we have $L((\bx,\by),(\bx,\by)) \leq e^{0.02\sqrt{d_1d_2}}$.
\end{proof}

\subsection{Bounding \texorpdfstring{$H$}{H} in Frobenius norm by bootstrapping}
\label{sec:bootstrap-offdiag}

In this section, we will prove Lemma~\ref{lem:bootstrap-offdiag}.
Let $(\bz,\bw) = ((z_1,w_1),\ldots,(z_t,w_t))$ be a sequence of unit vectors satisfying $t\le n$, $\norm{K((\bz,\bw))}_F \le \sqrt{nd_1d_2}\alpha$, and $\kappa((\bz,\bw)) \le \fr{1}{10} d_1d_2^2$, where $\gamma \gg 1$.
For $S\subseteq [t]$, let $(\bz,\bw)_S = ((z_i,w_i))_{i\in S}$.
Let
\[
    H_S = 
    \sum_{i\in S}
    \fr{z_iw_i^\dagger}{z_i^\dagger A z_i + w_i^\dagger B w_i}
    \cdot 
    \fr{L((\bz,\bw)_{S\setminus i})}{L((\bz,\bw)_S)}
    \qquad
    \text{and}
    \qquad
    K_S =
    \sum_{i\in S}
    \fr{z_iw_i^\dagger}{z_i^\dagger A z_i + w_i^\dagger B w_i}.
\]
\begin{lemma}
    \label{lem:crude-frob-ub-offdiag}
    For all $S\subseteq [t]$, $\norm{H_S}_F \le n/\sqrt{ab}$.
\end{lemma}
\begin{proof}
    For any fixed $\oG \in U$, for the $U$ given by Lemma~\ref{lem:gin-trunc}, and any unit vector $(z,w)$, 
    \[
        \fr{2\eps}{d_2} \lt|\fr{z^\dagger \oG w}{z^\dagger A z + w^\dagger B w}\rt|
        \le 
        \fr{2\eps}{d_2} \cdot \fr{3\norm{z}\norm{w}}{a\norm{z}^2 + b\norm{w}^2}
        \le 
        \fr{3\eps}{d_2 \sqrt{ab}} \le \fr{3}{10^6}.
    \]
    Thus, for all $i$, $L((\bz,\bw)_S)/L((\bz,\bw)_{S\setminus i}) \in [1/(1 + 10^{-5}), 1 + 10^{-5}]$ which implies  
    \begin{equation}
        \label{eq:lr-const-bd-offdiag}
        \fr{L((\bz,\bw)_{S\setminus i})}{L((\bz,\bw)_S)}
        \in 
        \lt[\frac{1}{1 + 10^{-5}}, 1 + 10^{-5}\rt].
    \end{equation}
    Thus
    \begin{align}
        \norm{H_S}_F
        &\le 
        \sum_{i\in S} 
        \norm{\fr{z_iw_i^\dagger}{z_i^\dagger A z_i + w_i^\dagger B w_i}}_F
        \cdot (1 + 10^{-5})
        \le 
        \fr{n}{2\sqrt{ab}}
        \cdot (1 + 10^{-5})
        \le 
        \frac{n}{\sqrt{ab}}. \qedhere
    \end{align}
\end{proof}

\begin{proof}[Proof of Lemma~\ref{lem:bootstrap-offdiag}]
    Let $D = \log (n/\sqrt{ab})$.
    If $t<D$, then by equation \eqref{eq:lr-const-bd-offdiag} and the assumption $(\log \fr{n}{\sqrt{ab}})^2 / (d_1 a) \le n$,
    \[
        \norm{H((\bz,\bw))}_F
        \le 
        \sum_{i=1}^t  
        \norm{\fr{z_iw_i^\dagger}{z_i^\dagger A z_i + w_i^\dagger B w_i}}_F
        \cdot 
        (1 + 10^{-5})
        \le 
        \fr{D}{\sqrt{ab}}
        \le
        3 \sqrt{nd_1d_2}\gamma.
    \]
    Otherwise $t\ge D$. 
    We will prove by induction on $a\ge 0$ that if $S\subseteq [t]$ satisfies $|S|=t-D+a$, then
    \[
        \norm{H_S}_F \le \xi_a \triangleq 2\sqrt{nd_1d_2}\gamma + e^{-a} \fr{n}{\sqrt{ab}}.
    \]
    The base case $a=0$ holds by Lemma~\ref{lem:crude-frob-ub-offdiag}.
    For the inductive step, assume $a\ge 1$. 
    By the inductive hypothesis and equation \eqref{eq:L-ratio-recursion-offdiag}, 
    \[
        \lt|\fr{L((\bz,\bw)_S)}{L((\bz,\bw)_{S\setminus i})} - 1\rt|
        \le 
        \fr{4\eps^2}{d_1d_2^2}
        \cdot 
        \lt|\fr{z_i^\dagger H_{S\setminus i} w_i}{z^{\dagger}_iAz_i + w^{\dagger}_i Bw_i}\rt|
        \le
        \fr{4\eps^2\xi_{a-1}}{d_1d_2^2}
        \cdot 
        \fr{\norm{z_i}\norm{w_i}}{z_i^\dagger A z_i + w_i^\dagger B w_i}.
    \] %\sitan{should the denominator in the middle expression in the absolute value be $z^{\dagger}_iAz_i + w^{\dagger}_i Bw_i$?}
    %\bh{woops, thanks, fixed!}
    Thus,
    \[
        \norm{H_S}_F
        \le 
        \norm{K_S}_F
        +
        \norm{
        \sum_{i\in S}
            \fr{z_iw_i^\dagger}{z_i^\dagger A z_i + w_i^\dagger B w_i}
            \cdot 
            \lt(\fr{L((\bz,\bw)_{S\setminus i})}{L((\bz,\bw)_S)}-1\rt)
        }_F 
    \]
    These terms are bounded by (in light of \eqref{eq:lr-const-bd-offdiag})
    \begin{align}
        \norm{K_S}_F
        & \le 
        \norm{K((\bz,\bw))}_F
        +
        (1 + 10^{-5})\sum_{i\in [t]\setminus S}
        \norm{\fr{z_iw_i^\dagger}{z_i^\dagger A z_i + w_i^\dagger B w_i}}_F \\
        & \le 
        \sqrt{nd_1d_2}\gamma
        +
        \fr{D}{\sqrt{ab}}
        \le 
        1.01 \sqrt{nd_1d_2}\gamma
    \end{align}
    and, for some $b_1,\ldots,b_t \in [-1,1]$, 
    \begin{align*}
        \norm{
            \sum_{i\in S}
            \fr{z_iw_i^\dagger}{z_i^\dagger A z_i + w_i^\dagger B w_i}
            \cdot 
            \lt(\fr{L((\bz,\bw)_{S\setminus i})}{L((\bz,\bw)_S)}-1\rt)
        }_F
        &= 
        \fr{4\eps^2\xi_{a-1}}{d_1d_2^2}
        \norm{
            \sum_{i=1}^t
            b_i
            \fr{z_iw_i^\dagger \norm{z_i}\norm{w_i}}{(z_i^\dagger A z_i + w_i^\dagger B w_i)^2}
        }_F \\
        &\le 
        \fr{4\eps^2 \kappa((\bz,\bw))}{d_1d_2^2} \xi_{a-1}
        \le 
        e^{-1} \xi_{a-1}.
    \end{align*}
    Since $1.01 + \fr{2}{e} \le 2$, this implies $\norm{H_S} \le \xi_a$.
    Therefore
    \begin{align}
        \norm{H((\bz,\bw))}_F
        &= 
        \norm{H_{[t]}}_F
        \le 
        2\sqrt{nd_1d_2}\gamma + e^{-D} \fr{n}{\sqrt{ab}}
        =
        2\sqrt{nd_1d_2}\gamma + 1
        \le 
        3\sqrt{nd_1d_2}\gamma. \qedhere
    \end{align}
\end{proof}

\subsection{Uniform Frobenius bound on \texorpdfstring{$K((\bx,\by)_{\le t})$}{Kxyt}}
\label{sec:doob-offdiag}

In this subsection, we will prove Lemma~\ref{lem:doob-offdiag}.
Let $(\bx,\by)\sim p_0$, $K_t = K((\bx,\by)_{\le t})$ and $X = \sup_{1\le t\le n} \norm{K_t}_F$. 

\begin{lemma}
    \label{lem:doob-aux1-offdiag}
    We have that $\E{X^2}\le 4\E{\norm{K_n}_F^2}$.
\end{lemma}
\begin{proof}
    Analogous to Lemma~\ref{lem:doob-aux1}.
\end{proof}

\begin{lemma}
    \label{lem:doob-aux2-offdiag}
    We have that $\E{\norm{K_n}_F^2} \lesssim nd_1d_2$.
\end{lemma}
\begin{proof}
    We expand
    \[
        \E{\norm{K_n}_F^2}
        =
        \sum_{i=1}^n
        \E*{\norm{\fr{x_iy_i^\dagger}{x_i^\dagger A x_i + y_i^\dagger B y_i}}_F^2}
        +
        2\sum_{1\le i<j\le n}
        \E*{\lt\la
            \fr{x_iy_i^\dagger}{x_i^\dagger A x_i + y_i^\dagger B y_i},
            \fr{x_jy_j^\dagger}{x_j^\dagger A x_j + y_j^\dagger B y_j}
        \rt\ra}
    \]
    The cross terms have expectation $0$ like in the proof of Lemma~\ref{lem:doob-aux2}.
    Now,
    \[
        \E*{\norm{\fr{x_iy_i^\dagger}{x_i^\dagger A x_i + y_i^\dagger B y_i}}_F^2}
        =
        \E*{\fr{\norm{x_i}^2 \norm{y_i}^2 }{(x_i^\dagger A x_i + y_i^\dagger B y_i)^2}}
        \le 
        \fr{1}{b} \E*{\fr{\norm{x_i}^2}{x_i^\dagger A x_i + y_i^\dagger B y_i}}
        = 
        \fr{d_1}{b}
        \lesssim 
        d_1d_2
    \]
    in light of Fact~\ref{fact:d2b}.
\end{proof}

\begin{proof}[Proof of Lemma~\ref{lem:doob-offdiag}]
    Follows from Lemmas~\ref{lem:doob-aux1-offdiag} and \ref{lem:doob-aux2-offdiag}.
\end{proof}

\subsection{Balancedness of realizations}
\label{sec:balanced-offdiag}

% The analysis will be a bit more complicated than in the diagonal case so we first prove a few preliminary claims.
Finally, it remains to prove Lemma~\ref{lem:balanced-offdiag}.  Recall that 
\begin{equation}\label{eq:balanced-def}
\kappa((\bx,\by)) = 
    \sup_{b_1,\ldots,b_n\in [-1,1]}
    \norm{
        \sum_{i=1}^n
        b_i \fr{x_iy_i^\dagger \norm{x_i} \norm{y_i}}{(x_i^\dagger A x_i + y_i^\dagger B y_i)^2}
    }_F \,.
\end{equation}
We will first prove a few preliminary inequalities about the individual terms in the summation above.  In particular, since $x_i,y_i$ are the measurements obtained from a POVM, we argue that over the randomness in the $i$th measurement, the term
\[
\fr{x_iy_i^\dagger \norm{x_i} \norm{y_i}}{(x_i^\dagger A x_i + y_i^\dagger B y_i)^2}
\]
is not too aligned with any given direction.  Thus, intuitively, over $i = 1, \dots , n$, the individual terms will be very weakly correlated and this will allow us to bound the signed sum.  In the next two claims below, we think of $\mathcal{P}$ as a POVM that we will use to measure our state.  

% \bh{flavortext: this is analogous to lemma blah and gives us the bootstrapping bound}
% \bh{I think we want these to be generic weighted sums}
% \bh{$S$ is for event, maybe different letter?}
% \bh{we only need to prove this for $t=n$ now}

\begin{claim}\label{claim:single-POVM1}
Let $\mathcal{P}$ be a set of vectors (not necessarily unit vectors) in $\C^d$ such that 
\[
\sum_{x \in \mathcal{P}} xx^{\dagger} = I \,.
\]
Then for any coefficients $c_x$ for all $x \in \mathcal{P}$, we have
\[
\norm{\sum_{x \in \mathcal{P}}c_x  x} \leq \sqrt{\sum_{x \in \mathcal{P}} c_x^2} \,.
\]
\end{claim}
\begin{proof}
For any unit vector $v$,
\[
\lt\langle v, \sum_{x \in \mathcal{P}}c_x  x \rt\rangle 
\leq 
\sum_{x \in \mathcal{P}} c_x |\langle v, x \rangle | \leq \left( \sqrt{ \sum_{x \in \mathcal{P}} c_x^2}\right) \left( \sqrt{\sum_{x \in \mathcal{P}} |\langle v, x \rangle |^2 } \right) =  \sqrt{ \sum_{x \in \mathcal{P}} c_x^2}
\]
where we used Cauchy-Schwarz and the hypothesis. 
Because $\norm{u} = \max_{\norm{v}=1} \la v,u\ra$ for all vectors $u$, we are done.
\end{proof}

\begin{claim}\label{claim:single-POVM2}
Let $\mathcal{P} = \{(x,y) \}$ be a set of vectors where $x \in \C^{d_1}, y \in \C^{d_2}$ such that 
\[
\sum_{z = (x,y), z \in \mathcal{P}} zz^{\dagger} = I \,.
\]
% Let $A \in \R^{d_1 \times d_1}$ and $B \in \R^{d_2 \times d_2}$ be diagonal matrices with non-negative entries, where all diagonal entries of $B$ are at least $b$.  
Then for any choice of $c_{x,y} \in [-1,1]$ for each $(x,y) \in \mathcal{P}$, we have
\[
\norm{ \sum_{(x,y) \in \mathcal{P}} \frac{c_{x,y} x\norm{x} \norm{y}^2}{x^\dagger A x + y^\dagger B y}} \leq \frac{ \sqrt{d_1}}{b}\,.
\]
\end{claim}
\begin{proof}
For all choices of $c_{x,y} \in [-1,1]$
\[
\norm{ \sum_{(x,y) \in \mathcal{P}} \frac{c_{x,y} x\norm{x} \norm{y}^2}{x^\dagger A x + y^\dagger B y}} 
\leq 
\frac{1}{b} \norm{ \sum_{(x,y) \in \mathcal{P}} c_{x,y} \norm{x} x }_F \,.
\]
By Claim~\ref{claim:single-POVM1},
\begin{align}
    \norm{ \sum_{(x,y) \in \mathcal{P}} c_{x,y} \norm{x} x }_F
    &\le 
    \sqrt{ \sum_{(x,y) \in \mathcal{P}} c_{x,y}^2 \norm{x}^2}
    \le 
    \sqrt{d_1}. \qedhere
\end{align}
\end{proof}

We can think of the bound in Claim~\ref{claim:single-POVM2} as a bound on a single term in  \eqref{eq:balanced-def} over the randomness of the measurement.  Now, we will use Claim~\ref{claim:single-POVM2} on all of the terms in \eqref{eq:balanced-def} to bound the signed sum.   Of course, the sum in the expression in Claim~\ref{claim:single-POVM2} is inside the norm and it is not immediately clear how to use this to reason about the sum in \eqref{eq:balanced-def}.  Actually relating the two expressions requires several additional arguments.

\begin{claim}\label{claim:key-signed-sum}
Consider measuring the state $H_0$ with respect to POVMs $\mathcal{P}_1, \dots , \mathcal{P}_n$ (which may be chosen adaptively).  WLOG, the POVMs are rank-$1$ and can be viewed as sets of vectors such that for all $i \in [n]$
\[
\sum_{z \in \mathcal{P}_i} zz^{\dagger} = I \,.
\]
Let the results of the measurements be $(x_1,y_1), \dots , (x_n, y_n)$.  If $n \leq d_2d_1^{1/2}/(10^{20}\eps^2)$. % and $d_2 a/\eps^2 \geq  10^5$.  
Then with probability $ 1 - e^{-5d_1}$ over the randomness in the measurements, we have the following inequality for any choice of $c_1, \dots , c_n \in [-1,1]$:
\[
\norm{ \sum_{i = 1}^n \frac{c_i x_i y_i^\dagger \norm{x_i} \norm{y_i}}{(x_i^\dagger A x_i+ y_i^\dagger B y_i)^2}}_F \leq \frac{d_1d_2}{4 \cdot 10^3 b \eps^2} \,.
\]
\end{claim}
\begin{proof}

For each $i \in [n]$, define
\[
\theta_i = \frac{ \norm{x_i} \norm{y_i}}{(x_i^\dagger A x_i+ y_i^\dagger B y_i)^2} \,.
\]
Let $Q$ be the expression inside the Frobenius norm on the LHS of the desired inequality.  Define the matrix $X \in \C^{d_1 \times n}$ to have columns given by $x_1, \dots , x_n$ and the matrix $Y \in \C^{d_2 \times n}$ to have columns given by $y_1, \dots , y_n$.  Let $N, D \in \R^{n \times n}$ be the diagonal matrices whose entries are $\norm{y_1}, \dots , \norm{y_n}$ and $c_1\theta_1, \dots , c_t\theta_n$ respectively.
Now we can rewrite
\begin{align*}
\norm{Q}_F^2 = \sum_{i,j} c_ic_j\theta_i\theta_j \la x_i, x_j \ra \la y_i , y_j \ra  = \left\la (XD)^\dagger XD,  Y^\dagger Y \right \ra = \left\la (XDN)^\dagger (XDN) , (YN^{-1})^\dagger (YN^{-1})\right\ra \,.
\end{align*}
%\sitan{should $(N^{-1} Y)^\dagger(N^{-1} Y)$ be $(YN^{-1})^\dagger (YN^{-1})$?}
Now by Grothendieck's inequality, we can replace the expression $(N^{-1}Y)^\dagger (N^{-1}Y)$ with $\sigma^\dagger \sigma$ for some sign vector $\sigma \in \{-1 , 1 \}^n$ while incurring at most a factor of $2$ loss.  Thus,
\begin{equation}\label{eq:grothendieck-reduction}
\max_{c_i}\norm{Q}_F^2 \leq 2 \max_{c_i} \norm{ \sum_{i = 1}^n \frac{c_i x_i \norm{x_i} \norm{y_i}^2}{(x_i^\dagger A x_i+ y_i^\dagger B y_i)^2} }^2 \,.
\end{equation}
Let $R$ denote the quantity inside the norm on the RHS above.  Now it suffices to bound $\norm{R}$.   First, consider any fixed unit vector $v$.  Note that 
\[
\langle R, v \rangle \leq \sum_{i = 1}^n \frac{\norm{x_i} \norm{y_i}^2}{(x_i^\dagger A x_i+ y_i^\dagger B y_i)^2} |\la x_i, v \ra|  \,.
\]
Next, by Claim~\ref{claim:single-POVM2}, we have that  
\[
\bE \left[ \frac{\norm{x_i} \norm{y_i}^2}{(x_i^\dagger A x_i+ y_i^\dagger B y_i)^2} |\la x_i, v \ra| \right] \leq \frac{\sqrt{d_1}}{b}
\]
where the randomness is over the $i$th measurement.  Also note that the individual terms 
\[
\frac{\norm{x_i} \norm{y_i}^2}{(x_i^\dagger A x_i+ y_i^\dagger B y_i)^2} |\la x_i, v \ra| 
\]
are always bounded in magnitude by $1/(ab)$.  Note that the above two observations also imply that 
\[
\bE \left[ \left(\frac{\norm{x_i} \norm{y_i}^2}{(x_i^\dagger A x_i+ y_i^\dagger B y_i)^2} |\la x_i, v \ra|\right)^2 \right] \leq \frac{\sqrt{d_1}}{a b^2} \,.
\]
Thus, by Freedman's inequality, we have
\[
\langle R, v \rangle \leq \frac{d_1d_2}{10^4 b \eps^2}
\]
with failure probability at most
\[
\exp\left( -\frac{1}{2}\frac{\frac{d_1^2d_2^2}{10^8b^2\eps^4}}{\frac{d_1d_2}{10^{20}\eps^2a b^2} + \frac{d_1d_2}{10^4ab^2\eps^2}} \right)  \leq \exp \left( -\frac{d_1d_2 a}{10^5\eps^2 }\right)\leq \exp(-10d_1) \,,
\]
where in the second step we used that $\frac{d_2 a}{\eps^2} \ge \frac{10^{40}}{d_2 b} \ge 10^{40}$.  Finally, we can take some $0.1$-net of possible choices of $v$, which we call $\Gamma$, and union bound over all elements of $\Gamma$.  If the desired inequality were false i.e. 
\[
\norm{ Q}_F \geq \frac{d_1d_2}{4 \cdot 10^3 b \eps^2}
\]
then by \eqref{eq:grothendieck-reduction}, we must have
\[
\norm{R} \geq \frac{d_1d_2}{8 \cdot 10^3 b \eps^2} \,.
\]
By the construction of $\Gamma$, there must be some $v \in \Gamma$ such that 
\[
\langle R, v \rangle \geq 0.9 \norm{R} > \frac{d_1d_2}{10^4 b \eps^2}
\]
and thus we are done.
\end{proof}

\begin{proof}[Proof of Lemma~\ref{lem:balanced-offdiag}]
Plugging in the parameter settings at the beginning of Section~\ref{sec:offdiag}, and using Claim~\ref{claim:key-signed-sum}, we get the desired property.
\end{proof}

\section{Instance Near-Optimal Lower Bounds}
\label{sec:full}

Our main result on general state certification is the following.
Recall that for two quantum states $\sigma, \rho$, the fidelity $F$ between them is defined to be $F(\sigma, \rho) = \Tr \left( \sqrt{\sigma^{1/2} \rho \sigma^{1/2}} \right)^2$.

\begin{theorem}\label{thm:inst_opt}
    Let $0 < \eps < \wt{O}(1/\log\log(d))$. Let $\sigma\in\mathbb{C}^{d\times d}$ be a density matrix. Then any algorithm that uses incoherent measurements which, given $n$ copies of $\rho \in \mathbb{C}^{d \times d}$, can distinguish between the case where $\rho = \sigma$ and $\| \rho - \sigma \|_1 > \eps$ with probability at least $2/3$ must satisfy
    \begin{equation}
        n \geq \Omega\left( \frac{d\sqrt{d_{\mathsf{eff}}}}{\eps^2\polylog(d/\eps)} \cdot F \left( \sigma^*,\rho_\mm \right) \right).
    \end{equation}
    Here, $\sigma^*$ is an explicit density matrix given by zeroing out $O(\eps)$ mass from $\sigma$ and normalizing, and $d_{\mathsf{eff}}$ is the rank of $\sigma^*$ .
\end{theorem}
\noindent As before, the choice of failure probability can be taken to be any constant greater than $1/2$.

In this section we use Theorems~\ref{thm:main_standard} and \ref{thm:main_offdiag} to give a simple proof of a slightly weaker version of Theorem~\ref{thm:inst_opt} where the construction of $\sigma^*$ requires removing up to $O(\eps\log(d/\eps))$ mass. 
The analysis is a simplified version of the analysis from Sections 5.5 and A.3 of \cite{chen2021toward}. Later, in Appendix~\ref{app:block} and \ref{app:full_full} we give a full proof of Theorem~\ref{thm:inst_opt}, which involves slightly generalizing Theorem~\ref{thm:main_standard} and carrying out a more delicate version of the analysis below.

As a first step, notice that since we are given an explicit description of $\sigma$, by applying an appropriate rotation, we may assume without loss of generality that $\sigma$ is diagonal.
For the remainder of this section, we will let $\sigma_1 \geq \ldots \geq \sigma_d$ be its eigenvalues (equivalently, its diagonal entries in sorted order).

\subsection{Bucketing and mass removal}
\label{sec:bucket}

For $j\in\mathbb{Z}_{\ge 0}$, let $S_j$ denote the set of indices $i\in[d]$ for which $2^{-j-1} < \sigma_i \le 2^{-j}$, and define $d_j \triangleq |S_j|$. Let $\calJ$ denote the set of $j$ for which $S_j \neq \emptyset$. We will refer to $j\in\calJ$ as \emph{buckets}. Given $i\in[d]$, let $j(i)$ denote the index of the bucket for which $i\in S_j$.

% Let $\Ssing\subseteq[d]$ denote the set of $i$ belonging to a size-1 bucket, and let $\Smany\subseteq[d]$ denote the set of $i$ belonging to a bucket of size at least 2. 
Let $\calJ^*\subseteq\calJ$ denote the buckets $j$ for which $\sum_{i\in S_j}\sigma_i \ge \eps$, and let $\Slight\subseteq[d]$ denote all $i\in[d]$ for which $j(i) \in\calJ^*$.  Let $\sigma'$ denote the unnormalized density matrix given by zeroing out the $i$-th entry of $\sigma$ for every $i\in\Slight$, and let $\sigma^*$ denote the density matrix $\sigma' / \Tr(\sigma')$.

\begin{fact}\label{fact:fewbuckets}
    $|\calJ^*| \le O(\log(d/\eps))$. In particular, $\sigma'$ is given by removing $O(\eps\log(d/\eps))$ mass from $\sigma$.
\end{fact}

\begin{proof}
    Note that for any $j\in \calJ^*$, $2^{-j} > \sigma_i \ge \eps/d$ for all $i\in S_j$, so $j < \log_2(\eps/d)$.
\end{proof}

\subsection{Helper lemmas}

Here we collect some elementary observations that will be useful in our proof of the weaker version of Theorem~\ref{thm:inst_opt}. We begin by noting an alternative way of representing fidelity with respect to the maximally mixed state.

\begin{fact}\label{fact:fidelity}
    Given psd matrix $\sigma\in\mathbb{C}^{d\times d}$, let $\wh{\sigma}\triangleq \sigma/\Tr(\sigma)$. Then $F(\wh{\sigma},\frac{1}{d} I_d) = \frac{1}{d}\norm{\sigma}_{1/2}\cdot \Tr(\sigma)^{-2}$.
\end{fact}

\begin{fact}\label{fact:halfnorm}
    Let $S$ be any set of distinct positive integers. Given a collection of numbers $\brc{d_j}_{j\in S}$ satisfying $\sum_j d_j 2^{-j} \le 2$, let $p$ be the vector with $d_j$ entries equal to $2^{-j}$ for every $j$. Then $\norm{p}_{1/2} \le |S|^2 \cdot \max_j d^2_j 2^{-j}$.
\end{fact}

\begin{proof}
    Let $j^*\triangleq \arg\max d^2_j 2^{-j}$. Then $\norm{p}^{1/2}_{1/2} = \sum_j d_j 2^{-j/2} \le |S|\cdot d_{j^*} 2^{-j^*/2}$.
\end{proof}

\noindent Our lower bound instances in the proof of the weaker version of Theorem~\ref{thm:inst_opt} will be based on perturbing certain submatrices of $\sigma$. We will use the following basic fact to analyze these instances.

\begin{lemma}\label{lem:pad}
    Consider the task of distinguishing between the following alternatives:
    \begin{equation}
        H_0: \rho = \begin{pmatrix}
            S & 0 \\
            0 & P
        \end{pmatrix} \qquad \text{and} \qquad
        H_1: \rho = \begin{pmatrix}
            \wt{S} & 0 \\
            0 & P
        \end{pmatrix} \label{eq:padC}
    \end{equation}
    where $S\in\R^{d'\times d'}$ and $P\in\R^{(d-d')\times(d-d')}$ are deterministic psd matrices for which $\Tr(S) + \Tr(P) = 1$, and $\wt{S}\in\R^{d'\times d'}$ is drawn from some distribution over psd matrices with trace $\Tr(S)$.
    
    Then the copy complexity of this task using incoherent measurements is $\Omega(n/\Tr(S)) \ge \Omega(n)$, where $n$ is the copy complexity of the following distinguishing task:
    \begin{equation}
        H_0: \rho = S/\Tr(S) \qquad \text{and} \qquad H_1: \rho = \wt{S}/\Tr(S) \label{eq:nopad}
    \end{equation} using incoherent measurements.
\end{lemma}

% proof commented out, probably don't need to include it
% \begin{proof}
%     It is apparent that at every step, the optimal incoherent measurement strategy for the distinguishing task in \eqref{eq:padC} uses some POVM $\brc{F_i}$ which contains a matrix $\Pi$ which is identity in the block corresponding to $C$ and zero elsewhere, and whose other elements are of the form $\begin{pmatrix} 
%         F'_i & 0 \\
%         0 & 0
%     \end{pmatrix}$ for $\sum_i F'_i = I_{d'}$. 
%     Under such a POVM, we obtain no information if we observe the outcome corresponding to $\Pi$, so we can assume without loss of generality that after using such a POVM and observing that outcome, the optimal measurement strategy simply uses the same POVM until it observes any other outcome. The number of times the POVM is used before a nontrivial outcome is observed is $\Tr(S)$ in expectation. This naturally implies a measurement strategy for the task \eqref{eq:nopad} with a corresponding factor $\Tr(S)$ speedup in expected copy complexity.
% \end{proof}

\noindent Theorem~\ref{thm:main_standard} gives a $\Omega(d^{3/2}/\eps^2)$ lower bound for mixedness testing for $d$ larger than some absolute constant. In the following lemma, we complement this with a weaker lower bound that holds for all $d$, based on the classical lower bound for uniformity testing. For this, consider the task of distinguishing between the two alternatives:
\begin{equation}
    H_0: \rho = \frac{1}{d}A \qquad \text{and} \qquad H_1: \rho = \frac{1}{d}\bigl(A + \eps  PZP^{\top}\bigr),
\end{equation}
where $A\in\R^{d\times d}$ is a diagonal matrix with diagonal entries $a_1 \ge \cdots \ge a_d > 0$ and $\Tr(A) = d$, where $Z = \diag(1,\cdots,-1,\cdots)$ if $d$ is even and $Z = \diag(1,\cdots,-1,\cdots,-1,0)$ otherwise, and where $P\in\R^{d\times d}$ is a random permutation matrix on the first $2\lfloor d/2 \rfloor$ coordinates.

\begin{lemma}\label{lem:constantd}
    For all $d > 1$ and $\eps < 1$, the copy of complexity of distinguishing between $H_0$ and $H_1$ with incoherent measurements is $\Omega(\sqrt{d}/\eps^2)$.
\end{lemma}

\begin{proof}
    By Lemma~\ref{lem:pad}, it suffices to prove the lemma when $d$ is even. As the states under $H_0$ and $H_1$ are both diagonal, we can assume without loss of generality that the measurements are all in the standard basis. Let $p_0$ denote the uniform distribution over $[d]$. Given $S\subset[d]$ of size $d/2$, let $p_S$ denote the discrete distribution over $[d]$ which places mass $\frac{1 +\eps}{d}$ on elements in $S$ and mass $\frac{1 - \eps}{d}$ on elements in $[d]\backslash S$. Under $H_0$, if one measures $n$ copies of $\rho$, the $n$ measurement outcomes are a sample from $p_0^{\otimes n}$. Under $H_1$, if one measures $n$ copies of $\rho$, the measurement outcomes are a sample from $\E[S]{p_S^{\otimes n}}$ where $S$ is a random subset of $[d]$ of size $d/2$. It is a standard result in distribution testing that distinguishing between $d_{\mathrm{TV}}(p_0^{\otimes n},\E[S]{p_S^{\otimes n}}) = o(1)$ if $n = o(\sqrt{d}/\eps^2)$ (see e.g. the proof of \cite[Theorem 24.1]{wu2017lecture} which is based on \cite{paninski2008coincidence}).
\end{proof}

\noindent Finally, we will use the following lower bound to handle a corner case where $\sigma$ has one especially large eigenvalue.

\begin{lemma}[Lemma 5.24 from \cite{chen2021toward}]\label{lem:corner}
    Let $\eps \le 1/2$. If $\sigma_1 \ge 3/4$, then state certification to error $\eps$ with respect to $\sigma$ using incoherent measurements is $\Omega(1/\eps^2)$.
\end{lemma}

\subsection{Proof of weaker variant of Theorem~\ref{thm:inst_opt}}

We now give a simple proof of a slight weakening of Theorem~\ref{thm:inst_opt} where one removes $O(\eps\log(d/\eps))$ mass from $\sigma$, instead of $O(\eps)$. We strengthen this analysis in Appendix~\ref{app:full_full}.

\begin{proof}
    Note that $\Tr(\sigma') \ge 1 - O(\eps\log(d/\eps)) \ge \Omega(1)$, so by Fact~\ref{fact:fidelity} it suffices to lower bound the copy complexity by 
    \begin{equation}
        \Omega\left(d_{\mathsf{eff}} \norm{\sigma'}_{1/2} / (\eps^2\log(d/\eps))\right).
    \end{equation}
    
    We proceed by casework depending on whether or not $d_j = 1$ for all $j\in \calJ^*$.
    
    \noindent\textbf{Case 1.} $d_j = 1$ for all $j\in\calJ^*$. Note that in this case,
    \begin{equation}
        \norm{\sigma'}^{1/2}_{1/2} = \sum_{j\in\calJ^*} 2^{-j/2} = O(1).
    \end{equation} and $\norm{\sigma^*}_{1/2} = \Theta(\norm{\sigma'}_{1/2})$. As $d_{\mathsf{eff}} = 1$, it thus suffices to show a copy complexity lower bound of $\Omega(1/\eps^2)$ in this case.
    
    If additionally we have $|\calJ^*| = 1$, then for $\eps \le \wt{O}(1/\log d)$ sufficiently small, the maximum entry of $\sigma$ is at least $3/4$, so we can apply Lemma~\ref{lem:corner} to obtain a lower bound of $\Omega(1/\eps^2)$ as desired.
    
    Otherwise, let $j,j'$ be the two smallest bucket indices in $\calJ^*$, and let $i,i'$ be the elements of the singleton sets $S_j, S_{j'}$. If $\eps \le c2^{-j/2-j'/2-1}$ for sufficiently small constant $c > 0$, we can invoke \cite[Lemma A.4]{chen2021toward} to conclude a copy complexity lower bound of $\Omega(1/\eps^2)$.\footnote{Note that \cite[Lemma A.4]{chen2021toward} gives a (suboptimal) lower bound for the distinguishing task in Section~\ref{sec:offdiag}. The reason we invoke it instead of Theorem~\ref{thm:main_offdiag} is that unlike the latter, it holds for the setting $d_j = d_{j'} = 1$ that we consider here.}
    % Let $P$ denote the $(d-2)\times(d-2)$ submatrix of $\sigma$ containing the remaining diagonal entries of $\sigma$. Then consider the distinguishing task
    % \begin{equation}
    %     H_0: \rho = \begin{pmatrix}
    %         \sigma_i & 0 & 0 \\
    %         0 & \sigma_{i'} & 0 \\
    %         0 & 0 & P
    %     \end{pmatrix} (= \sigma) \qquad \text{and} \qquad 
    %     H_1: \rho = \begin{pmatrix}
    %         \sigma_i & \eps \overline{g} & 0 \\
    %         \overline{g} & \sigma_{i'} & 0 \\
    %         0 & 0 & P
    %     \end{pmatrix},
    % \end{equation}
    % where $\overline{g}$ is sampled from $\calN(0,1)$ conditioned on having magnitude at most $C$. By Lemma~\ref{lem:pad} and Theorem~\ref{thm:main_offdiag}, this has copy complexity at least $\Omega(1/\eps^2)$.
    
    Otherwise, suppose $\eps > c2^{-j/2-j'/2-1}$. Because $2^{-j} > 2^{-j'}$, we know that $2^{-j'} \le O(\eps)$. In particular, consider the state $\sigma^{**}$ given by zeroing out $\sigma_{i'}$ from $\sigma'$ and normalizing. For this matrix, $d_{\mathsf{eff}} = 1$ and $\norm{\sigma^{**}}_{1/2} = O(1)$. Furthermore, because $\eps \le \wt{O}(1/\log(d))$, we have $\sigma_i \ge 3/4$, so we can apply Lemma~\ref{lem:corner} to conclude a lower bound of $\Omega(1/\eps^2)$.
    
    \noindent\textbf{Case 2.} $d_j > 1$ for some $j\in\calJ^*$. In this case, let $j_1 \triangleq \arg\max_{j\in\calJ^*} d_j$ and $j_2\triangleq \arg\max_{j\in\calJ^*} d^2_j 2^{-j}$. 
    
    If $\eps \le cd_{j_2} 2^{-j_1/2 - j_2/2-1} / j_1$ for sufficiently small constant $c > 0$, we can apply the lower bound instance in Section~\ref{sec:offdiag} to these two buckets. Let $P$ denote the submatrix of $\sigma$ containing the diagonal entries of $\sigma$ outside of $S_{j_1}\cup S_{j_2}$. Let $\Sigma_{j_1}$ and $\Sigma_{j_2}$ denote the submatrices of $\sigma$ containing the diagonal entries indexed by $S_{j_1}$ and $S_{j_2}$. 
    
    If $j_1\neq j_2$, then consider the distinguishing task
    \begin{equation}
        H_0: \rho = \begin{pmatrix}
            \Sigma_{j_1} & 0 & 0 \\
            0 & \Sigma_{j_2} & 0 \\
            0 & 0 & P
        \end{pmatrix} (= \sigma) \qquad \text{and} \qquad
        H_1: \rho = \begin{pmatrix}
            \Sigma_{j_1} & \frac{\eps}{d_{j_1}^{1/2} d_{j_2}} \overline{G} & 0\\
            \frac{\eps}{d_{j_1}^{1/2} d_{j_2}}\overline{G}^{\top} & \Sigma_{j_2} & 0\\
            0 & 0 & P
        \end{pmatrix},
    \end{equation}
    where $\overline{G}$ is a $C$-truncated $d_{j_1}\times d_{j_2}$ Ginibre matrix. 
    If $d_{j_1}$ is sufficiently large that Theorem~\ref{thm:main_offdiag} applies, then by Lemma~\ref{lem:pad} and Theorem~\ref{thm:main_offdiag}, this has copy complexity at least \begin{equation}
        \Omega\left(\frac{\sqrt{d_{j_1}}\cdot d_{j_2}}{(\eps/(d_{j_1}2^{-j_1} + d_{j_2}2^{-j_2}))^2}\right) \ge \Omega(\sqrt{d_{j_1}} \cdot d^2_{j_2} 2^{-j_2}/\eps^2) \ge \Omega(d_{\mathsf{eff}}\norm{\sigma'}_{1/2} / (\eps^2\polylog(d/\eps))),
    \end{equation}
    where in the first step we used that $d_{j_2}2^{-j_2} \le 2$, and in the last step we used Fact~\ref{fact:fewbuckets} and Fact~\ref{fact:halfnorm}. Note that $\norm{\sigma^*}_{1/2} = \Theta(\norm{\sigma'}_{1/2})$. Otherwise, if $d_{j_1} = O(1)$ and Theorem~\ref{thm:main_offdiag} doesn't apply, we can still apply \cite[Lemma A.6]{chen2021toward} which only differs in its suboptimal dependence of $d^{1/3}_{j_1}$ on the parameter $d_{j_1}$, which does not affect our overall bound as $d_{j_1} = O(1)$ in this case.
    
    If $j_1 = j_2$, then because we are in Case 2 we know $d_{j_1} > 1$, so let $\Sigma^{(1)}_{j_1}$ and $\Sigma^{(2)}_{j_1}$ denote an arbitrary partition of $\Sigma_{j_1}$ into two $\frac{d_1}{2}\times \frac{d_1}{2}$ diagonal submatrices. Consider the distinguishing task
    \begin{equation}
        H_0: \rho = \begin{pmatrix}
            \Sigma^{(1)}_{j_1} & 0 & 0 \\
            0 & \Sigma^{(2)}_{j_1} & 0 \\
            0 & 0 & P
        \end{pmatrix} (= \sigma) \qquad \text{and} \qquad
        H_1: \rho = \begin{pmatrix}
            \Sigma^{(1)}_{j_1} & \frac{\eps}{d_{j_1}^{1/2} d_{j_2}} \overline{G} & 0\\
            \frac{\eps}{d_{j_1}^{1/2} d_{j_2}}\overline{G}^{\top} & \Sigma^{(2)}_{j_1} & 0\\
            0 & 0 & P
        \end{pmatrix},
    \end{equation}
    where $\overline{G}$ is a $C$-truncated $\frac{d_{j_1}}{2}\times \frac{d_{j_1}}{2}$ Ginibre matrix. 
    If $d_{j_1}$ is sufficiently large that Theorem~\ref{thm:main_offdiag} applies, then by Lemma~\ref{lem:pad} and Theorem~\ref{thm:main_offdiag}, this has copy complexity at least 
    \begin{equation}
        \Omega\left(\frac{\sqrt{d_{j_1}}\cdot d_{j_1}}{(\eps/d_{j_1}2^{-j_1} + d_{j_2} 2^{-j_2})^2}\right) \ge \Omega(\sqrt{d_{j_1}}\cdot d^2_{j_1} 2^{-j_1}/\eps^2) \ge \Omega(d_{\mathsf{eff}} \norm{\sigma'}_{1/2} / (\eps^2\polylog(d/\eps))),
    \end{equation}
    where in the first step we used that $d_{j_1}2^{-j_1} \le 2$, and in the last step we used Fact~\ref{fact:fewbuckets} and Fact~\ref{fact:halfnorm}. Otherwise, if $d_{j_1} = O(1)$, we can apply \cite[Lemma A.4]{chen2021toward} as above.
    
    It remains to consider the case where $\eps > c2^{-j_1/2-j_2/2-1} / j_1$. Let $j^*\triangleq \arg\max_{j\in\calJ^*} d_j 2^{-5j/2}$. We can apply the lower bound instance in Section~\ref{sec:paninski} to bucket $j^*$. Letting $\Sigma_{j^*}$ denote the submatrix of $\sigma$ containing the diagonal entries of $S_{j^*}$ and $P$ denote the submatrix containing the remaining diagonal entries, we consider the distinguishing task
    \begin{equation}
        H_0: \rho = \begin{pmatrix}
            \Sigma_{j^*} & 0 \\
            0 & P
        \end{pmatrix} (= \sigma) \qquad \text{and} \qquad 
        \rho = \begin{pmatrix}
            \Sigma_{j^*} + \frac{\eps}{d_{j^*}} \cdot \overline{M} & 0 \\
            0 & P
        \end{pmatrix},
    \end{equation}
    where $\overline{M}$ is a $C$-truncated trace-centered GOE matrix. By Lemma~\ref{lem:pad} together with either Theorem~\ref{thm:main_standard} if $d \gg 1$ or Lemma~\ref{lem:constantd} if $d = O(1)$, this has copy complexity at least
    \begin{equation}
        \Omega\left(\frac{d^{3/2}_{j^*}}{(\eps/(d_{j^*}2^{-j^*}))^2}\cdot \frac{1}{\Tr(\Sigma_{j^*})}\right) = \Omega(d^{5/2}_{j^*} 2^{-j^*}/\eps^2) \ge \Omega(d^{5/2}_{j_1} 2^{-j_1}/\eps^2).
    \end{equation}
    To complete the proof of the theorem, it suffices to show that
    \begin{equation}
        d^{5/2}_{j_1} 2^{-j_1} \polylog(d/\epsilon) \ge \Omega\left(\sqrt{d_{j_1}} d^2_{j_2} 2^{-j_2} \right)
    \end{equation}
    Suppose to the contrary. Then we would get
    \begin{equation}
        d^2_{j_1} 2^{-j_1} \polylog(d/\epsilon) = o\left(d^2_{j_2} 2^{-j_2}\right). \label{eq:littleo}
    \end{equation}
    But by assumption on $\eps$,
    \begin{equation}
        cd_{j_2} 2^{-j_1/2-j_2/2-1} / j_1 \le \eps \le d_{j_1}2^{-j_1},
    \end{equation}
    where in the last step we used the fact that $d_j 2^{-j} \ge \Omega(\eps)$ for any $j\in\calJ^*$. Squaring both sides and rearranging, we find that 
    \begin{equation}
        d^2_{j_2} 2^{-j_2} \le O(d^2_{j_1} 2^{-j_1} j^2_1) \le O(d^2_{j_1} 2^{-j_1} \log^2(d/\epsilon)),
    \end{equation}
    where the last step follows by the fact that $j_1 \le \log(d/\epsilon)$ because $2^{-j_1}d \ge 2^{-j_1}d_{j_1} \ge \epsilon$, contradicting \eqref{eq:littleo}.
\end{proof}

\paragraph{Acknowledgments.} SC and JL would like to thank Jordan Cotler, Hsin-Yuan Huang, and John
Wright for many illuminating discussions on mixedness testing. Part of this work was completed
while SC and BH were visiting the Simons Institute for the Theory of Computing. The authors thank Oufkir Aadil for pointing out a bug in the proofs of Claims 6.10 and 7.13 in an earlier version of this manuscript.

\bibliographystyle{alpha}
\bibliography{biblio}

\appendix

\section{Multi-Block Distinguishing Task}
\label{app:block}
% \sitan{gonna attempt to generalize Section 4 to multi-block case so we don't lose the $\log(d/\eps)$ factor in how much mass we need to throw out for the instance-optimal result, let's see how this goes... will delete if it gets out of hand}

The proofs from the preceding sections imply a slightly weaker version of Theorem~\ref{thm:inst_opt} where the lower bound involves removing $\eps\log(d/\eps)$ mass from $\sigma$. Avoiding the extra log factor requires working with a slightly more involved instance than the one from Section~\ref{sec:paninski} in which diagonal blocks of $\sigma$ at \emph{all} scales are perturbed.

To that end, here we analyze the following more general distinguishing task.
\begin{equation}
    \label{blockeq:paninski}
    H_0: \rho = \frac{1}{d} \begin{pmatrix}
        A_1 &  &  \\
         & \ddots &  \\
         &  & A_m
    \end{pmatrix}
    \qquad 
    \text{and} 
    \qquad 
    H_1: \rho = \frac{1}{d} \begin{pmatrix}
        A_1 + \eps_1 \oM_1 & &  \\
        & \ddots & \\
         & & A_m + \eps_m \oM_m
    \end{pmatrix}.
\end{equation}
Here, each $A_\nu \in \bR^{d_\nu\times d_\nu}$ is a diagonal matrix. There exist numbers $j_1,\ldots,j_m$ such that each $A_{\nu}$ has diagonal entries in the interval $[d\cdot 2^{-j_\nu},d\cdot 2^{-j_\nu+1}]$, and $\sum_\nu \Tr(A_\nu) = d$.
Furthermore, for every $\nu\in[m]$, $\overline{M}_\nu \sim \sGOE_{U_\nu}(d_\nu)$ for the events $U_{\nu} \triangleq U_{d_\nu}$ given by Lemma~\ref{blocklem:goe-trunc} below.

We will refer to the set of $d_\nu$ row/column indices of $\rho$ which correspond to $A_\nu$ as $\calB_\nu$. Let $\S_\nu$ denote the set of unit vectors in $\mathbb{C}^{d-1}$ with entries supported on $\calB_\nu$.

The following lemma follows easily from the proof of Lemma~\ref{lem:goe-trunc}.
\begin{lemma}
    \label{blocklem:goe-trunc}
    There is an absolute constant $a^* > 0$ such that for any integer $d'\ge a^*$, there exists $U_{d'}\subseteq \bR^{d'\times d'}$ such that if $M\sim \sGOE(d')$, then $\Pr{M\not\in U_{d'}} \le o(1/m)$ and on the event $M\in U_{d'}$, we have $\norm{M}_{\op} \le 3 + \Theta(\sqrt{\log(m)/d'})$ and $\norm{M}_1 \ge d'/12$.
\end{lemma}

\noindent Our main result for the distinguishing task \eqref{blockeq:paninski} is the following.

\begin{theorem}\label{blockthm:main_standard}
    Let $\eps \triangleq \frac{1}{d}\sum_\nu d_\nu \eps_\nu$ and $N\triangleq \frac{1}{m}\min_{\nu\in[m]}\frac{d^{1/2}_\nu d^2}{\eps^2_\nu 2^{j_\nu}}$. For $a^*$ from Lemma~\ref{blocklem:goe-trunc}, if $N \ge 1/\epsilon$,  $d_\nu \ge a^*$ for all $\nu$, and
    \begin{equation}
        \eps_\nu \le d\cdot 2^{-j_\nu}/(12 + \Theta(\sqrt{\log(m)/d_\nu})) \qquad \text{and} \qquad d_\nu / 2^{j_\nu} \ge 2\eps/\log(d/\eps) \ \ \forall \ \nu\in[m] \label{blockeq:eps_assume}
    \end{equation} 
    then the copy complexity of distinguishing between $H_0$ and $H_1$ with incoherent measurements with success probability at least $2/3$ is $\wt{\Omega}(N)$. %\sitan{we will take $\lambda = 2/\log(d/\eps_\nu)$ in the instance-optimal result, but for mixedness testing we think of $\lambda$ as $\Theta(1)$.}
\end{theorem}

% \TODO{state as a corollary that this shows that mixedness testing has copy complexity $\Theta(d^{3/2}/\eps^2)$}

\noindent Note that the bounds in Lemma~\ref{blocklem:goe-trunc} and the first part of \eqref{blockeq:eps_assume} ensure that under $H_1$, $\rho$ is psd (and thus a valid quantum state) and has trace distance at least $\Omega(\frac{1}{d}\sum_\nu d_\nu \eps_\nu)$ to the null hypothesis.

\paragraph{Block structure of POVMs.}

Take any learning tree $\calT$ corresponding to an algorithm for this task that uses $n$ incoherent measurements. The following lemma shows that we can assume without loss of generality that every POVM respects the block structure in the distinguishing task, that is, it consists of $\omega_x d\cdot xx^{\dagger}$ for which each $x\in\S_\nu$ for some $\nu\in[m]$. 

\begin{lemma}\label{lem:POVM_block}
    Given an arbitrary $d$-dimensional POVM $\brc{E_x}$, there is a corresponding rank-1 POVM $\brc{E'_y}$ satisfying the following. Let $p, p'$ be the distributions over measurement outcomes from measuring a state $\rho$ with these POVMs respectively. Then:
    \begin{itemize}
        \item For every $E'_y$, there is some $\nu$ such that $E'_y$ is zero outside the principal submatrix indexed by $\calB_\nu$.
        \item There is an explicit function $f$ mapping outcomes $x$ of the former POVM to outcomes of the latter for which the pushforward of $p'$ under $f$ is $p$.
    \end{itemize}
\end{lemma}

\begin{proof}
    This immediately follows from \cite[Lemma 5.6]{chen2021toward} and the fact that we can always assume without loss of generality that POVMs are rank-1.
\end{proof}

\noindent Given any $x\in\S_\nu$, we use the notation $\nu(x)$ to denote the $\nu\in[m]$ for which $x\in\S_\nu$.
% Define $\omega_\nu \triangleq \sum_{x:\nu(x) = \nu} \omega_x$ and $\omega_{\nu,x} \triangleq \omega_x / \omega_{\nu,x}$. 
Observe that for any $\nu\in[m]$,
\begin{equation}
    \sum_{x: \nu(x) = \nu} \omega_x (xx^{\dagger} - I_{d_\nu}) = 0 \label{eq:resolution}
\end{equation}

The fact that we can assume every POVM respects the block structure will allow our proof to proceed along very similar lines to that of Theorem~\ref{thm:main_standard}, the key distinction being that instead of tracking the overall likelihood ratio, we track one likelihood ratio for each set of coordinates $\calB_\nu$.

Recalling the terminology from Definition~\ref{def:tree}, we let $p_0$ and $p_1$ denote the distributions over leaves of $\calT$ induced by $\rho$ under $H_0$ and $H_1$ respectively. 
In the rest of this section, let $\xi$ be a slowly-growing function satisfying $\xi \gg \log^c(d/\eps)$ for some absolute constant $c > 0$. We assume %1\ll \xi \ll \min_\nu \frac{d^2d^{1/2}_\nu}{\eps^2_\nu 2^\nu m^{1/2}}$. We assume 
\begin{equation}
    n \llsim \frac{1}{\xi m} \cdot \min_{\nu\in[m]} \frac{d^{1/2}_\nu d^2}{\eps^2_\nu 2^{j_\nu}}. \label{eq:nubd}
\end{equation}
and will prove $\TV(p_0,p_1) = o(1)$. 
By the hypothesis in Theorem~\ref{blockthm:main_standard} that $N \ge 1 / \eps$, we may also assume 
\begin{equation}
    n\ge 1 / \eps \label{eq:nlbd}
\end{equation}
by adding superfluous measurements to the algorithm.
We set $\alpha, \beta$ to be slowly-growing functions such that $m^{1/2}\xi^{1/2} \ll \alpha \ll \beta \ll \min_\nu \frac{d^{1/2}_\nu d^2}{\eps^2 2^{j_\nu} n}$. % \wedge \brc{\min_{\nu} (n^{1/2}2^{-\nu/2}d/d_\nu)\cdot m\xi}$.
Note that these choices are possible by \eqref{eq:nubd} and \eqref{eq:nlbd}.

We let $L^*(\cdot)$ denote the likelihood ratio between $p_1$ and $p_0$. 
That is, for a sequence of vectors $\bx = (x_1,\ldots,x_n)$, let $L^*(\bx) \triangleq p_1(\bx)/p_0(\bx)$. 
For $\oM_1\sim\sGOE_{U_1}(d_1), \ldots, \oM_m\sim\sGOE_{U_m}(d_m)$, note that
% \begin{equation}
%     \label{blockeq:def-Lstar}
%     L^*(\bx) 
%     = 
%     \E[\oM_1,\ldots,\oM_m]*{
%         \prod^n_{i=1} \lt(1 + \eps_{\nu(x_i)} \fr{x^{\dagger}_i \oM_{\nu(x_i)} x_i}{x^{\dagger}_i A_{\nu(x_i)} x_i}
%         \rt)
%     }. 
% \end{equation}
because $\oM_1,\ldots,\oM_m$ are independent,
$L^*(\bx) = \prod^m_{\nu = 1} L^*_{\nu}(\bx)$ where 
\begin{equation}
    L^*_{\nu}(\bx) \triangleq \E[\oM_\nu]*{\prod_{i\in[n]: \nu(x_i) = \nu} \left(1 + \eps_\nu \frac{x^{\dagger}_i \oM_{\nu} x_i}{x^{\dagger}_i A_\nu x_i}\right)}.
\end{equation}
For $M_1\sim \sGOE(d_1),\ldots,M_m\sim\sGOE(d_m)$, define similarly $L(\bx) = \prod^m_{\nu = 1} L_\nu(\bx)$ for
% \begin{equation}
%     \label{blockeq:def-L}
%     L(\bx) 
%     = 
%     \E[M_1,\ldots,M_m]*{
%         \prod^n_{i=1} \lt(1 + \eps_{\nu(x_i)} \fr{x^{\dagger}_i M_{\nu(x_i)} x_i}{x^{\dagger}_i A_{\nu(x_i)} x_i}
%         \rt)
%     }.
% \end{equation}  As $M_1,\ldots,M_m$ are independent, we similarly have that $L(\bx) = \prod^m_{\nu = 1} L_{\nu}(\bx)$ for
\begin{equation}
    L_{\nu}(\bx) \triangleq \E[M_\nu]*{\prod_{i\in[n]: \nu(x_i) = \nu} \left(1 + \eps_\nu \frac{x^{\dagger}_i \oM_{\nu} x_i}{x^{\dagger}_i A_\nu x_i}\right)}. \label{eq:Lnu}
\end{equation}

Note that $L(\bx)$ (resp. $L_\nu(\bx)$) is an estimate for the likelihood ratio $L^*(\bx)$ (resp. $L^*_\nu(\bx)$) where the conditioned Gaussian integral is replaced by a true Gaussian integral.
Most of the computations in this section will be done in terms of $L$ instead of $L_\nu$; the proof of Theorem~\ref{blockthm:main_standard} below quantifies that $L(\bx)$ is a close approximation of $L^*(\bx)$.

As before, we will somewhat abuse notation and write $L(\bz)$ for any sequence of unit vectors $\bz = (z_1,\ldots,z_t)$ of length not necessarily $n$, that is, $L(\bz) = \prod^m_{\nu = 1}L_\nu(\bz)$ where $L_\nu(\bz)$ is defined the same way as in \eqref{eq:Lnu}.
% This is defined the same way as in \eqref{blockeq:def-L}.
We also write $L_\nu(\bx,\bx)$ to denote the value of $L_\nu$ on input $(x_1,x_1,x_2,x_2,\ldots,x_n,x_n)$.
% Additionally, given $\bz$, we will use $\bz^{(\nu)}$ to denote the subsequence of $\bz$ consisting of all vectors in $\S_\nu$.
% \bh{flavortext here: this is a ``fake likelihood ratio," $L$ makes sense as a function on any vector sequence, we will slightly abuse notation and write $L(\bx)$ for vector sequence $\bx = (x_1,\ldots,x_t)$ of length not necessarily $n$}

The main ingredient in the proof of Theorem~\ref{blockthm:main_standard} is the following analogue of Proposition~\ref{prop:lstar-reduction} giving a high-probability bound on each $L_\nu$ evaluated at the leaves of $\calT$.

\begin{proposition}
    \label{blockprop:lstar-reduction}
    There exists a subset $S$ of the leaves of $\calT$ such that $\Pr[p_0]*{S} = 1-o(1)$ and for all $\bx \in S$, $|L_\nu(\bx)-1|=o(1/m)$ and $L_\nu(\bx,\bx) \ll e^{\sqrt{d_\nu}}$ for all $\nu\in[m]$.
\end{proposition}

\noindent Let us first prove Theorem~\ref{blockthm:main_standard} assuming Proposition~\ref{blockprop:lstar-reduction}. 

\begin{proof}[Proof of Theorem~\ref{blockthm:main_standard}]
    % Let $U_1,$ be as in Lemma~\ref{blocklem:goe-trunc}.
    Let $U$ denote the event that $M_\nu\in U_\nu$ for all $\nu\in[m]$. Define
    \begin{align}
        \oL(\bx) &\triangleq \E[M_1\sim \sGOE(d_1),\ldots,M_m\sim\sGOE(d_m)]*{
            \Id\{U\}
            \prod_{i=1}^n
            \lt(
                1 + \eps_{\nu(x_i)} \fr{x_i^\dagger M_{\nu(x_i)} x_i}{x_i^\dagger A_{\nu(x_i)} x_i}
            \rt)
        } \\
        &= \prod^m_{\nu = 1} \E[M_\nu\sim\sGOE(d_\nu)]*{\bone{U_\nu} \prod_{i\in[n]: \nu(x_i) = \nu} \left(1 + \eps_\nu \frac{x^{\dagger}_i M_{\nu} x_i}{x^{\dagger}_i A_{\nu} x_i}\right)} \triangleq \prod^m_{\nu = 1} \oL_\nu(\bx).
    \end{align}
    It is clear that $L^*_\nu(\bx) = \Pr*{U_{\nu}}^{-1} \oL_\nu(\bx)$.
    For all $\bx \in S$ and $\nu\in[m]$, by Cauchy-Schwarz
    \begin{align*}
        |L_\nu(\bx) - \oL_\nu(\bx)|
        &= 
        \lt|
            \E[M_{\nu}]*{
                \Id\{U^c_\nu\}
                \prod_{i\in[n]: \nu(x_i) = \nu}
                \lt(
                    1 + \eps_\nu \fr{x_i^\dagger M_\nu x_i}{x_i^\dagger A_\nu x_i}
                \rt)
            }
        \rt| \\
        % &\le
        % \sum_{\emptyset\neq S\subset[m]} 
        % \lt|
        %     \prod_{\nu\in S}\E[M_\nu]*{
        %         \Id\{U^c_{\nu}\}\prod_{i\in[n]: \nu(x_i) = \nu} \frac{x^{\dagger}_iM_{\nu}x_i}{x^{\dagger}_iA_{\nu}x_i}
        %     } \cdot
        %     \prod_{\nu\not\in S}\E[M_\nu]*{
        %         \Id\{U_{\nu}\}\prod_{i\in[n]: \nu(x_i) = \nu} \frac{x^{\dagger}_iM_{\nu}x_i}{x^{\dagger}_iA_{\nu}x_i}
        %     }
        % \rt| \\
        % &\le 
        % \sum_{\emptyset\neq S\subseteq[m]}
        % \prod_{\nu\in S}\sqrt{{U^c_\nu}^{1/2} L(\bx^{(\nu)},\bx^{(\nu)})}\cdot
        % \\
        % &\le 
        % \Pr*{U^c}^{1/2}
        % \E[M_1,\ldots,M_m]*{
        %     \prod_{i=1}^n
        %     \lt(
        %         1 + \eps_{\nu(x_i)} \fr{x_i^\dagger M_{\nu(x_i)} x_i}{x_i^\dagger A_{\nu(x_i)} x_i}
        %     \rt)^2
        % }^{1/2} \\
        &\le 
        \sqrt{\Pr*{U^c} L_\nu(\bx,\bx)}
        =o(1/m).
    \end{align*}
    Here we use that $\Pr*{U^c_\nu} \le \mathop{\textrm{poly}}(1/m)\cdot \exp(-\Omega(d_\nu))$ and $L_\nu(\bx,\bx) \ll e^{\sqrt{d_\nu}}$.
    Moreover, we have $|L_\nu(\bx)-1| = o(1/m)$.
    Thus, for all $\bx \in S$ and $\nu\in[m]$, $\oL_\nu(\bx) = 1+o(1/m)$ and 
    \begin{align*}
        |L^*_\nu(\bx)-1|
        &\le 
        |L^*_\nu(\bx)-\oL_\nu(\bx)|
        +
        |\oL_\nu(\bx) - 1| \\
        &= 
        \fr{\Pr*{U^c_\nu}}{\Pr*{U_\nu}}
        \oL_\nu(\bx) + o(1/m) 
        = o(1/m).
    \end{align*}
    Recalling that $L^*(\bx) = \prod^m_{\nu = 1}L^*_{\nu}(\bx)$, we conclude that $L^*(\bx) = 1 + o(1)$.
    Finally, 
    \begin{align*}
        \TV(p_0,p_1)
        &= 
        2\E[\bx\sim p_0]*{
            (L^*(\bx)-1)_-
        } \\
        &= 
        2\E[\bx\sim p_0]*{
            \Id\{\bx\in S\} (L^*(\bx)-1)_-
        } +
        2\E[\bx\sim p_0]*{
            \Id\{\bx\not\in S\} (L^*(\bx)-1)_-
        } \\
        &\le 
        2\sup_{\bx \in S} (L^*(\bx)-1)_-
        +
        2\Pr[p_0]*{S^c} 
        = 
        o(1). \qedhere
    \end{align*}
\end{proof}

\subsection{Recursive evaluation of likelihood ratio}

Let $\bz = (z_1,\ldots,z_t)$ be a sequence of unit vectors. 
% \sitan{thoughts on maybe changing $\underline{y}$ to $\mathbf{y}$? might make things like $\overline{K}(\underline{y})$ look less weird}
% \bh{yeah sounds good, will do}
For $1\le i\le t$, let $\bz_{\sim i}$ be the sequence $\bz$ with $z_i$ omitted.
Similarly, for $1\le i<j\le t$, let $\bz_{\sim i,j}$ be the sequence $\bz$ with $z_i,z_j$ omitted.
The main result of this subsection is the following recursive formula for $L(\bz)$. 
\begin{lemma}
    \label{blocklem:L-recursion}
    The function $L_\nu$ satisfies
    \[
        L_\nu(\bz) = L_\nu(\bz_{\sim t}) + \fr{2\eps^2_\nu}{d^2_\nu} \cdot \bone{\nu(z_t) = \nu}
        \sum_{i<t : \nu(z_i) = \nu}
        \lt[
            \fr{d_\nu\la z_i,z_t\ra^2-1}{(z_i^\dagger A_\nu z_i)(z_t^\dagger A_\nu z_t)}
            L_\nu(\bz_{\sim i,t})
        \rt].
    \]
\end{lemma}

\noindent As with Lemma~\ref{lem:L-recursion}, the proof is based on Isserlis' theorem. For $k$ even, recall that $\PMat(k)$ denotes the set of perfect matchings of $\{1,\ldots,k\}$.

\begin{proof}[Proof of Lemma~\ref{blocklem:L-recursion}]
    The case of $\nu(z_t) \neq \nu$ is clear. We now suppose $\nu(z_t) = \nu$. Let $J\subseteq[t]$ denote the indices $s$ for which $\nu(x_s) = \nu$. For a set $S\subseteq J$ with $|S|$ even, let $\PMat(S)$ denote the set of perfect matchings of $S$. 
    For even $k$, let $\Mat(J,k)$ denote the set of matchings of $S$ consisting of $k/2$ pairs.
    We compute that
    \begin{align}
        \notag
        L_{\nu}(\bz)
        &= 
        \sum_{S\subseteq J}
        \eps^{|S|}_\nu
        \E[M_\nu\sim \sGOE(d_\nu)]*{\prod_{i\in S}\fr{z_i^\dagger M_\nu z_i}{z_i^\dagger A_\nu z_i}} 
        \qquad \text{(expanding definition of $L$)} \\
        \notag
        &= 
        \sum_{\substack{S\subseteq J \\ |S|~\text{even}}}
        \eps^{|S|}_\nu
        \sum_{\{\{a_1,b_1\},\ldots,\{a_{|S|/2},b_{|S|/2}\}\} \in \PMat(S)}
        \prod_{i=1}^{|S|/2} \E[M_\nu\sim \sGOE(d_\nu)]*{
            \fr{z_{a_i}^\dagger M_\nu z_{a_i}}{z_{a_i}^\dagger A_\nu z_{a_i}} \cdot 
            \fr{z_{b_i}^\dagger M_\nu z_{b_i}}{z_{b_i}^\dagger A_\nu z_{b_i}}
        } 
        \ \text{(Th.~\ref{thm:isserlis})} \\
        \notag
        &= 
        \sum_{k=0}^{\lfloor |J|/2\rfloor}
        \eps^{2k}_\nu
        \sum_{\{\{a_1,b_1\},\ldots,\{a_k,b_k\}\} \in \Mat(J,2k)}
        \prod_{i=1}^{k} \E[M_\nu\sim \sGOE(d_\nu)]*{
            \fr{z_{a_i}^\dagger M_\nu z_{a_i}}{z_{a_i}^\dagger A_\nu z_{a_i}} \cdot 
            \fr{z_{b_i}^\dagger M_\nu z_{b_i}}{z_{b_i}^\dagger A_\nu z_{b_i}}
        } \\
        \label{blockeq:L-expansion}
        &= 
        \sum_{k=0}^{\lfloor |J|/2\rfloor}
        \lt(\fr{2\eps^2_\nu}{d^2}\rt)^{k}
        \sum_{\{\{a_1,b_1\},\ldots,\{a_{k},b_{k}\}\} \in \Mat(J,2k)}
        \prod_{i=1}^{k} 
        \fr{d_\nu\la z_{a_i}, z_{b_i}\ra^2-1}{(z_{a_i}^\dagger A_\nu z_{a_i})(z_{b_i}^\dagger A_\nu z_{b_i})}.
    \end{align}
    % In the final step we use that for unit vectors $x,y \in \C^{d_\nu}$, 
    % \[
    %     \E[M_\nu\sim \sGOE(d_\nu)]*{
    %         (x^\dagger M_\nu x)(y^\dagger M_\nu y)
    %     }
    %     = 
    %     \fr{2}{d^2_\nu}(d_\nu\la x,y\ra^2-1),
    % \]
    % which can be verified by direct computation.
    The lemma follows by partitioning the summands in \eqref{blockeq:L-expansion} based on whether $t$ appears in the matching, and if so which $i\in J$ it is paired with.
\end{proof}

\subsection{High probability bound on likelihood ratio at leaves}
% \bh{remaining subsection titles TBD}

This subsection gives the main part of the proof of Proposition~\ref{blockprop:lstar-reduction}.
For any sequence of vectors $\bz = (z_1,\ldots,z_t)$ and $\nu\in[m]$, define
\[
    H_\nu(\bz) 
    = 
    \sum_{i\le t: \nu(z_i) = \nu}
    \fr{d_\nu z_iz_i^\dagger - I_{d_\nu}}{z_i^\dagger A_\nu z_i} \cdot 
    \fr{L(\bz_{\sim i})}{L(\bz)}
    \qquad
    \text{and}
    \qquad
    K_\nu(\bz) 
    = 
    \sum_{i\le t: \nu(z_i) = \nu}
    \fr{d_\nu z_iz_i^\dagger - I_{d_\nu}}{z_i^\dagger A_\nu z_i}.
\]
$H_\nu$ enters our calculations by the following rewriting of Lemma~\ref{blocklem:L-recursion}:
\begin{equation}
    \label{blockeq:L-ratio-recursion}
    \fr{L_\nu(\bz)}{L_\nu(\bz_{\sim t})}
    = 
    1 + 
    \fr{2\eps^2_\nu}{d^2_\nu}\cdot \bone{\nu(z_t) = \nu} 
    \cdot 
    \fr{z_t^\dagger H_\nu(\bz_{\sim t}) z_t}{z_t^\dagger A_\nu z_t}.
\end{equation}
If $\bz = \bx_{\le t} \triangleq (x_1,\ldots,x_t)$ is a prefix of $\bx \sim p_0$, then $\prod^m_{\nu = 1} \fr{L_\nu(\bz)}{L_\nu(\bz_{\sim t})} = \prod^m_{\nu = 1} \fr{L_\nu(\bx_{\le t})}{L_\nu(\bx_{\le t-1})} = \frac{L(\bx_{\le t})}{L(\bx_{\le t-1})}$ is one step in the likelihood ratio martingale.
We will control the contribution from each multiplicative martingale $L_\nu$ separately. As we will see (proof of Claim~\ref{blockcl:lx-bd}) below, the multiplicative fluctuation of any such step is 
\[
    \E[x_t]*{\lt(\fr{L_\nu(\bx_{\le t})}{L_\nu(\bx_{\le t-1})}\rt)^2} 
    = 1 + \frac{\eps^4_\nu 2^{j_\nu}}{d^2 d^4_\nu} \cdot \norm{H_{\nu}(\bx_{\le t - 1})}_F^2.
    % 1 + O\lt(\fr{\eps^4}{d^5}\rt) \norm{H(\bx_{\le t-1})}_F^2.
\]

Thus, an upper bound on $\norm{H_\nu(\bz)}_F$ over all $\nu\in[m]$ and all prefixes $\bz$ of $\bx$ controls the fluctuations of the likelihood ratio martingale.
Because the matrices output by $H_\nu$ are hard to control directly, we will use the function $K_\nu$ as a proxy for $H_\nu$. 
The following analogue of Lemma~\ref{lem:bootstrap} quantifies this relationship, showing that if $K_\nu(\bz)$ is bounded in Frobenius norm, $H_\nu(\bz)$ is bounded at the same scale.
% The function $K$ will enter our calculations as a proxy for $H$, in a sense made precise by the following lemma.
\begin{lemma}
    \label{blocklem:bootstrap}
    Suppose $\gamma \gg m^{1/2}\xi^{1/2}$.
    If $\bz=(z_1,\ldots,z_t)$ is a sequence of unit vectors satisfying $t \le n$ and $\norm{K_\nu(\bz)}_F \le \left(n^{1/2}\cdot 2^{j_\nu/2} d^{3/2}_\nu / d\right) \gamma$, % n^{1/2} d \gamma$, 
    and the number of $s\in[t]$ for which $\nu(z_s) = \nu$ is at most $n\cdot(d_\nu/2^{j_\nu})\cdot m\xi$, then $\norm{H_\nu(\bz)}_F \le C\left(n^{1/2}\cdot 2^{j_\nu/2} d^{3/2}_\nu / d\right) \gamma$ for some absolute constant $C > 0$. %\le 3 n^{1/2} d \gamma$.
\end{lemma}
\noindent This lemma is a ``deterministic" statement about a sequence of vectors.
We will prove this in Subsection~\ref{blocksec:bound-signed-sum} using the same bootstrap argument from earlier.

Lemma~\ref{blocklem:bootstrap} requires that for every $\nu$, the number of POVM elements supported on the coordinates $\calB_\nu$ is not much greater than its expectation, which we show holds with high probability:

\begin{lemma}\label{lem:bucket_freqs}
    With probability $1 - o(1)$ over $\bx\sim p_0$, for all $\nu\in[m]$ there are at most $n\cdot (d_\nu/2^{j_\nu})\cdot m\xi$ indices $s\in[n]$ for which $\nu(x_s) = \nu$.
\end{lemma}

\begin{proof}
    Take any POVM $\brc{\omega_x d\cdot xx^{\top}}$ where for every $x$ there is some $\nu$ for which $x\in \S_{\nu}$. Now fix $\nu\in[m]$ and note that the probability of observing $x$ for which $\nu(x) = \nu$ upon measuring a copy of $\rho$ under the null hypothesis is
    \begin{equation}
        \sum_{x: \nu(x) = x} \omega_x x^{\dagger} A_\nu x \le d\cdot 2^{-j_\nu+1} \cdot \sum_{x: \nu(x) = x} \omega_x \norm{x}^2 = 2^{-j_\nu + 1}d_\nu,
    \end{equation}
    where in the last step we used \eqref{eq:resolution}. The lemma follows by Markov and a union bound over $\nu\in[m]$.
\end{proof}

Finally, Lemma~\ref{blocklem:bootstrap} also requires a bound on $K_\nu(\bz)$. The following analogue of Lemma~\ref{lem:doob} bounds $K_\nu(\bz)$ in Frobenius norm uniformly over all prefixes $\bz$ of $\bx$. 
We will prove this lemma in Subsection~\ref{blocksec:doob}.

\begin{lemma}
    \label{blocklem:doob}
    If $\bx \sim p_0$, then $\E*{\sup_{1\le t\le n} \norm{K_\nu(\bx_{\le t})}_F^2} \lesssim n\cdot 2^{j_\nu} d^3_\nu / d^2$. % \lesssim nd^2$.
\end{lemma}

% \bh{oops, it's a little unfortunate that $M$ means martingale and matrix... can fix this later. maybe bold M for matrix?}

\noindent We will now prove Proposition~\ref{blockprop:lstar-reduction} assuming Lemmas~\ref{blocklem:bootstrap} and \ref{blocklem:doob}. Let $\bx \sim p_0$.
% For $1\le t\le n$, define $H_t = H(\bx_{\le t})$, $K_t = K(\bx_{\le t})$.
% Further define the filtration  $\cF_t = \sigma(\bx_{\le t})$ and the sequence $\Phi_t = L(\bx_{\le t})$.
For $1\le t\le n$, define the filtration $\cF_t = \sigma(\bx_{\le t})$ and the sequences
\begin{equation}
    H_{\nu,t} = H_\nu(\bx_{\le t}), \qquad K_{\nu,t} = K_\nu(\bx_{\le t}), \qquad \Phi_{\nu,t} = L_\nu(\bx_{\le t}), \qquad \Phi_{t} = L(\bx_{\le t}).
\end{equation}
Consider the times
\begin{multline}
    \tau_\nu = \{\infty\} \cup \inf \Bigg\{
        t : \norm{K_{\nu,t}}_F > (n^{1/2}\cdot 2^{j_\nu/2}d^{3/2}_\nu / d)\alpha \qquad \text{or}\qquad\\
        |s\in[t]: \nu(x_s) = \nu| > n\cdot(d_\nu/2^{j_\nu})\cdot m\xi \qquad \text{or}\qquad  |\Phi_{\nu,t}-1| > n\cdot \fr{\eps^2 2^{j_\nu}}{d^{1/2}_\nu d^2} \beta
    \Bigg\}
\end{multline}
% and
% \begin{equation}
%     \tau \triangleq \left(\inf_{\nu} \tau_{\nu}\right) \wedge \inf\brc*{t: |\Phi_t-1| > n\cdot \left(\sum^m_{\nu=1} \fr{\eps^4 2^{2\nu}}{d_\nu d^4} \beta^2\right)^{-1/2}},
% \end{equation}
which are clearly stopping times with respect to $\cF_t$.
Also define the stopped sequences $\Psi_{\nu,t} = \Phi_{\nu,t\wedge \tau_\nu}$. % and $\Psi_t = \Phi_{t\wedge\tau}$.

\begin{claim}
    \label{blockcl:frob-sup}
    With probability $1-o(1)$, $\norm{K_{\nu,t}}_F \le n^{1/2}\cdot 2^{j_\nu/2}d^{3/2}_\nu/d$ % n^{1/2} d\alpha$ 
    for all $t\in[n]$ and all $\nu\in[m]$.
\end{claim}
\begin{proof}
    By Lemma~\ref{blocklem:doob}, 
    \begin{align}
        \Pr*{\sup_{1\le t\le n} \norm{K_{\nu,t}}_F > n^{1/2}\cdot 2^{j_\nu/2}d^{3/2}_\nu/d}
        &\le 
        \fr{\E*{\sup_{1\le t\le n} \norm{K_{\nu,t}}_F^2}}{(n\cdot 2^{j_\nu}d^{3}_\nu/d^2) \alpha^2}
        \lesssim \alpha^{-2} = o(1/m). \qedhere
    \end{align}
    The claim follows by a union bound over $\nu$.
\end{proof}

\begin{claim}
    \label{blockcl:lx-bd}
    With probability $1-o(1)$, $|\Psi_{n,\nu}-1| \le n\cdot\fr{\eps^2_\nu 2^{j_\nu}}{d^{1/2}_\nu d^2}\beta$ for all $\nu\in[m]$. %\fr{\eps^2 n}{d^{3/2}}\beta$.
\end{claim}
\begin{proof}
    Note that $\Psi_{\nu,t}$ is a multiplicative martingale: if $\tau \le t-1$ then certainly $\E{\fr{\Psi_{\nu,t}}{\Psi_{\nu,t-1}} | \cF_{t-1}}=1$, and if $\tau > t-1$, \eqref{blockeq:L-ratio-recursion} implies 
    \[
        \E*{\fr{\Psi_{\nu,t}}{\Psi_{\nu, t-1}} | \cF_{t-1}}
        = 
        1 + 
        \fr{2\eps^2}{d^2_\nu} \E*{\bone{\nu(x_t) = \nu}\cdot \fr{x_t^\dagger H_{\nu, t-1} x_t}{x_t^\dagger A_{\nu} x_t} | \cF_{t-1}}
        =1,
    \]
    using that 
    \begin{multline}
        \label{blockeq:disc-deg1}
        \E*{\bone{\nu(x_t) = \nu}\cdot \fr{x_t^\dagger H_{\nu, t-1} x_t}{x_t^\dagger A_{\nu} x_t} | \cF_{t-1}}
        = 
        \sum_{x_t: \nu(x_t) = \nu} \omega_{x_t} (x_t^\dagger H_{\nu, t-1} x_t)
        \\ = 
        \lt\la H_{\nu, t-1}, \sum_{x_t: \nu(x_t) = \nu} \omega_{x_t} x_tx_t^\dagger \rt\ra 
        = 
        \la H_{\nu, t-1}, I_{d_{\nu}}/d\ra
        = 0.
    \end{multline}
    % In a similar fashion, we see that $\Psi_t$ is a multiplicative martingale: if $\tau \le t - 1$ then certainly $\E{\frac{\Psi_t}{\Psi_{t-1}} | \cF_{t-1}} = 1$, and if $\tau > t-1$, then because the events $\brc{\bone{\nu(x_t) = \nu}}_{\nu}$ are disjoint, we have
    % \begin{equation}
    %     \E*{\fr{\Psi_t}{\Psi_{t-1}} | \cF_{t-1}} = 1 + \sum^m_{\nu = 1} \fr{2\eps^2}{d^2_\nu} \E*{\bone{\nu(x_t) = \nu}\cdot \fr{x_t^\dagger H_{\nu, t-1} x_t}{x_t^\dagger A_{\nu} x_t} | \cF_{t-1}} = 1.
    % \end{equation}
    
    We next bound the quadratic increment $\E{(\fr{\Psi_{\nu,t}}{\Psi_{\nu,t-1}})^2 | \cF_{t-1}}$.
    If $\tau \le t-1$ this is $1$, and otherwise it is given by
    \begin{equation}
        \label{blockeq:N-ratio-sq}
        % \E*{\lt(\fr{\Psi_{\nu,t}}{\Psi_{\nu,t-1}}\rt)^2 | \cF_{t-1}}
        % = 
        1 + 
        \fr{4\eps^2}{d^2_\nu} \E*{\bone{\nu(x_t) = \nu}\cdot \fr{x_t^\dagger H_{\nu,t-1} x_t}{x_t^\dagger A_{\nu} x_t} | \cF_{t-1}}
        +
        \fr{4\eps^4}{d^4_\nu}
        \E*{\bone{\nu(x_t) = \nu}\cdot \fr{(x_t^\dagger H_{\nu,t-1} x_t)^2}{(x_t^\dagger A_{\nu} x_t)^2} | \cF_{t-1}}.
    \end{equation}
    The first expectation is zero by \eqref{blockeq:disc-deg1}.
    To bound the remaining expectation, note that for any unit vector $x\in\S_\nu$,
    \begin{equation}
        \label{blockeq:xAx-lb}
        x^\dagger A_{\nu} x \ge d\cdot 2^{-j_\nu}.
    \end{equation}
    So, 
    \begin{align}
        \notag
        \E*{\bone{\nu(x_t) = \nu}\cdot \fr{(x_t^\dagger H_{\nu,t-1} x_t)^2}{(x_t^\dagger A_{\nu} x_t)^2} | \cF_{t-1}}
        &\le
        (2^{j_\nu}/d)\E*{\bone{\nu(x_t) = \nu}\cdot \fr{(x_t^\dagger H_{\nu,t-1} x_t)^2}{x_t^\dagger A_{\nu} x_t} | \cF_{t-1}} \\
        &= 
        (2^{j_\nu}/d)\sum_{x_t: \nu(x_t) = \nu}
        \omega_{x_t}
        x_t^\dagger H_{\nu,t-1} (x_t x_t^\dagger) H_{\nu,t-1} x_t 
        \\
        &\le 
        (2^{j_\nu}/d)\sum_{x_t: \nu(x_t) = \nu}
        \omega_{x_t}
        x_t^\dagger H_{\nu,t-1}^2 x_t  \\
        %\label{eq:disc-deg2}
        &= 
        (2^{j_\nu}/d) \lt\la H_{\nu,t-1}^2, \sum_{x_t:\nu(x_t) = \nu} \omega_{x_t} x_tx_t^\dagger \rt\ra
        \\ 
        &= 
        (2^{j_\nu}/d)\la H_{\nu,t-1}^2, I_{d_\nu}/d\ra 
        = 
        \fr{2^{j_\nu}}{d^2}\norm{H_{\nu,t-1}}_F^2.
    \end{align}
    Moreover, since $\tau > t-1$, $\norm{K_{\nu, t-1}}_F \le (n^{1/2}\cdot 2^{j_\nu/2}d^{3/2}_\nu/d)\alpha$ %n^{1/2}d\alpha$
    and Lemma~\ref{blocklem:bootstrap} implies $\norm{H_{\nu,t-1}}_F \le (Cn^{1/2}\cdot 2^{j_\nu/2}d^{3/2}_\nu/d)\alpha$. % 3n^{1/2}d\alpha$.  
    Thus,
    \[
        \E*{\lt(\fr{\Psi_{\nu,t}}{\Psi_{\nu,t-1}}\rt)^2 | \cF_{t-1}}
        \le 
        1 + \fr{4\eps^4_\nu 2^{j_\nu}}{d^2 d_{\nu}^4} \norm{H_{\nu,t-1}}_F^2
        \le 
        1 + \fr{36\eps^4_\nu 2^{2j_\nu}\cdot n}{d_\nu d^4}\alpha^2.
    \]    
    So, for all $1\le t\le n$,
    \[
        \E{\Psi_{\nu,t}^2}
        =
        \E*{\E*{\lt(\fr{\Psi_{\nu, t}}{\Psi_{\nu, t-1}}\rt)^2 | \cF_{t-1}} \Psi_{\nu, t-1}^2}
        \le 
        \lt(1 + \fr{36\eps^4_\nu 2^{2j_\nu}\cdot n}{d_\nu d^4}\alpha^2\rt)
        \E{\Psi_{\nu,t-1}^2},
    \]
    and therefore
    \[
        \E{\Psi_{\nu,t}^2}
        \le 
        \lt(1 + \fr{36\eps^4_\nu 2^{2j_\nu}\cdot n}{d_\nu d^4}\alpha^2\rt)^n
        \le 
        \exp\lt(\fr{36\eps^4_\nu 2^{2j_\nu}\cdot n^2}{d_\nu d^4}\alpha^2\rt)
        \le 2
    \]
    since $\fr{\eps^4_\nu 2^{2j_\nu}\cdot n^2}{d_\nu d^4}\alpha^2 \ll 1$ by assumption.
    
    Moreover, 
    \begin{align*}
        \E{(\Psi_{\nu,t}-1)^2}
        &=
        \E*{\E*{\lt(\fr{\Psi_{\nu,t}}{\Psi_{\nu,t-1}}\rt)^2 | \cF_{t-1}} \Psi_{\nu,t-1}^2 - 2 \E*{\fr{\Psi_{\nu,t}}{\Psi_{\nu,t-1}} | \cF_{t-1}}\Psi_{\nu,t-1} + 1} \\
        &\le 
        \fr{36\eps^4_\nu 2^{2j_\nu}\cdot n}{d_\nu d^4}\alpha^2\cdot \E{\Psi_{\nu,t-1}^2} + \E{(\Psi_{\nu,t-1}-1)^2} \\
        &\le 
        \fr{72\eps^4_\nu 2^{2j_\nu}\cdot n^2}{d_\nu d^4}\alpha^2 + \E{(\Psi_{\nu,t-1}-1)^2},
    \end{align*}
    so by induction
    \[
        \E{(\Psi_{\nu,n}-1)^2}
        \le 
        \fr{72\eps^4_\nu 2^{2j_\nu}\cdot n^2}{d_\nu d^4}\alpha^2.
    \]
    Thus
    \[
        \Pr*{|\Psi_{\nu,n}-1| > n\cdot\left(\fr{72\eps^4_\nu 2^{2j_\nu}}{d_\nu d^4}\beta^2\right)^{1/2}}
        \le 
        \fr{72\alpha^2}{\beta^2} =o(1).
    \]
    Therefore, $|\Psi_{\nu,n}-1|\le n\cdot\fr{\eps^2_\nu 2^{j_\nu}}{d^{1/2}_\nu d^2}\beta$ with probability $1-o(1)$.
\end{proof}

% \bh{todo rewrite this with allen and sitan's simpler argument} \sitan{took a pass at writing this up} \bh{looks good, thanks! much cleaner than what I was doing before lol}
\begin{claim}
    \label{blockcl:lxx-bd}
    If $\norm{K_{\nu,n}}_F \le (n^{1/2}\cdot 2^{j_\nu/2}d^{3/2}_\nu/d)\alpha$, then $L_\nu(\bx,\bx) \ll e^{\sqrt{d_\nu}}$.
\end{claim}
\begin{proof}
    Using the elementary inequality $e^z\le 1 + z$, we can upper bound $L_\nu(\bx,\bx)$ by
    \begin{equation}
        L_\nu(\bx,\bx) \le \mathop{\mathbb{E}}_{M\sim\sGOE(d_\nu)}\biggl[\exp\biggl(\biggl\langle 2\eps_\nu M, \sum_{i\in[n]: \nu(x_i) = \nu} \frac{x_ix^{\dagger}_i}{x^{\dagger}_i A_\nu x_i}\biggr\rangle\biggr)\biggr] = \E[M]*{\exp\left(\frac{2\eps_\nu}{d_\nu}\iprod*{M, K_{\nu,n}}\right)}, \label{blockeq:lux2}
    \end{equation}
    where in the second step we used that $\Tr(M) = 0$. As $M = G - \frac{\Tr(G)}{d}I_{d_\nu}$ for $G\sim\GOE(d_\nu)$, we have that $\iprod{M,K_{\nu,n}} = \iprod{G,K_{\nu,n}}$ is distributed as a Gaussian with variance $\frac{2}{d_\nu}\norm{K_{\nu,n}}^2_F \le (2n 2^{j_\nu} d^2_\nu / d^2)\alpha^2$. So we can bound \eqref{blockeq:lux2} by
    \begin{equation}
        \E[g\sim\cN(0,8\eps^2_\nu n\alpha^2 / d)]{\exp(g)} = e^{8\eps^2_\nu n 2^{j_\nu} \alpha^2 / d^2} 
        \ll e^{\sqrt{d_\nu}}
    \end{equation}
    where the last step follows by \eqref{eq:nubd}.
\end{proof}

\begin{proof}[Proof of Proposition~\ref{blockprop:lstar-reduction}]
    Define the event
    \[
        S = \lt\{
            \sup_{1\le t\le n} \norm{K_{\nu,t}}_F \le \left(n^{1/2}\cdot 2^{j_\nu/2}d^{3/2}_\nu/d\right)\alpha
            ~\text{and}~
            |\Psi_{\nu,n} - 1| \le n\cdot \frac{\eps^2_\nu 2^{j_\nu}}{d^{1/2}_\nu d^2}\beta \ \ \forall \ \nu\in[m]
        \rt\}.
    \]
    By Claims~\ref{blockcl:frob-sup} and \ref{blockcl:lx-bd}, $\Pr[p_0]{S}=1-o(1)$.
    We will show that if $S$ holds, then $\tau=\infty$. 
    Indeed, if $\tau = t<\infty$, then there exists $\nu$ such that either $\norm{K_{\nu,t}}_F > (n^{1/2}\cdot 2^{j_\nu/2}d^{3/2}_\nu/d)\alpha$ or $|\Phi_{\nu,t}-1|>n\cdot \frac{\eps^2_\nu 2^{j_\nu}}{d^{1/2}_\nu d^2}\beta$ holds.
    Since $\Psi_{\nu,n} = \Phi_{\nu,t}$, this contradicts $S$.
    
    So, $\tau=\infty$ on $S$.
    This implies that for all $\nu$, $|L_\nu(\bx)-1| = |\Phi_{\nu,n}-1| \le n\cdot \frac{\eps^2_\nu 2^{j_\nu}}{d^{1/2}_\nu d^2}\beta = o(1/m)$. 
    Moreover $\norm{K_{\nu,n}}_F \le (n^{1/2}\cdot 2^{j_\nu/2}d^{3/2}_\nu/d)\alpha$ for all $\nu$, so by Claim~\ref{blockcl:lxx-bd} we have $L_\nu(\bx,\bx) \ll e^{\sqrt{d}_\nu}$ for all $\nu$.
\end{proof}

\subsection{Bounding \texorpdfstring{$H$}{H} in Frobenius norm by bootstrapping}\label{blocksec:bound-signed-sum}

In this subsection, we prove Lemma~\ref{blocklem:bootstrap}.
Throughout this subsection, fix some $\nu\in[m]$. To ease notation, we will drop subscripts and refer to $K_\nu$ and $H_\nu$ simply as $K$ and $H$. Let $\bz = (z_1,\ldots,z_t)$ be a sequence of unit vectors satisfying $t \le n$ and
\begin{equation}
    \norm{K(\bz)}_F \le \left(n^{1/2}\cdot 2^{j_\nu/2} d^{3/2}_\nu / d\right) \gamma \label{blockeq:assume_K}
\end{equation} for some $\gamma \gg m^{1/2}\xi^{1/2}$. Let $J\subseteq[t]$ denote the set of $s\in[t]$ for which $\nu(z_s) = \nu$. Suppose that
\begin{equation}
    |J| \le n\cdot (d_\nu/2^{j_\nu})\cdot m\xi \label{eq:Jbound}
\end{equation}
as in Lemma~\ref{lem:bucket_freqs}.

The following lemma bounds a variant of $K(\bz)$ where we multiply each summand by an adversarial $b_i\in [-1,1]$.
This will be used to control the discrepancy $H(\bz)-K(\bz)$ in the bootstrapping argument.
\begin{lemma}
    \label{blocklem:uniform-frob-bd}
    Uniformly over $b_1,\ldots,b_t\in [-1,1]$, we have
    \[
        \norm{\sum_{i\in J} b_i \fr{d_\nu z_iz_i^\dagger - I_{d_\nu}}{z_i^\dagger A_\nu z_i}}_F
        \lesssim nd^{1/2}_{\nu}\cdot m\xi + (n^{1/2}\cdot 2^{j_\nu/2} d^{3/2}_\nu / d) \gamma
    \]
\end{lemma}
\begin{proof}
    For any choice of $b_1,\ldots,b_t$,
    \begin{align*}
        \norm{\sum_{i\in J} b_i \fr{d_\nu z_iz_i^\dagger - I_{d_\nu}}{z_i^\dagger A_\nu z_i}}_F
        &\le 
        \norm{\sum_{i\in J} b_i \fr{d_\nu z_iz_i^\dagger}{z_i^\dagger A_\nu z_i}}_F +
        \norm{\sum_{i\in J} b_i \fr{I_{d_\nu}}{z_i^\dagger A z_i}}_F \\
        &\le 
        \norm{\sum_{i\in J}  \fr{d_\nu z_iz_i^\dagger}{z_i^\dagger A_\nu z_i}}_F +
        \norm{\sum_{i\in J}  \fr{I_{d_\nu}}{z_i^\dagger A_\nu z_i}}_F \\
        &\le 
        \norm{K(\bz)}_F + 2 \norm{\sum_{i\in J}  \fr{I_{d_\nu}}{z_i^\dagger A_\nu z_i}}_F.
    \end{align*}
    The second inequality holds because the matrices $d_\nu z_iz_i^\dagger$ and $I_{d_\nu}$ are both psd.
    Using \eqref{blockeq:xAx-lb} and the assume bound on $|J|$ in \eqref{eq:Jbound}, we have 
    \[
        \norm{\sum_{i\in J}  \fr{I_{d_\nu}}{z_i^\dagger A_\nu z_i}}_F
        \le 
        (2^{j_\nu}/d)|J|d^{1/2}_\nu \lesssim nd^{1/2}_{\nu}\cdot m\xi.
        %\lesssim 
        %nd^{1/2}.
    \]
    The result follows by our assumed bound on $\norm{K(\bz)}_F$.
\end{proof}

For $S\subseteq J$, let $\bz_S = (z_i)_{i\in S}$.
Further, let
\[
    H_{S} = \sum_{i\in S} 
    \fr{d_\nu z_iz_i^\dagger - I_{d_\nu}}{z_i^\dagger A_{\nu} z_i} \cdot \fr{L_\nu(\bz_{S\setminus \{i\}})}{L_\nu(\bz_{S})}
    \qquad
    \text{and}
    \qquad
    K_{S} = \sum_{i\in S} 
    \fr{d_\nu z_iz_i^\dagger - I_{d_\nu}}{z_i^\dagger A_\nu z_i}.
\]
The following lemma gives a preliminary bound on $\norm{H_S}_F$.
In the proof of Lemma~\ref{blocklem:bootstrap}, we will use this bound to control $\norm{H_S}_F$ for $|S| = t - O(\log n)$, followed by the bootstrap argument over $O(\log n)$ recursive rounds to contract the bound to $O((2^{j_\nu}/d)^{1/2} d_\nu n^{1/2} \gamma)$.
\begin{lemma}
    \label{blocklem:crude-frob-ub}
    There exists an absolute constant $C$ such that for all $S\subseteq [t]$, $\norm{H_S}_F \le C(nd^{1/2}_{\nu}\cdot m\xi + (n^{1/2}\cdot 2^{j_\nu/2} d^{3/2}_\nu / d) \gamma)$.
\end{lemma}
\begin{proof}
    We take $C$ to be twice the constant hidden by the $\lesssim$ in Lemma~\ref{blocklem:uniform-frob-bd}.
    Note that for any fixed $\oM_\nu \in U_\nu$, for the $U_\nu$ given by Lemma~\ref{blocklem:goe-trunc}, and any unit vector $z\in \S_\nu$,
    \[
        \eps \lt|\fr{z^\dagger \oM_\nu z}{z^\dagger A_\nu z}\rt|
        \le 
        \fr{1}{12} \cdot \fr{3}{1/2} = \fr12,
    \]
    so $1 + \eps \fr{z^\dagger \oM_\nu z}{z^\dagger A_\nu z} \in [1/2, 3/2]$.
    Thus, for all $i$, $L_\nu(\bz_S) / L_\nu(\bz_{S\setminus \{i\}}) \in [1/2, 3/2]$, which implies
    \begin{equation}
        \label{blockeq:ratio-crude}
        \fr{L(\bz_{S\setminus \{i\}})}{L(\bz_S)} \in [2/3,2].
    \end{equation}
    Lemma~\ref{blocklem:uniform-frob-bd} gives 
    \[
        \fr12 \norm{H_S}_F \le \fr12 C(nd^{1/2}_{\nu}\cdot m\xi + (n^{1/2}\cdot 2^{j_\nu/2} d^{3/2}_\nu / d) \gamma).
    \]
    as desired.

    % We will prove the lemma by induction on $|S|$. 
    % When $|S|=0$ the statement is trivial, establishing the base case.
    % For the inductive step, note that for all $i\in S$, equations \eqref{blockeq:L-ratio-recursion} and \eqref{blockeq:xAx-lb} imply
    % \begin{equation}
    %     \label{eq:ratio-dev-bd}
    %     \lt|\fr{L(\bz_{S})}{L(\bz_{S \setminus \{i\}})} - 1\rt|
    %     \le 
    %     \fr{2\eps^2}{d^2} \cdot \norm{\fr{H_{S\setminus i}}{z_i^\dagger A z_i}}_{\op}
    %     \le 
    %     \fr{4\eps^2}{d^2} \norm{H_{S\setminus i}}_F.
    % \end{equation}
    % By the inductive hypothesis and the fact that $n \ll d^{3/2}/\eps^2$, this upper bound is $o(1)$. 
\end{proof}

\begin{proof}[Proof of Lemma~\ref{blocklem:bootstrap}]
    Let $\eps^* \triangleq \left(n^{1/2}\cdot 2^{j_\nu/2} d^{3/2}_\nu / d\right) \gamma$ and $\eps' \triangleq Cnd^{1/2}_\nu\cdot m\xi$. If $\eps^* \ge \eps'$, then we are already done by Lemma~\ref{blocklem:crude-frob-ub}. Otherwise, suppose $\eps^* < \eps'$.
    % , this implies that 
    % \begin{equation}
    %     2^{\nu}d_\nu/d \le \frac{C^2nd}{d_\nu \gamma^2 m^2\xi^2}.
    % \end{equation} 
    and let $D = \log(\eps' / \eps^*)$.
    If $t < D$, then by equations \eqref{blockeq:xAx-lb} and \eqref{blockeq:ratio-crude},
    \[
        \norm{H(\bz)}_F
        \le 
        \sum_{i=1}^{t}
        \fr{\norm{d_\nu z_iz_i^\dagger - I_{d_\nu}}_F}{z_i^\dagger A z_i} \cdot \fr{L(\bz_{\sim i})}{L(\bz)}
        \le 2^{j_\nu+1}Dd_\nu / d.
    \]
    But note that $2^{j_\nu+1}Dd_\nu / d \ll \eps^*$ provided that $n \ggsim 2^{j_\nu}/(d_\nu \gamma^2)$. By \eqref{blockeq:eps_assume}, $2^{j_\nu}/d_\nu \le \log(d/\eps)/(2\eps)$, so this holds by \eqref{eq:nlbd} and our choice of $\gamma \gg m^{1/2}\xi^{1/2} \ge \polylog(d/\eps)$. So $\norm{H(\bz)}_F \ll \eps^*$ when $t < D$.
    
    Now suppose $t\ge D$.
    By Lemma~\ref{blocklem:crude-frob-ub} and the assumption that $\eps^* \le \eps'$, $\norm{H_S}_F \le 2\eps'$.
    We will prove by induction on $a\ge 0$ that if $S\subseteq J$ satisfies $|S\backslash J|= D-a$, then
    \[
        \norm{H_S}_F 
        \le 
        \xi_a
        \triangleq
        2\eps^* + e^{-a}\cdot 2\eps'.
    \]
    The base case $a=0$ clearly holds.
    For the inductive step, assume $a\ge 1$. 
    By the inductive hypothesis and equations \eqref{blockeq:L-ratio-recursion} and \eqref{blockeq:xAx-lb}, for all $i\in S$
    \[
        \lt|\fr{L_\nu(\bz_{S})}{L_\nu(\bz_{S \setminus \{i\}})} - 1\rt|
        \le 
        \fr{2\eps^2_\nu}{d^2} \cdot \norm{\fr{H_{S\setminus i}}{z_i^\dagger A_\nu z_i}}_{\op}
        \le 
        \fr{(2^{j_\nu}/d)\cdot 2\eps^2_\nu}{d^2} \norm{H_{S\setminus i}}_F
        \le 
        \fr{(2^{j_\nu}/d)\cdot 2\eps^2_\nu}{d^2}\xi_{a-1}.
    \]
    Since this upper bound is $o(1)$ by \eqref{eq:nubd} and the second part of \eqref{blockeq:eps_assume}, we also have 
    \[
        \lt|\fr{L(\bz_{S\setminus \{i\}})}{L(\bz_{S})}-1\rt|
        \le 
        \fr{(2^{j_\nu}/d)\cdot 3\eps^2_\nu}{d^2}\xi_{a-1}.
    \]
    Write $\fr{L_\nu(\bz_{S\setminus \{i\}})}{L_\nu(\bz_{S})}-1 = \fr{(2^{j_\nu}/d)\cdot 3\eps^2_\nu}{d^2}\xi_{a-1}b_i$ for $b_i\in [-1,1]$.
    By Lemma~\ref{blocklem:uniform-frob-bd}, there is a constant $c$ such that
    \begin{align*}
        \norm{
            \sum_{i\in S} 
            \fr{d_\nu z_iz_i^\dagger - I_{d_\nu}}{z_i^\dagger A_\nu z_i}
            \cdot
            \lt(\fr{L_\nu(\bz_{S\setminus \{i\}})}{L_\nu(\bz_{S})}-1\rt)
        }_F
        &=
        \fr{(2^{j_\nu}/d)\cdot 3\eps^2_\nu}{d^2}\xi_{a-1}
        \norm{
            \sum_{i\in S} 
            \fr{d_\nu z_iz_i^\dagger - I_{d_\nu}}{z_i^\dagger A_\nu  z_i}
            \cdot
            b_i
        }_F \\
        &\le 
        \fr{(2^{j_\nu}/d)\cdot 3\eps^2_\nu nd^{1/2}_\nu\cdot m\xi}{d^2}\xi_{a-1}
        % \fr{5c\eps^2n}{d^{3/2}}\xi_{a-1}
        \le 
        e^{-1}\xi_{a-1},
    \end{align*}
    where in the last step we used that $n \ll \frac{d^2 d^{1/2}_\nu}{\eps^2_\nu 2^{j_\nu} \cdot m\xi}$.
    By the triangle inequality, equation~\eqref{blockeq:xAx-lb}, and our choice of $D$,
    \[
        \norm{K_S}_F
        \le 
        \norm{K(\bz)}_F
        + 
        \sum_{i\in J\setminus S} 
        \fr{\norm{d_\nu z_iz_i^\dagger - I_{d_\nu}}_F}{z_i^\dagger A_\nu z_i}
        \le 
        \eps^* + 2^{j_\nu} D d_\nu/d
        \le 
        \fr{101}{100}\eps^*.
    \]
    Hence
    \begin{align*}
        \norm{H_S}_F
        &\le 
        \norm{K_S}_F
        +
        \norm{
            \sum_{i\in S} 
            \fr{dz_iz_i^\dagger - I_d}{z_i^\dagger A z_i}
            \cdot
            \lt(\fr{L(\bz_{S\setminus \{i\}})}{L(\bz_{S})}-1\rt)
        }_F \\
        &\le 
        \fr{101}{100}\eps^*
        +
        e^{-1}\xi_{a-1}
        \le \xi_a,
    \end{align*}
    as $\fr{101}{100} + 2e^{-1} \le 2$.
    This completes the induction.
    Finally, 
    \begin{align}
        \norm{H(\bz)}_F
        &=
        \norm{H_J}_F
        \le 
        2\eps^* + e^{-D} 2\eps'
        =
        4\eps^*. \qedhere
    \end{align}
\end{proof}

\subsection{Uniform Frobenius bound on the \texorpdfstring{$K_\nu(\bx_{\le t})$}{Kxt} matrix martingale}\label{blocksec:doob}

In this subsection, we will prove Lemma~\ref{blocklem:doob}. 
Fix any $\nu\in[m]$, let $\bx \sim p_0$, recall that $K_{\nu,t} = K_\nu(\bx_{\le t})$. To ease notation, we will drop the subscript $\nu$ and refer to this as $K_t$. Also define $X = \sup_{1\le t\le n}\norm{K_t}_F$.

\begin{lemma}
    \label{blocklem:doob-aux1}
    We have that $\E{X^2}\le 4\E{\norm{K_{\nu,n}}_F^2}$
\end{lemma}
\begin{proof}
    Analogous to Lemma~\ref{lem:doob-aux1}.
\end{proof}

\begin{lemma}
    \label{blocklem:doob-aux2}
    We have that $\E*{\norm{K_n}_F^2} \lesssim 2^{j_\nu} d_\nu^2 n / d$. %nd^2$.
\end{lemma}
\begin{proof}
    We can expand
    \begin{multline}
        \label{blockeq:kn-fnorm}
        \E{\norm{K_n}_F^2}
        = 
        \sum_{i=1}^n
        \E*{\bone{\nu(x_i) = \nu} \cdot \norm{
            \fr{d_\nu x_ix_i^\dagger - I_{d_\nu}}{x_i^\dagger A_\nu x_i}
        }_F^2}
        \\ 
        +
        2
        \sum_{1\le i<j\le n}
        \E*{\lt\la 
            \bone{\nu(x_i) = \nu, \nu(x_j) = \nu}\cdot  
            \fr{d_\nu x_ix_i^\dagger - I_{d_\nu}}{x_i^\dagger A_\nu x_i}, 
            \fr{d_\nu x_jx_j^\dagger - I_{d_\nu}}{x_j^\dagger A_\nu x_j}
        \rt\ra}.
    \end{multline}
    By \eqref{eq:resolution},
    \[
        \E*{\bone{\nu(x_j) = \nu}\cdot \fr{d_\nu x_jx_j^\dagger - I_{d_\nu}}{x_j^\dagger A_\nu x_j} | \cF_{j-1}}
        = 
        \sum_{x_j: \nu(x_j) = \nu} \omega_{x_j} (d_\nu x_jx_j^\dagger - I_{d_\nu})
        = 0,
    \]
    so for any $i < j$ we have 
    \begin{multline}
        \E*{\bone{\nu(x_i) = \nu, \nu(x_j) = \nu}\cdot \lt\la 
             \fr{d_\nu x_ix_i^\dagger - I_{d_\nu}}{x_i^\dagger A_\nu x_i}, 
            \fr{d_\nu x_jx_j^\dagger - I_{d_\nu}}{x_j^\dagger A_\nu x_j}
        \rt\ra}
        \\ = 
        \E*{
            \bone{\nu(x_i) = \nu}\lt\la 
                \fr{d_\nu x_ix_i^\dagger - I_{d_\nu}}{x_i^\dagger A_\nu x_i},
                \E*{\bone{\nu(x_j) = \nu}\cdot \fr{d_\nu x_jx_j^\dagger - I_{d_\nu}}{x_j^\dagger A_\nu x_j} | \cF_{j-1}}
            \rt\ra
        }
        =0.
    \end{multline}
    The other expectation in \eqref{blockeq:kn-fnorm} can be bounded by (recalling \eqref{blockeq:xAx-lb})
    \begin{align}
        \E*{\bone{\nu(x_i) = \nu}\cdot \norm{
             \fr{d_\nu x_ix_i^\dagger - I_{d_\nu}}{x_i^\dagger A_\nu  x_i}
        }_F^2} 
        &\le 
        (2^{j_\nu}/d)\E*{
           \bone{\nu(x_i) = \nu}\cdot  \fr{\la d_\nu x_ix_i^\dagger - I_{d_\nu}, d_\nu x_ix_i^\dagger - I_{d_\nu}\ra}{x_i^\dagger A_\nu x_i}
        } \\
        &=
        (2^{j_\nu}/d)\cdot d_\nu(d_\nu-1)\E*{
            \fr{\bone{\nu(x_i) = \nu}}{x_i^\dagger A_\nu x_i}
        } \\
        &\le(2^{j_\nu}/d)\cdot d_\nu(d_\nu-1)\cdot \sum_{x: \nu(x) = x} \omega_x \le 2^{j_\nu} d^3_\nu/d^2.
    \end{align}
    Therefore $\E*{\norm{K_n}_F^2} \le n\cdot 2^{j_\nu} d^3_\nu / d^2$.
\end{proof}

\begin{proof}[Proof of Lemma~\ref{blocklem:doob}]
    Follows immediately from Lemmas~\ref{blocklem:doob-aux1} and \ref{blocklem:doob-aux2}.
\end{proof}

\section{Refined Bounds for State Certification}
\label{app:full_full}

In this section we use the lower bound instance from Appendix~\ref{app:block} to give a refined version of the analysis in Section~\ref{sec:full} and prove Theorem~\ref{thm:inst_opt}. The steps in this section are essentially already present in \cite{chen2021toward} (see Sections 5.1, 5.2.2, and 5.5 therein), but we include them for the sake of completeness.

\subsection{Bucketing and mass removal}

We will use the following bucketing scheme from \cite[Definition 5.2]{chen2021toward}.

For $j\in\mathbb{Z}_{\ge 0}$, let $S_j$ denote the set of indices $i\in[d]$ for which $2^{-j-1} < \sigma_i \le 2^{-j}$, and define $d_j \triangleq |S_j|$. Let $\calJ$ denote the set of $j$ for which $S_j \neq \emptyset$. We will refer to $j\in\calJ$ as \emph{buckets}. Given $i\in[d]$, let $j(i)$ denote the index of the bucket for which $i\in S_j$.

Our bounds are based on the following modification of $\sigma$ given by removing a small fraction of its entries:

\begin{definition}\label{def:bucket2}
    % Without loss of generality, assume that $\sigma_1,\ldots,\sigma_d$ are sorted in increasing order based on $\sigma_i/d^2_{j(i)}$. Let $d' \le d$ denote the largest index for which $\sum^{d'}_{i=1}\sigma_i \le 3\eps$. Let $\Stail \triangleq [d']$, and let $\Slight$ be the set of $i > d'$ for which $\sum_{i'\in S_{j(i)}\backslash \Stail} \sigma_{i'} \le 2\eps/\log(d/\eps)$.
    If $b$ is the largest number for which the $b$ smallest entries of $\sigma$ sum to at most $\eps$, define $\Slight^1\subseteq[d]$ to be the indices of these $b$ smallest entries. Let $\Slight^2$ denote the set of $i\in[d]$ for which $\sum_{i'\in S_{j(i)}} \sigma_{i'} \le 2\eps/\log(d/\eps)$. Define $\Slight\triangleq \Slight^1\cup \Slight^2$.
    
    Without loss of generality, assume that all $\sigma_i$ are sorted in increasing order based on $\sigma_i/d^2_{j(i)}$. Recall the constant $a^*$ from Lemma~\ref{blocklem:goe-trunc} and Theorem~\ref{blockthm:main_standard}. Let $d' \le d$ denote the largest index for which $\sum_{i\not\in\Slight: i\le d'} \sigma_i \le C_{a^*}\eps$ for a constant $C_{a^*}$ sufficiently lage depending on $a^*$. Let $\Stail\triangleq \brc{i: i\not\in \Slight, i\le d'}$.
    
    Let $m$ denote the number of buckets $j\in\calJ$ for which $S_j$ and $\Slight$ are disjoint. Let $\Sfew\subseteq[d]$ denote the set of $i$ belonging to a bucket of size less than $a^*$, and let $\Smany\subseteq[d]$ denote the set of $i$ belonging to a bucket of size at least $a^*$. 
    
    Let $\sigma'$ denote the matrix given by zeroing out the entries indexed by $\Stail\cup\Slight$. Let $\sigma''$ denote the matrix by further zeroing out the largest entry of $\sigma'$. Let $\sigma^*$ denote the density matrix $\sigma' / \Tr(\sigma^*)$.
    
    Lastly, define $\calJ^*$ to be the set of $j\in\calJ$ for which $S_j$ has nonempty intersection with $\Smany\backslash \Slight$. % Note that by design, $\calJ^*$ and $\calJ^*$ denote the buckets containing indices of the nonzero diagonal entries of $\sigma'$ and $\sigma^*$ respectively. \TODO{maybe get rid of the $\sigma^*$ and $\calJ^*$ stuff}
\end{definition}

\begin{fact}\label{fact:fewbuckets2}
    We have $m\le O(\log(d/\eps))$, that is, there are at most $O(\log(d/\eps))$ indices $j\in\calJ$ for which $S_j$ and $\Slight$ are disjoint. Furthermore, the total mass of $\sigma$ in $\Slight\cup\Stail$ is $O(\eps)$.
\end{fact}

\begin{proof}
    This is a slight modification of \cite[Fact 5.3]{chen2021toward}. By definition of $\Slight^1$, the $(b+1)$-st smallest entry of $\sigma$ is at least $\eps/d$. There are thus at most $\log_2(d/\eps)$ buckets containing $[d]\backslash\Slight^1$, which concludes the proof of the first part. The second part follows by construction.
\end{proof}

\subsection{Tuning the perturbations}

The goal of this section will be to tune the perturbations $\brc{\eps_j}$ from the lower bound instance in Theorem~\ref{blockthm:main_standard} in order to show the following:

\begin{lemma}\label{lem:25}
    For $0 < \eps < \wt{O}(1/\log\log(d))$, for any mixed state $\sigma\in\mathbb{C}^{d\times d}$, the copy complexity of state certification with respect to $\sigma$ to error $\eps$ is at least $\Omega(1/\epsilon) \vee \wt{\Omega}(\norm{\sigma''}_{2/5}/(\eps^2\polylog(d/\eps)))$.
\end{lemma}

\noindent First we handle a minor corner case. Note that Theorem~\ref{blockthm:main_standard} can only be applied to the buckets of $\sigma$ which are of size at least $a^*$. We now verify that if the Schatten 2/5-norm of $\sigma'$ is dominated by such buckets, then the $\wt{\Omega}(\norm{\sigma''}_{2/5}/\eps^2)$ lower bound follows from \emph{classical} lower bounds.

\begin{lemma}\label{lem:geo}
    If $\sum_{i\in\Sfew\backslash(\Stail\cup\Slight)} \sigma^{2/5}_i \ge \frac{1}{2}\norm{\sigma''}^{2/5}_{2/5}$, then state certification with respect to $\sigma$ using incoherent measurements has copy complexity at least $\Omega(\norm{\sigma''}_{2/5}/\eps^2)$.
\end{lemma}

\noindent For this, we use the following instance-optimal lower bound for classical identity testing:

\begin{theorem}[Theorem 1.1 from \cite{valiant2017automatic}]
    Given a known distribution $p$ and samples from an unknown distribution $q$, any tester that can distinguish between $q = p$ and $\norm{p - q}_1 \ge \eps$ with probability $2/3$ must draw at least $\Omega(1/\eps) \vee \Omega(\norm{p^{-\max}_{-\eps}}_{2/3}/\eps^2)$ samples, where $p^{-\max}_{-\eps}$ denotes the vector given by removing from $p$ the largest element and the smallest elements summing up to at most $\eps$.
\end{theorem}

\noindent Note that this immediately implies a lower bound for state certification by considering only diagonal mixed states:

\begin{corollary}\label{cor:classical_instopt}
    State certification with respect to any known mixed state $\sigma$ to error $\eps$ using incoherent measurements requires at least $\Omega(1/\epsilon) \vee \Omega(\norm{\sigma^{-\max}_{-\eps}}_{2/3}/\eps^2)$ samples, where $\sigma^{-\max}_{-\eps}$ denotes the matrix given by projecting out from $\sigma$ the largest eigenvalue and the smallest eigenvalues summing up to at most $\eps$.
\end{corollary}

\begin{proof}[Proof of Lemma~\ref{lem:geo}]
    This is a slight modification of \cite[Lemma 5.12]{chen2021toward}. The idea is that if the hypothesis of the lemma holds, then the spectrum of $\sigma$ is essentially dominated by eigenvalues in geometric progression, in which case there is no distinction between the 2/5- and 2/3-quasinorms and we can simply apply Corollary~\ref{cor:classical_instopt}.
    
    Formally, Corollary~\ref{cor:classical_instopt} implies a lower bound of $\Omega(\norm{\sigma^{-\max}_{\eps}}_{2/3}/\eps^2)$. We would like to relate $\norm{\sigma^{-\max}_{\eps}}_{2/3}$ to
    \begin{equation}
        \biggl(\sum_{i\in\Sfew\backslash(\Stail\cup\Slight)} \sigma^{2/3}_i\biggr)^{3/2} \ge {a^*}^{-5/2} \cdot (1 - 2^{-2/5})^{5/2} \cdot \biggl(\sum_{i\in\Sfew\backslash\Stail} \sigma^{2/5}_i\biggr)^{5/2} \ge \Omega(\norm{\sigma''}_{2/5}), \label{eq:geoseries}
    \end{equation}
    where the last step follows by the hypothesis of the lemma and Fact~\ref{fact:fewbuckets2}.
    
    Suppose that there is some $i$ for which $d_{j(i)} \le a^*$ and $i$ is not among the indices removed in the definition of $\sigma^{-\max}_{-\eps}$. Then we can lower bound $\norm{\sigma^{-\max}_{-\eps}}_{2/3}$ by $\sigma_i$, which is at least ${a^*}^{-3/2}(1 - 2^{-2/3})^{3/2} = \Omega(1)$ times the left-hand side of \eqref{eq:geoseries}.
    
    On the other hand, suppose that all $i$ for which $d_{j(i)} \le a^*$ are removed in the definition of $\sigma^{-\max}_{-\eps}$. As long as $\sigma^{-\max}_{-\eps}$ has some nonzero entry, call it $\sigma_{i^*}$, then $\sigma_{i^*} \ge \max_{i\in\Sfew\backslash(\Stail\cup\Slight)}\sigma_i$, so we can similarly guarantee that $\norm{\sigma^{-\max}_{-\eps}}_{2/3}\ge \sigma_{i^*}$ is at least ${a^*}^{-3/2}(1 - 2^{-2/3})^{3/2} = \Omega(1)$ times the left-hand side of \eqref{eq:geoseries}. Otherwise, we note that $\sigma''$ is zero as well, in which case we are also done.
\end{proof}

It remains to consider the primary case where the hypothesis of Lemma~\ref{lem:geo} does not hold, which we can express as
\begin{equation}
    \sum_{i\in\Smany\backslash(\Slight\cup\Stail)}\sigma^{2/5}_i > \frac{1}{2}\norm{\sigma''}^{2/5}_{2/5}, \label{eq:maincase}
\end{equation}
and this is the case where we will use Theorem~\ref{blockthm:main_standard}. Because Corollary~\ref{cor:classical_instopt} already shows that the copy complexity is at least $\Omega(1/\epsilon)$, we will assume henceforth that the lower bound in Theorem~\ref{blockthm:main_standard} is at least $\Omega(1/\epsilon)$.

First for every $i\in\Smany\backslash\Slight$, define the perturbations 
\begin{equation}
    \eps_{j(i)} \triangleq d\cdot \brc*{2^{-j(i)-1} \Big/ \left(12 + \Theta\left(\sqrt{\log(m)/d_{j(i)}}\right)\right)} \wedge \brc*{\zeta 2^{-2/3(j(i)+1)}d^{2/3}_{j(i)}} \label{eq:epsdef}
\end{equation}
for normalizing quantity $\zeta$ satisfying
\begin{equation}
    \sum_{j\in\calJ^*} d_{j}\cdot \brc*{2^{-j-1} \wedge \zeta 2^{-2/3(j+1)}d^{2/3}_j } = \eps. \label{eq:normalize}
\end{equation}
Note that this choice $\zeta$ ensures that the trace distance between the two states under $H_0$ and $H_1$ in Theorem~\ref{blockthm:main_standard} is $\Omega(\eps)$.

The rest of the proof is devoted to analyzing what Theorem~\ref{blockthm:main_standard} gives for this choice of $\brc{\eps_j}$. The main step is to upper bound the normalizing quantity $\zeta$.

\begin{lemma}\label{lem:zetabound}
    $\zeta \le O(\eps)\cdot \biggl(\sum_{j\in\calJ^*} 2^{-2j/3} d^{5/3}_j\biggr)^{-1}$.
\end{lemma}

\noindent To prove this, we will need the following elementary fact. 

\begin{fact}\label{fact:sorting}
    Let $u_1 \le \cdots \le u_m$ be numbers for which there are at most $\ell$ elements in any interval $[2^{-j-1},2^{-j}]$. Let $v_1\le \cdots v_n$ and let $d_1,\ldots,d_n > 1$ be arbitrary integers. Let $w_1\le \cdots \le w_{m+n}$ be these numbers $u_1,\ldots,u_m,v_1,\ldots,v_n$ in sorted order. For $i\in[m+n]$, define $d^*_1$ to be $1$ if $w_i$ corresponds to some $u_j$, and $d_j$ if $w_i$ corresponds to some $v_j$.
    
    There is an absolute constant $C_{\ell}$ depending on $\ell$ such that the following holds. Let $s$ be the largest index for which $\sum^s_{i=1}w_i d^*_i \le C_\ell\eps$. Let $a,b$ be the largest indices for which $u_a,v_b$ are present among $w_1,\ldots,w_s$ (if none exists, take it to be 0). Then either $b = n$ or $\sum^{b+1}_{i=1}v_id_i > 2\eps$.
\end{fact}

\begin{proof}
    This is Fact 5.16 from \cite{chen2021toward} with minor modifications. We may assume $s < m + n$ (otherwise obviously $b = n$). Assume to the contrary that $\sum^{b+1}_{i=1} v_id_i \le \eps$. We proceed by casework based on whether $w_{s'+1} = u_{a+1}$ or $w_{s'+1} = v_{b+1}$.
    
    If $w_{s'+1} = u_{a+1}$, then
    \begin{equation}
        C_\ell \eps < \sum^{s+1}_{i=1} w_id^*_i = \sum^{a+1}_{i=1} u_i + \sum^b_{i=1}v_id_i \le \sum^{a+1}_{i=1} v_{b+1}\cdot 2^{\lceil(1 -  i)/\ell\rceil} + \sum^b_{i=1} v_id_i \le O_{\ell}(1)\eps + \sum^b_{i=1}v_id_i,
    \end{equation}
    where in the first step we used maximality of $s$, in the third step we used that $u_{a+1} \le v_{b+1}$ and the assumption on $\brc{u_i}$, and in the last step we used that $v_{b+1} \le \sum^{b+1}_{i=1}v_id_i \le \eps$. From this, if $C_\ell$ is sufficiently large, then we conclude that $\sum^b_{i=1}v_id_i > 2\eps$, a contradiction. The argument for $w_{s'+1} = v_{b+1}$ is analogous.
\end{proof}

\begin{corollary}\label{cor:outoftail}
    If \eqref{eq:maincase} holds, then $\Smany\backslash(\Slight\cup\Stail)$ is nonempty, and there exists an absolute constant $c > 0$ such that for any $i\in \Smany\backslash(\Slight\cup\Stail)$ in some bucket $j$, $\zeta \cdot 2^{-2/3(j+1)}d^{2/3}_j \le 2^{-j-1}$.
\end{corollary}

\begin{proof}
    The first part immediately follows from \eqref{eq:maincase}. For the second part, take some constant $c \ge 1$ to be optimized later and suppose to the contrary that for some $i^*\in\Smany\backslash(\Slight\cup\Stail)$, lying in some bucket $j^*$, we have $2^{-j^* - 1} < \zeta\cdot 2^{-2/3(j^*+1)}d^{2/3}_{j^*}$, or equivalently $2^{-j^*-1}/d^2_{j^*} < \zeta^3$. Because in the definition of $\Stail$, we sorted by $\sigma_i/d^2_{j(i)}$, we then also have that $2^{-j(i) - 1}/d^2_{j(i)} < \zeta^3$ for all $i\in\Stail$, or equivalently, $2^{-j(i)-1} < \zeta\cdot 2^{-2/3(j+1)}d^{2/3}_{j(i)}$.
    
    To induce a contradiction, we lower bound the sum on the left-hand side of \eqref{eq:normalize} by the contribution from $j\in\calJ^*$ for which $S_j$ contains an index $i$ satisfying $i\le i^*$. The above discussion implies that for such $j$, the corresponding summand in \eqref{eq:normalize} is given by $d_j\cdot 2^{-j(i)-1}$. So the left-hand side of \eqref{eq:normalize} is at least
    \begin{equation}
        \sum_{i\in\Smany\backslash\Slight: i\le i^*} \sigma_i > \eps,
    \end{equation}
    where in the latter inequality we used Fact~\ref{fact:sorting} applied to the numbers $\ell \triangleq a^*$, $\brc{u_i}\triangleq \brc{\sigma_i}_{i\in\Sfew\backslash\Slight}$, $\brc{v_i}\triangleq \brc{\sigma_i/d^2_{j(i)}}_{i\in\Smany\backslash\Slight}$, and $\brc{d_i} \triangleq \brc{d^2_{j(i)}}_{i\in\Smany\backslash\Slight}$, in light of our definition for $\Stail$. This contradicts \eqref{eq:normalize}.
\end{proof}

We are finally ready to upper bound the normalizing constant $\zeta$.

\begin{proof}[Proof of Lemma~\ref{lem:zetabound}]
    By Corollary~\ref{cor:outoftail} and \eqref{eq:normalize},
    \begin{equation}
        \eps\ge \Omega(\zeta)\cdot\sum_{j\in\calJ^*} d_{j}\cdot 2^{-2/3(j+1)}d^{2/3}_j \ge \Omega(\zeta)\sum_{j\in\calJ^*} 2^{-2j/3} d^{5/3}_j.
    \end{equation} The claimed bound follows.
\end{proof}

We are now ready to complete the proof of Lemma~\ref{lem:25}.

\begin{proof}[Proof of Lemma~\ref{lem:25}]
    As discussed above, because of Lemma~\ref{lem:geo} it suffices to consider the case where \eqref{eq:maincase} holds. We will apply Theorem~\ref{blockthm:main_standard} to the principal submatrix of $\sigma$ indexed by the indices from buckets in $\calJ^*$. It suffices to show that the copy complexity in that theorem, when specialized to $\eps_j$ from \eqref{eq:epsdef}, is at least $\wt{\Omega}(\norm{\sigma''}_{2/5}/(\eps^2\polylog(d/\eps)))$. Note that we can apply Theorem~\ref{blockthm:main_standard} to this submatrix because by our definition of $\Slight$, the second part of \eqref{blockeq:eps_assume} holds, by the first argument of each minimum in \eqref{eq:epsdef}, the first part of \eqref{blockeq:eps_assume} holds, and by the definition of $\Smany$, $d_j$ is sufficiently large for every $j$ that appears in this submatrix. Note that our definition of $m$ in Definition~\ref{def:bucket2} is the same as the parameter $m$ in Theorem~\ref{blockthm:main_standard}. Recall from Fact~\ref{fact:fewbuckets2} that $m\le O(\log(d/\eps))$.
    
    \begin{equation}
        \eps_{j} \triangleq d\cdot \brc*{2^{-j-1} \Big/ \left(12 + \Theta\left(\sqrt{\log(m)/d_{j}}\right)\right)} \wedge \brc*{\zeta 2^{-2/3(j+1)}d^{2/3}_{j}}
    \end{equation}
    
    First, let us rewrite the lower bound from that theorem as
    \begin{align}
        \frac{1}{m}\min_{j\in\calJ^*} \frac{d^{1/2}_j d^2}{\eps^2_j 2^{j}} &\ge \frac{1}{m}\left(\sum_{j\in\calJ^*} \frac{\eps^4_j 2^{2j}}{d_j d^4}\right)^{-1/2} \\
        \intertext{Substituting our choice of $\brc{\eps_j}$ from \eqref{eq:epsdef} and denoting $\alpha_j\triangleq 12 + \Theta(\sqrt{\log(m)/d_j})$, we get}
        &\gtrsim \frac{1}{m}\left(\sum_{j\in\calJ^*} \frac{2^{-2j}}{\alpha^4_{j} d_j} \wedge \zeta^4 2^{-2j/3} d^{5/3}_j \right)^{-1/2} \\
        &\ge \frac{1}{m}\left(\sum_{j\in\calJ^*} \alpha^{-1}_j \zeta^{3} 2^{-j} d_j \wedge \zeta^4 2^{-2j/3} d^{5/3}_j \right)^{-1/2} \\
        &\gtrsim \frac{\zeta^{-3/2}}{m}\left(\sum_{j\in\calJ^*} d_j \brc{\alpha^{-1}_j 2^{-j-1} \wedge \zeta 2^{-2/3(j+1)}d^{2/3}_j} \right)^{-1/2} \\
        &\gtrsim \zeta^{-3/2} \cdot \eps^{-1/2}/m \gtrsim (\eps^{-2} / m) \cdot \left(\sum_{j\in \calJ^*} 2^{-2j/3} d^{5/3}_j\right)^{3/2}  \\
        &\ge \max_{j\in\calJ^*, i\in S_j} \sigma_i d^{5/2}_j/(\eps^2\log(d/\eps)) \gtrsim \norm{\sigma''}_{2/5}/(\eps^2\polylog(d/\eps)),
    \end{align}
    where in the second step we used that the minimum of two nonnegative numbers increases if we replace one of them by a weighted geometric mean of the two, in the fourth step we used \eqref{eq:normalize}, in the fifth step we used Lemma~\ref{lem:zetabound}, in the penultimate step we used Fact~\ref{fact:fewbuckets2}, and in the last step we used \eqref{eq:maincase}.
\end{proof}

\subsection{Putting everything together}

\begin{proof}[Proof of Theorem~\ref{thm:inst_opt}]
    The proof will be given by modifying a few places in the proof in Section~\ref{sec:full}. We proceed by the same casework of whether or not $d_j = 1$ for all $j\in\calJ^*$ (note that our definition of $\calJ^*$ is slightly different from the one used in Section~\ref{sec:full}). 
    
    First by Fact~\ref{fact:fewbuckets2} we have that $\Tr(\sigma') \ge 1 - O(\eps) \ge \Omega(1)$, so by Fact~\ref{fact:fidelity} it suffices to lower bound the copy complexity by
    \begin{equation}
        \Omega\left(d_{\mathsf{eff}}\norm{\sigma'}_{1/2}/(\eps^2\log^{\Theta(1)}(d/\eps))\right).
    \end{equation}
    \noindent\textbf{Case 1.} $d_j = 1$ for all $j\in\calJ^*$. Note that in this case,
    \begin{equation}
        \norm{\sigma'}^{1/2}_{1/2} = \sum_{j\in\calJ^*} 2^{-j/2} = O(1)
    \end{equation}
    and $\norm{\sigma^*}_{1/2} = \Theta(\norm{\sigma'}_{1/2})$. As $d_{\mathsf{eff}} = 1$, it thus suffices to show a lower bound of $\Omega(1/\eps^2)$ in this case.
    
    If additionally we have $|\calJ^*| = 1$, then for $\eps$ at most a sufficiently small constant, the maximum entry of $\sigma$ is at least $3/4$, so we can apply Lemma~\ref{lem:corner} to obtain a lower bound of $\Omega(1/\eps^2)$ as desired.
    
    Otherwise, let $j,j'$ be the two smallest bucket indices in $\calJ^*$, and let $i,i'$ be the elements of the singleton sets $S_j, S_{j'}$. If $\eps \le c2^{-j/2-j'/2-1}$ for sufficiently small absolute constant $c > 0$, we can invoke \cite[Lemma A.4]{chen2021toward} to conclude a lower bound of $\Omega(1/\eps^2)$.
    
    Otherwise, suppose $\eps > c2^{-j/2-j'/2-1}$. Because $2^{-j} > 2^{-j'}$, we know that $2^{-j} \le O(\eps)$. In particular, consider the state $\sigma^{**}$ given by zeroing out $\sigma_{i'}$ from $\sigma'$ and normalizing. For this matrix, $d_{\mathsf{eff}} = 1$ and $\norm{\sigma^{**}}_{1/2} = O(1)$. Furthermore, because $\eps$ is smaller than some absolute constant, we conclude that the nonzero entry of $\sigma^{**}$ is at least $3/4$, so we can again apply Lemma~\ref{lem:corner} to conclude a lower bound $\Omega(1/\eps^2)$.
    
    \noindent\textbf{Case 2.} $d_j > 1$ for some $j\in\calJ^*$. In this case, let $j_1\triangleq \arg\max_{j\in\calJ^*}d_j$ and $j_2\triangleq \arg\max_{j\in\calJ^*} d^2_j 2^{-j}$. Note that $d2^{-j_1} \ge d_{j_1}2^{-j_1} \gtrsim \epsilon/\log(d/\epsilon)$, so 
    \begin{equation}
        j_1 \le O(\log(d/\epsilon)). \label{eq:j1bound}
    \end{equation}
    If $\eps \le cd_{j_2} 2^{-j_1/2 - j_2/2-1} / j_1$ for sufficiently small constant $c > 0$, then we can invoke the lower bound instance from Section~\ref{sec:offdiag}. The proof in this case is identical to the corresponding part of the proof in Section~\ref{sec:full}.
    
    It remains to consider the case of 
    \begin{equation}
        \eps > cd_{j_2}2^{-j_1/2-j_2/2-1} / j_1. \label{eq:eps_annoying}
    \end{equation}
    We would like to use the lower bound from Lemma~\ref{lem:25}. We would first like to relate $\norm{\sigma'}_{2/5}$ to $\norm{\sigma''}_{2/5}$ (recall that the difference is that $\sigma''$ is defined by removing the largest entry from $\sigma'$.
    \begin{lemma}\label{lem:casework}
        Either $\norm{\sigma''}_{2/5}\ge\Omega(\norm{\sigma'}_{2/5})$, or the following holds. Let $j^{\circ}$ be the index maximizing $d^{5/2}_j 2^{-j}$. Then 1) $j^{\circ} = \min_{j\in\calJ^*} j$, 2) $d_{j^\circ} = 1$, and 3) $j^{\circ} = 0$.
    \end{lemma}
    
    \begin{proof}
        This is essentially Lemma 5.26 from \cite{chen2021toward}. We will assume that $\norm{\sigma''}_{2/5} = o(\norm{\sigma'}_{2/5})$ and show that 1)-3) must hold. Let $i_{\max}$ be the index of the top entry of $\sigma'$. Suppose 1) does not hold. Then
        \begin{equation}
            \frac{\norm{\sigma'}^{2/5}_{2/5}}{\norm{\sigma''}^{2/5}_{2/5}} \le \frac{\sigma^{2/5}_{i_{\max}} + \sum_{i\in S_{j^{\circ}} \sigma^{2/5}_i}}{\sum_{i\in S_{j^{\circ}} \sigma^{2/5}_i}} \le 2,
        \end{equation}
        where the first inequality follows by the elementary fact that for $a\ge b\ge 0$ and $c\ge 0$, $\frac{a+c}{b+c} \le \frac{a}{b}$, and the second inequality follows by the definition of $j^{\circ}$. This contradicts the assumption that $\norm{\sigma''}_{2/5} = o(\norm{\sigma'}_{2/5})$.
        
        Next, suppose 1) holds but 2) does not. Then
        \begin{equation}
            \frac{\norm{\sigma'}^{2/5}_{2/5}}{\norm{\sigma''}^{2/5}_{2/5}} \le \frac{\sum_{i\in S_{j^\circ}}\sigma^{2/5}_i }{\sum_{i\in S_{j^\circ}\backslash\brc{i_{\max}}}\sigma_i^{2/5}} \le O(1),
        \end{equation}
        where in the first step we again used the above elementary fact and in the second step we used that 2) does not hold. We again get a contradiction.
        
        Finally, suppose 1) and 2) hold, but 3) does not. Because 1) holds and $j^{\circ} > 0$, this implies that $\norm{\sigma'}_{\op} \le 1/2$. On the other hand, $\norm{\sigma''}_{2/5} \ge \norm{\sigma''}_1 \ge (1 - O(\eps)) - 1/2 \ge 1/2 - O(\eps)$. So for $\eps$ smaller than a sufficiently large constant, we get that $\norm{\sigma''}_{2/5} \ge \Omega(\norm{\sigma'}_{\op})$, so $\norm{\sigma''}_{2/5} \ge \Omega(\norm{\sigma'}_{2/5})$, a contradiction.
    \end{proof}
    
    Suppose the latter scenario in Lemma~\ref{lem:casework} happens, but the former does not. In this case, because $d_{j^\circ} = 1$, we also have that $j^\circ = j_2$, so $1 \ge d_{j^{\circ}} 2^{-j^\circ} = d^2_{j^\circ} 2^{-j^\circ} = d^2_{j_c} 2^{-j_2}$. Note that this implies that $\norm{\sigma'}_{1/2} \le \log(d/\eps)$. Furthermore, it implies that
    \begin{equation}
        1 \ge d^2_{j_2} 2^{-j_2} \ge d^2_{j_1} 2^{-j_1} \ge \Omega(d^{3/2}_{j_1}\eps/\log(d/\eps)),
    \end{equation}
    where in the second step we used that $j_2 \arg\max_{j\in\calJ^*} d^2_j 2^{-j}$, and in the last step we used that $d_j2^{-j} \ge \Omega(\eps/\log(d/\eps))$ by definition of $\Slight$ and $\calJ^*$. We conclude that 
    \begin{equation}
        \eps \le O(d^{-3/2}_{j_1}\log(d/\eps)). \label{eq:eps1}
    \end{equation} 
    But recall that we are assuming that \eqref{eq:eps_annoying} holds, i.e.
    \begin{equation}
        \eps \gtrsim d_{j_2} \cdot 2^{-j_1/2-j_2/2} / j_1 = 2^{-j_1/2-j_2/2} / j_1 \ge \Omega(\eps/(d_{j_1}\log(d/\eps)))^{1/2} / j_1, \label{eq:eps2}
    \end{equation}
    where the second step is by $d_{j_2} = d_{j^{\circ}} = 1$ and the last step is by 3) in Lemma~\ref{lem:casework} and the fact that $d_j2^{-j} \ge \Omega(\eps/\log(d/\eps))$ for all $j\in\calJ^*$. Combining \eqref{eq:eps1} and \eqref{eq:eps2}, we conclude that $d_{j_1} \le \polylog(d/\eps)\cdot j_1 \le \polylog(d/\eps)$, where in the last step we used \eqref{eq:j1bound}. But if $d_{j_1} \le \polylog(d/\eps)$, then $d_{\mathsf{eff}} \le \polylog(d/\eps)$.  Then because we also have $\norm{\sigma'}_{1/2}\le O(\log(d/\eps))$, the claimed lower bound in the theorem would follow from a lower bound of $\Omega(1/\eps^2)$. This then follows in a similar fashion to the analysis from Case 1 above.
    
    Finally, suppose instead that the former scenario in Lemma~\ref{lem:casework} happens, in which case Lemma~\ref{lem:25} gives a lower bound of $\Omega(\norm{\sigma'}_{2/5}/(\eps^2\log(d/\eps)))$. Let $j^\circ$  be as defined in Lemma~\ref{lem:casework}. As $\norm{\sigma'}_{2/5} \ge d^{5/2}_{j^\circ} 2^{-j^\circ}$, to complete the proof, it suffices to show that \begin{equation}
        d^{5/2}_{j^\circ} 2^{-j^\circ} \polylog(d/\epsilon) \ge \Omega\left(\sqrt{d_{j_1}}d^2_{j_2} 2^{-j_2}\right). \label{eq:supposecontrary}
    \end{equation}
    Suppose to the contrary. Then because $d^{5/2}_{j_1}2^{-j_1} \le d^{5/2}_{j^\circ} 2^{-j^\circ}$, we would get from the negation of \eqref{eq:supposecontrary} that
    \begin{equation}
        d^2_{j_1}2^{-j_1} \polylog(d/\epsilon) = o(d^2_{j_2}2^{-j_2}). \label{eq:do}
    \end{equation} 
    But by \eqref{eq:eps_annoying} and \eqref{eq:j1bound},
    \begin{equation}
        cd_{j_2}2^{-j_1/2-j_2/2-1} / O(\log(d/\epsilon)) \le \eps \le O(d_{j_1}2^{-j_1}\log(d/\eps)),
    \end{equation}
    where in the last step we used that $d_{j_1}2^{-j_1} \ge \Omega(\eps\log(d/\eps))$ by definition of $\Slight$. Squaring and rearranging, we find that $d^2_{j_2}2^{-j_2} \le O(d^2_{j_1}2^{-j_1} \log^2(d/\epsilon))$, contradicting \eqref{eq:do}.
\end{proof}

% Let $\sigma^*$ denote the matrix given by zeroing out only the entries indexed by $\Stail \cup \Slight$. Finally, let $\sigma^{**}$ denote the matrix given by further zeroing out from $\sigma^*$ as many of the smallest entries as possible without removing more than $2\eps$ mass.

% Let $i_{\max}$ denote the index of the largest entry of $\sigma$. 

    % Lastly, define $\calJ^*$ (resp. $\calJ^*'$ and $\calJ^*$) to be the set of $j\in\calJ$ for which $S_j$ has nonempty intersection with $(([d]\backslash\brc{i_{\max}})\cap \Smany)\backslash (\Stail\cup\Slight)$ (resp. $\Smany\backslash \Slight$ and $[d]\backslash(\Stail\cup\Slight)$). Note that by design, $\calJ^*$ and $\calJ^*$ denote the buckets containing indices of the nonzero diagonal entries of $\sigma'$ and $\sigma^*$ respectively. \TODO{maybe get rid of the $\sigma^*$ and $\calJ^*$ stuff}

\section{Regularity Bounds on GOE and Ginibre Ensembles}
\label{sec:basic-comps}

% \begin{definition}
% For $k$ even, let $\PMat_2(k)$ denote the set of perfect matchings of $\{1,\ldots,k\}$.  If $k$ is odd then $\PMat_2(k)$ is just empty.
% \end{definition}
% \bh{I wrote this in a more ``generic" way, see Lemma~\ref{lem:gaussian-moments}. Perhaps ref that? It might make sense to move that lemma to here.}
% \begin{lemma}\label{lem:gaussian-matchings}
%     Let $x_1,\dots , x_k \in \bR^{d_1}, y_1,\dots , y_k \in \bR^{d_2}$.
%     If $k$ is odd, then
%     \[
%         \bE_{M \sim N(0,1)^{d_1 \times d_2}} \prod_{i=1}^k (x_i^\dagger M y_i)
%         = 
%         0.
%     \]
%     If $k$ is even, then
%     \[
%         \bE_{M \sim N(0,1)^{d_1 \times d_2}} \left[ \prod_{i=1}^k (x_i^\dagger M y_i) \right]
%         =
%         \sum_{\{\{a_1,b_1\},\ldots,\{a_{k/2},b_{k/2}\}\} \in \PMat_2(k)}
%         \prod_{i=1}^{k/2} \bE_{M \sim N(0,1)^{d_1 \times d_2}}\left[(x_{a_i}^\dagger M y_{a_i} \cdot x_{b_i}^\dagger M y_{b_i}) \right] \,.
%     \]
% \end{lemma}

Here we provide the proofs of Lemmas~\ref{lem:goe-trunc} and \ref{lem:gin-trunc}, restated for convenience.

\goetrunc*

\begin{proof}[Proof of Lemma~\ref{lem:goe-trunc}]
    Let $U$ denote the event that $\norm{M}_{\op} \le 3$ and $\norm{M}_F^2 \ge d/4$.
    Let $\lambda_1,\ldots,\lambda_d$ denote the eigenvalues of $U$.
    On the event $U$, we have
    \[
        \norm{M}_F^2 = \sum_{i=1}^d \lambda_d^2 \le \lt(\max_{1\le i\le d} |\lambda_i|\rt) \sum_{i=1}^d |\lambda_i| = \norm{M}_{\op} \norm{M}_1,
    \]
    so $\norm{M}_1 \ge \norm{M}_F^2 / \norm{M}_{\op} \ge d/12$.
    We will show $\Pr{U^c} \le \exp(-\Omega(d))$.
    We generate $M = G - \fr{\Tr(G)}{d}I_d$, where $G\sim \GOE(d)$.
    Note that
    \begin{align*}
        \Pr{\norm{M}_{\op} > 3}
        &\le 
        \Pr{\norm{G}_{\op} > 5/2} + \Pr{|\Tr(G)| > d/2} \\
        &\le 
        \exp(-\Omega(d)) + \exp(-\Omega(d^2)),
    \end{align*}
    where the first term is bounded by \cite[Theorem 6.2]{benarous2001aging} (because $5/2>2$) and the second term is bounded by $\Tr(G) \sim \cN(0,2)$.
    Moreover, since 
    \[
        \norm{M}_F^2 = \norm{G}_F^2 - \Tr(G)^2/d,
    \]
    we have
    \[
        \Pr{\norm{M}_F^2 < d/4}
        \le 
        \Pr{\norm{G}_F^2 < d/2} + \Pr{|\Tr(G)| > d/2}
    \]
    and the second probability is $\exp(-\Omega(d^2))$ as explained above.
    To bound the first probability, write $G_{i,i} = \sqrt{\fr{2}{d}} Z_{i,i}$ and for $i<j$, $G_{i,j}=G_{j,i} = \fr{1}{\sqrt{d}} Z_{i,j}$ for i.i.d. $Z_{i,i},Z_{i,j}\sim \cN(0,1)$.
    Then 
    \[
        \norm{G}_F^2 = \fr{2}{d}
        \sum_{1\le i\le j\le d} Z_{i,j}^2.
    \]
    By a standard Chernoff bound, if $X \sim \chi^2(n)$, then $\Pr{X \le (1-\eps)n} \le ((1-\eps)e^\eps)^{n/2}$.
    Thus 
    \begin{align}
        \Pr*{\norm{G}_F^2 \le d/4}
        &=
        \Pr*{\chi^2(d(d+1)/2) \le d^2/8}
        \le 
        \exp(-\Omega(d^2)). \qedhere
    \end{align}
\end{proof}

\gintrunc*

\begin{proof}[Proof of Lemma~\ref{lem:gin-trunc}]
    Let $U$ be the event $s_{\max}(G) \le 3$ and $\norm{G}_F^2 \ge d_2/2$, where $s_{\max}$ denotes the largest singular value.
    On this event, certainly $\norm{M}_{\op} \le 3$ and $\norm{M}_F^2 = 2\norm{G}_F^2 \ge d_2$.
    Similarly to the proof of Lemma~\ref{lem:goe-trunc}, we have $\norm{M}_1 \ge \norm{M}_F^2 / \norm{M}_{\op} \ge d_2/3$.
    It remains to show $\Pr{U^c} \le \exp(-0.1d_1)$
    By \cite[Corollary 5.35]{vershynin2012matrices}, 
    $\Pr{s_{\max}(G) > 3} \le \exp(-0.11d_1)$.
    Moreover, $\norm{G}_F^2 =_d \fr{1}{d_1} \chi^2(d_1d_2)$, so similarly to the proof of Lemma~\ref{lem:goe-trunc} we have $\Pr{\norm{G}_F^2 < d_2/2} \le \exp(-\Omega(d_1^2))$.
\end{proof}

\section{Separating \texorpdfstring{$K$}{K} and \texorpdfstring{$\kappa$}{kappa}}
\label{app:Kvskappa}

In this short section we construct an example of a transcript $(\bz,\bw)$ for which $K((\bz,\bw)) = 0$, but for which $\kappa((\bz,\bw)) \gg d_1d^2_2/\eps^2$. For simplicity, consider $A = a\cdot \Id_{d_1}$ and $B = b\cdot \Id_{d_2}$.

Consider a unit vector $(z,w)\in\S^{d_1+d_2-1}$ for which $\norm{z}^2 = b/(a+b)$ and $\norm{w}^2 = a/(a+b)$. Now note that if $(\bz,\bw) = ((z,w),(z,-w),(z,w),(z,-w),\ldots)$, then clearly $K((\bz,\bw)) = 0$. On the other hand, if we take $b_i = (-1)^{i+1}$, we find that
\begin{equation}
    \kappa((\bz,\bw)) \ge \norm{\sum^t_{i=1} \frac{zw^{\dagger}\norm{z}\norm{w}}{(a\|z\|^2 + b\|w\|^2)^2}}_F = t\cdot\frac{\norm{z}^2\norm{w}^2}{(a\|z\|^2 + b\|w\|^2)^2} = \frac{t}{4ab}.
\end{equation}
Note that for $\eps \asymp d_2\sqrt{ab}$, $d_1d^2_2/\eps^2 \asymp d_1/(ab)$, so for $t \gg d_1$, $\kappa((\bz,\bw)) \gg d_1d^2_2/\eps^2$.

\end{document}